\documentclass[reqno]{amsart}
\usepackage{amssymb}
\usepackage[dvips]{epsfig}
\usepackage{graphicx}
\usepackage{color}
\usepackage{ulem}
\usepackage{tikz}
\usepackage{amsmath}
\usepackage{amssymb}
\usepackage{amsfonts}
\usepackage{amsthm}
\usepackage{tikz}
\usepackage{bm}

\newcommand{\e}{{\bm e}}
\newcommand{\R}{\mathbb R}

\newcommand{\Z}{\mathbb Z}

\newcommand{\C}{\mathbb C}

\renewcommand{\l}{\lambda}

\newcommand{\T}{\mathbb{T}}

\newtheorem{thm}{Theorem}[section]
\newtheorem{hp}{\bf Hypothesis}

\newtheorem{lem}[thm]{Lemma}
\newtheorem{prop}[thm]{Proposition}
\newtheorem{cor}{\bf Corollary}[section]
\newtheorem{rem}[thm]{\bf Remark}
\theoremstyle{definition}
\newtheorem{defn}[thm]{Definition}

\newtheorem{claim}[thm]{\bf Claim}
\theoremstyle{statement}

\numberwithin{equation}{section}
\usepackage[colorlinks=true,pdfstartview=FitV,linkcolor=magenta,citecolor=cyan]{hyperref}
\usepackage{bm}
\begin{document}
	\title[Localization for monotone quasi-periodic Schr\"odinger operators]{Localization for  Lipschitz  monotone quasi-periodic Schr\"odinger operators on $\Z^d$ via Rellich functions analysis}
	\author[Cao]{Hongyi Cao}
	\address[H. Cao] {School of Mathematical Sciences,
		Peking University,
		Beijing 100871,
		China}
	\email{chyyy@stu.pku.edu.cn}
	\author[Shi]{Yunfeng Shi}
	\address[Y. Shi] {School of Mathematics,
		Sichuan University,
		Chengdu 610064,
		China}
	\email{yunfengshi@scu.edu.cn}
	
	\author[Zhang]{Zhifei Zhang}
	\address[Z. Zhang] {School of Mathematical Sciences,
		Peking University,
		Beijing 100871,
		China}
	\email{zfzhang@math.pku.edu.cn}
	\date{\today}
	\maketitle
	\begin{abstract}
		We establish the Anderson localization and exponential dynamical localization for  a class of quasi-periodic Schr\"odinger operators on $\Z^d$ with bounded or unbounded Lipschitz  monotone potentials via multi-scale analysis based on Rellich function analysis in the perturbative regime. We show that at each scale, the resonant  Rellich function uniformly  inherits the Lipschitz monotonicity property of the potential via a novel Schur complement argument.
	\end{abstract}
	\section{Introduction}
	In this paper, we consider a class of quasi-periodic Schr\"odinger operators on $\Z^d$ \begin{align}\label{model}
		H(\theta)=\varepsilon \Delta+v(\theta+ x\cdot{\omega})\delta_{x, y},\ x,y\in\Z^d,
	\end{align}
	where $\varepsilon\geq0$ and  the discrete Laplacian $\Delta$ is defined as
	\begin{align*}
		\Delta(x, y)=\delta_{{\|x- y\|_{1}, 1}},\ \| x\|_{1}:=\sum_{i=1}^{d}\left|x_{i}\right|.
	\end{align*}
	For the diagonal part of \eqref{model},  we call $\theta$ the phase  and $\omega$ the frequency.  We  let $\theta\in\R$ and $\omega$ be Diophantine  (denoted by $\omega\in 	{\rm DC}_{\tau, \gamma}$ for some $\tau>d,\gamma>0$), where 	\begin{equation}\label{DC}
		{\rm DC}_{\tau, \gamma}=\left\{\omega\in[0,1]^d:\ \|x\cdot\omega\|_\T=\inf_{l\in\mathbb{Z}}|l- x\cdot\omega|\geq \frac{\gamma}{\| x\|_1^{\tau}},\ \forall\  x\in\mathbb{Z}^d\setminus\{o\}\right\}
	\end{equation}
	with $x\cdot\omega=\sum\limits_{i=1}^dx_i \omega_i$  and $o$ being the origin of $\Z^d$. We   assume the  function $v(\theta)$ belongs to the bounded Lipschitz  monotone (BLM) class  or the unbounded  Lipschitz  monotone  (UBLM) class, defined as follows: 
	\begin{itemize}
		\item[({\bf BLM})]   $v$ is a 1-periodic, real-valued function defined on $\R$, continuous on $[0,1)$, satisfying
		for some	constant $L>0$, 
		\begin{equation}\label{LC}
			v(\theta_2)-v(\theta_1)\geq L(\theta_2-\theta_1), \ \  0\leq\theta_1\leq \theta_2<1, 
		\end{equation} 
		and 
		$$-\infty<v(0)<v(1-0)<+\infty,$$ where $``-0"$ denotes the  left limit.   
		\begin{rem}
			Through translation and scaling arguments, we can assume  $v(0)=0, v(1-0)=1$.
		\end{rem} 
		\item[({\bf UBLM})] $v$ is a 1-periodic, real-valued function defined on $\R\setminus\Z$, continuous on $(0,1)$, satisfying for some $L>0$, 
		\begin{equation*}\label{LC'}
			v(\theta_2)-v(\theta_1)\geq L(\theta_2-\theta_1), \ \  0<\theta_1\leq \theta_2<1, 
		\end{equation*} 
		and 
		$$v(0+0)=-\infty, \ \ v(1-0)=+\infty,$$ where $``+0"$ denotes  the right  limit. 
		
	\end{itemize}
	The class BLM (resp. UBLM) was first introduced in  \cite{JK19}  (resp. \cite{Kac19}) as a natural generation of the famous Maryland potential $v(\theta)=\tan\pi\theta$ (resp. $v(\theta)=\theta\mod 1$ originated from \cite{Cra83})  dropping  the meromorphy assumption.  Indeed, the Anderson localization (i.e., pure point spectrum with exponentially decaying eigenfunctions) of quasi-periodic Schr\"odinger operator with  tangent  potential (called the Maryland model)  plays a central role in the study of quantum suppression of classical chaos concerning the quantum kicked rotor (c.f. \cite{FGP82}).  The monotonicity property of Maryland type potentials results in the absence of resonances. As a consequence, the perturbative KAM diagonalization method can be employed to prove the Anderson localization \cite{Cra83, BLS83, Pos83}. More importantly,  it turns out that the exact Maryland model is solvable (c.f. \cite{GFP82})  and all couplings Anderson localization holds true \cite{FP84, Sim85}.  The perturbative  localization results proved in  the 1980s  require Diophantine type frequencies, and {\it the uniform in  phase   Anderson localization}\footnote{Anderson localization holds true for all $\theta$ with the coupling perturbation  $\varepsilon$ being independent of $\theta$.}  via the KAM method  \cite{BLS83} even hinges  on the meromorphy property  of monotone potentials.    Recently, many efforts have been made  to investigate  localization for quasi-periodic Schr\"odinger operators with monotone potentials mainly going  beyond the above restrictions.  For the one-dimensional case,  Jitomirskaya-Liu \cite{JL17} gave an almost complete description of the spectral types for the standard Maryland model  for {\it all parameters}   (c.f. \cite{JY21, HJY22} for more results).   Jitomirskaya-Kachkovskiy \cite{JK19,Kac19}  established  all couplings Anderson localization  for  one-dimensional quasi-periodic Schr\"odinger operators with BLM or UBLM potentials   via the non-perturbative method, which was  first developed for  almost Mathieu operators in the breakthrough works of Jitomirskaya \cite{Jit94,Jit99}.   In \cite{JK24}, they even proved the universality of sharp arithmetic localization for all one-dimensional
	quasi-periodic Schr\"odinger operators with Lipschitz monotone potentials.  For the multi-dimensional case, Shi \cite{Shi23} developed a Nash-Moser iteration type diagonalization  method to prove power-law localization, thus extended some results of Aizenman-Molchanov \cite{AM93}  to the quasi-periodic operators case.  The uniform in phase Anderson localization for a   class of  multi-dimensional  quasi-periodic Schr\"odinger   operators with unbounded monotone potentials was established by Kachkovskiy-Parnovski-Shterenberg  \cite{KPS22}.   Their  method, without the need of multiple KAM-type steps in the perturbative setting, is based on the analysis of   the convergence of perturbation series for  eigenvalues and eigenfunctions, and also applies to  some non-strictly monotone potentials with  a flat segment under some additional assumptions on the length of the segment and the frequencies.  However, the full localization  of \cite{KPS22} relies  on the  {\it unboundedness}  of $v$ and does   not cover the  BLM potentials, such as the one given by  $v(\theta)=\theta \mod 1$.

	In this paper, we prove the uniform in phase Anderson localization and exponential dynamical localization for \eqref{model} with bounded or unbounded Lipschitz  monotone $v$ in the  perturbative regime by  extending  the multi-scale analysis scheme of \cite{FSW90,FV21,CSZ23,CSZ24b} to the monotone potentials one. We construct a sequence of increasing blocks  $B_n\subset\Z^d$ and  associates each $B_n$ with  a unique   Rellich function (parameterized eigenvalue) $E_n(\theta)$, the one   resonant with the previous scale one $E_{n-1}(\theta)$. And then we verify  that $E_n(\theta)$ uniformly  inherits the Lipschitz monotonicity property of $v(\theta)$.   The construction and verification  are  based on a  novel  Schur complement type estimate  concerning the energy parameters. As a consequence, resonant points  of each scale are uniformly separated, which yields   localization for all phases.  Specifically, our main results are: 
	\subsection{Main results}
	\begin{thm}\label{mainb}
		Assume that $\omega\in 	{\rm DC}_{\tau, \gamma}$ and $v$ belongs to the BLM class with some Lipschitz parameter $L>0$.	Then there exists some $\varepsilon_0=\varepsilon_0(L,d,\tau,\gamma)>0$ such that for all $0\leq \varepsilon\leq \varepsilon_0$, we have 
		\begin{enumerate}
			\item  For any $\theta\in \R$,	 $H(\theta)$ satisfies the  Anderson localization.
			\item   For any $\theta\in \R$,	$H(\theta)$ satisfies the  exponential dynamical localization,  that is, there exists some $C(\varepsilon)$ depending on $\varepsilon$, such that for any $x,y\in \Z^d$, 
			\begin{equation}\label{1536}
				\sup_{t\in \mathbb{R}}|\langle e^{itH(\theta)}{\bm e}_x, {\bm e}_y\rangle|\leq C(\varepsilon)e^{-\frac{1}{800}|\ln\varepsilon|\cdot\|x-y\|_1}, 
			\end{equation}
			where $\{\bm e_x\}_{x\in\Z^d}$ is  the standard basis of $\ell^2(\Z^d)$.
		\end{enumerate}
	\end{thm}

	\begin{thm}\label{mainub}
		Assume that $\omega\in 	{\rm DC}_{\tau, \gamma}$ and $v$ belongs to the UBLM class with some Lipschitz parameter $L>0$.	Then there exists some $\varepsilon_0=\varepsilon_0(L,d,\tau,\gamma)>0$ such that for all $0\leq \varepsilon\leq \varepsilon_0$, we have 
		\begin{enumerate}
			\item   For any $\theta\in \R\setminus(\Z+\Z^d\cdot\omega)$,	 $H(\theta)$ satisfies the  Anderson localization.
			\item   For any $\theta\in \R\setminus(\Z+\Z^d\cdot\omega)$,	$H(\theta)$ satisfies the  exponential dynamical localization with the same estimate as in \eqref{1536}.
		\end{enumerate}
	\end{thm}
	\begin{rem}
		 Shortly after the completion of this paper,  a similar result
		was also proved by Kachkovskiy-Parnovski-Shterenberg \cite{KPS24} independently  with a  different approach. 
	\end{rem}
	\begin{rem}
		Our approach  applies to piecewise continuous   Lipschitz  monotone  functions $v$, that is, $v$ is continuous on $(0,1)$ except on a finite set of points $0=\vartheta_1<\vartheta_2<\cdots<\vartheta_K<\vartheta_{K+1}=1$ with  $$v(\vartheta_k-0)<v(\vartheta_k+0), \ \ 2\leq k\leq K,$$
		$$ v(\theta_2)-v(\theta_1)\geq L(\theta_2-\theta_1), \ \ \vartheta_k<\theta_1< \theta_2<\vartheta_{k+1},$$
		and also applies to exponentially decaying  long-range operators, namely, the discrete Laplacian $\Delta$ in \eqref{model}  can be replaced by a Toeplitz operator $S$ satisfying 
		$$|S(x,y)|\leq e^{-\rho\|x-y\|_1}, \ \ \rho>0.$$
	\end{rem}

	\begin{rem}
		For the special case of meromorphic monotone quasi-periodic potentials, we also give an alternative proof of the classical localization result of 	Bellissard-Lima-Scoppola  \cite{BLS83}, in which the uniform in phase Anderson localization was established by taking account of both the phase  and space directions  in  a KAM diagonalization scheme.   
	\end{rem}
	\begin{rem}
		The Anderson localization results in Theorem \ref{mainb} and \ref{mainub} are also uniform in the sense of  Jitomirskaya  \cite{Jit97}  that every normalized  eigenfunction   $\psi_{\theta,s}$ satisfies 	$$|\psi_{\theta,s}(x)|\leq C(\varepsilon) |\psi_{\theta,s}(x_{\theta,s})|e^{-\frac{1}{400}|\ln\varepsilon|\cdot\|x-x_{\theta,s}\|_1}$$
		with   $x_{\theta,s}\in \Z^d$  such that $|\psi_{\theta,s}(x_{\theta,s})|=  \|\psi_{\theta,s}\|_{\ell^\infty(\Z^d)}$ {\rm (c.f. Corollary \ref{802})}. 
	\end{rem}
	
	\begin{rem}
		The above exponential decay rates ``$\frac{1}{400}"$ and   ``$\frac{1}{800}"$  are not optimal. In fact, it can be seen from the proof that the exponential decay rate can be $c(\varepsilon)|\ln\varepsilon|$ with $c(\varepsilon)\to 1$ as $\varepsilon\to 0$.
	\end{rem}
	To our best knowledge, this is the first application of multi-scale analysis method to monotone potentials setting. 
	
	\subsection{Main new ingredients of the proof}
	Our  proof relies on Green's function estimates  via the multi-scale analysis  in the perturbative regime originated from \cite{FS83}. Of particular importance includes the analysis of Rellich functions in the spirit of  \cite{FSW90, FV21,CSZ23,CSZ24b} together with the adaptation  of  the Schur complement argument of Bourgain \cite{Bou00}  (c.f. \cite{CSZ24a} for recent refinement).  Although the general scheme of our proof seems expected, almost all of the known elements are not available and have to be replaced by new ingredients or a more complicated setup.  The details will be explained as follows. 
	
	\subsubsection{Schur complement argument}    In \cite{Bou00}, Bourgain  proved  the  $(\frac{1}{2}-)$-H\"older continuity  of the integrated density  of states (IDS)  for   almost Mathieu operator  via  introducing  the  Schur complement  argument and Weierstrass preparation theorem.  The key point is to analyze the roots of the determinants of  the finite Dirichlet restriction of $H(\theta)$   for a fixed energy $E$.  Resonances at the  $n$-scale induction  there  depend on the pairs of the roots $\{\theta_{n,1}(E),\theta_{n,2}(E)\}$. This celebrated  method was later largely extended by the authors  \cite{CSZ24a} to prove arithmetic Anderson localization for multi-dimensional quasi-periodic Schr\"odinger operators with  cosine potential.  Definitely, all these  estimates of \cite{Bou00, CSZ24a} were obtained by fixing  the energy parameter $E$ and varying the phase  parameters $\theta.$ As we will see below, in the present work, we employ  some type of Schur complement argument via  varying   energy parameters. 
	\subsubsection{Rellich functions analysis}
	In \cite{FV21}, Forman and VandenBoom developed a novel  inductive procedure by constructing the Rellich functions   tree to handle the resonances depending on energy, 
	thus removed  the crucial  symmetrical assumption in \cite{FSW90}.  This scheme was recently extended by the authors \cite{CSZ24b} to  the multi-dimensional case via overcoming the level crossing issue. 
	More precisely,  for a finite subset $\Lambda\subset\Z^d$,   denote by $H_\Lambda(\theta)$ the Dirichlet restriction of $H(\theta)$ to $\Lambda$. It follows from the perturbation theory of  one-parameter self-adjoint operator families that   $H_\Lambda(\theta)$ exhibits $|\Lambda|$ branches of  Rellich functions. The idea is then based on carefully choosing  increasing  scales of finite blocks and   constructing each generation of  the Rellich function, which is resonant with the previous generation.
	{\it The main issue is to verify that the  Rellich  functions maintain the  crucial properties of  $v$}.

	In the previous works \cite{FV21,CSZ24b}, this issue was settled by eigenvector perturbation and eigenvalue variation arguments from \cite{FSW90}. Concretely,  restricting  $H$ to  a fixed block  $B$, one can try to fix an eigenvalue and  its eigenvector and  extend it by a small perturbation to a new eigenvalue and the corresponding  eigenvector  of a bigger block. Then one can use  Feynman-Hellman eigenvalue variation formulas to compute the derivatives of the new  eigenvalue parameterized by $\theta$.  Together  with the Cauchy  interlacing  idea originated from \cite{FV21}, using the above two  arguments   can  successfully yield the two-monotonicity  Morse structure of  Rellich functions for  quasi-periodic Schr\"odinger operators with $C^2$-cosine type  potentials (c.f. \cite{CSZ24b} for details).  
	
	In the monotone  potentials  setting, the lack of resonances leads to obvious conveniences:  One can easily establish the local Lipschitz  monotone property of all branches of   Rellich functions (c.f. \eqref{mno} and \eqref{mne} below), and establish the  Lipschitz continuity of the integrated density of states by the argument of \cite{Kac19} (c.f.  Theorem 4.1 of  \cite{Kac19}). However, there comes the key obstacle for proving the localization,  namely, {\it the  presence of jump discontinuities  of  $v$  at integer points, that will destroy the continuity of the  Rellich function,  making it extraordinary for the   monotonicity  property of the  Rellich function  to hold at these  jump discontinuities}. This becomes a very delicate issue for following reasons:  Since the size of  the Dirichlet block  will tend  to infinity, one may eventually lose the separation and  stability  of  eigenvalues, making the jump discontinuities unclear.  Moreover, the  Rellich function of a given block  can be very close  to other eigenvalues, in which case a rank one perturbation will totally change  the order of eigenvalues, making it hard to compare the size of the  Rellich function near these discontinuities (c.f. Remark \ref{.}),  where the previous eigenvector perturbation and eigenvalue variation ideas do not seem to work. 
	
	To analyze  these  jump discontinuities of the Rellich function  in the present paper, we construct the the Rellich function in a   more quantitative way.   {\it We  deal with  the root of the  characteristic polynomial of a certain scale Dirichlet restriction in a small neighborhood of the value of  the Rellich function from the previous scale, and utilize  Schur complement argument and  Rouch\'e theorem to obtain a unique eigenvalue, which is  chosen to be the value of the  present scale Rellich function (c.f. Proposition \ref{528}).  This  idea is motivated by the work of Bourgain \cite{Bou00}, which made use of   the Schur complement as a perturbation of quadratic polynomial concerning   the phase parameters  $\theta$. The  new aspect   is that we regard the Schur complement as a perturbation of linear function concerning the energy parameters  $E$ (c.f. Lemma \ref{sln})}. Now the issue of jump discontinuities of the Rellich function can be handled by comparing the roots of the two Schur complements with a disparity of a rank one perturbation due to the jump discontinuity  of $v$.   By carefully estimating the difference of the two Schur complements via the multi-scale analysis argument (c.f. Proposition \ref{qiang} and Lemma \ref{1448}), we confirm the sign of this difference and ultimately prove that all the jump discontinuities except those at integer points  of the Rellich function  are non-negative.  Hence each generation of the Rellich function maintains the uniform Lipschitz monotonicity property, which enables the whole induction.
	\subsection{Outline of the inductive scheme}
	The main motivation of  multi-scale analysis is to establish ``good''  Green's function estimates for certain subsets of $\Z^d$, which will be specified by ``good'' later.  Results of this variety are well-established in the literature (c.f. e.g. \cite{FS83,FV21,CSZ24b}).  We sketch this scheme relative to our specific application for the reader's convenience. 
	
	First, we fix a sufficiently small  $\varepsilon_0$ depending on $L,d,\tau,\gamma$ and let $\varepsilon$ be such that  $0\leq \varepsilon\leq \varepsilon_0$.  We choose  $\delta_0:=\varepsilon_0^\frac{1}{20}\gg\varepsilon$ to be the $0$-scale  resonance parameter. For fixed  $\theta$ and $E$, the $0$-resonant points set will be defined as 
	$$S_0(\theta,E):=\left\{p\in \Z^d:\ |v (\theta+p\cdot \omega)-E|<\delta_0\right\}.$$
For a finite set $\Lambda\subset\Z^d$, we define the Green's function   	$$G_\Lambda^{\theta,E}:=\left ( H_\Lambda (\theta)-E\right) ^{-1}.$$
	If a finite set $\Lambda\subset\Z^d$ satisfies  \begin{equation}\label{1418}
		\Lambda\cap S_0(\theta,E)=\emptyset,
	\end{equation} then one can obtain  the ``good'' estimates (i.e.,  one over the resonance bound  of the operator norm and off-diagonal exponential decay of  elements) of  $G_\Lambda^{\theta,E}$ by the Neumann series argument. If \eqref{1418} is not fulfilled,  one can still try to estimate the Green's function via the by-now classical resolvent identity iteration method, which roughly means that if for every $x\in \Lambda$, there exists  a ``good'' subset $U(x)\subset \Lambda$  with $x\in U(x)$ such that $G_{U(x)}^{\theta,E}$ has ``good'' estimates, then one can obtain  ``good'' estimates of   $G_\Lambda^{\theta,E}$, with a small loss on the exponential decay rate, by iterating the resolvent identity. Hence to employ the  resolvent identity method, one needs to associate a ``good''  block $U(x)$ with  every $x\in 	\Lambda\cap S_0(\theta,E)$. For this purpose, we will choose  $l_1:=[|\ln\delta_0|^4]$  ($[\cdot]$ denotes the integer part) to be the first scale length and consider the block  
	$$B_1(x):=\{y\in \Z^d:  \  \|y-x\|_1\leq l_1\}$$ for each $x\in 	\Lambda\cap S_0(\theta,E)$. 
	Since the Dirichlet restrictions of the   original operator satisfy the following translation property 
	\begin{equation}\label{1538}
		H_{B_1(x)}(\theta)=H_{B_1}(\theta+x\cdot\omega),\ \  B_1:=B_1(o),
	\end{equation}
it suffices to consider  $H_{B_1}(\theta)$ for all $\theta$.   
	Since, by self-adjointness, we have 
	\begin{equation}\label{1540}
		\|G_{B_1}^{\theta,E}\|=\frac{1}{\operatorname{dist}(\operatorname{Spec}(H_{B_1}(\theta), E))}, 
	\end{equation}
	 it follows that the issue of  operator norm bound   turns into  the one  of the difference between the eigenvalues of $H_{B_1}(\theta)$ and $E$. Since   $x\in 	\Lambda\cap S_0(\theta,E)$, it follows that   $|v (\theta+x\cdot \omega)-E|<\delta_0$. Thus by \eqref{1538} and \eqref{1540},   it suffices to consider  the  eigenvalue of $H_{B_1}(\theta+x\cdot\omega)$ that  is close  to  $E$ (hence is  also close  to   $v(\theta+x\cdot\omega)$).  Since Diophantine property of $\omega$ and Lipschitz monotonicity property of $v$ provide sufficiently large separation of  resonant points, we can  utilize the Schur complement argument and the  smallness of $\varepsilon$ to construct a unique Rellich function $E_1(\theta)$ of $H_{B_1}(\theta)$ in a certain neighborhood of $v(\theta)$.  Moreover, we can verify that $E_1(\theta)$ also satisfies  the  Lipschitz monotonicity property. Now the resonance depends on the difference $|E_1(\theta+x\cdot\omega)-E|$ and we classify the blocks $B_1(x)$ as ``good''  or ``bad''  depending on the difference  bigger or smaller than the first scale resonance parameter $\delta_1:=e^{-l_1^{2/3}}$. Moreover, by the special structure of $B_1(x)$, one can derive the off-diagonal exponential  decay  of the  Green's function from the sub-exponential operator norm  bound via the resolvent identity argument. Hence we can establish ``good''  estimates of $G_\Lambda^{\theta,E}$ provided all the blocks $B_1(x)\subset\Lambda$ with $x\in 	\Lambda\cap S_0(\theta,E)$ being ``good'', completing   the first inductive step. 
	
	For the general inductive step,  suppose  that for $1\leq k\leq n$, we have inductively constructed increasing block $B_k\subset\Z^d$ centered at origin  and the   corresponding Rellich function $E_k(\theta)$ satisfying the  Lipschitz monotonicity property. Suppose further that we  have  established ``good''   estimates of $G_\Lambda^{\theta,E}$ for  those  $\Lambda$ which do not contain any  $k$-scale ``bad''   block $B_k(x)$ with $|E_k(\theta+x\cdot\omega)-E|$ smaller than the $k$-scale resonance  parameter $\delta_k=e^{-l_k^{2/3}}$. Then,  since resonant points  are uniformly separated by Diophantine property of $\omega$ and Lipschitz monotonicity property of $E_k$,  one can properly choose  a bigger block $B_{n+1}$, and utilize the Schur complement argument together with the smallness of the perturbation term  due to the  off-diagonal exponential decay of Green's function from the  inductive hypothesis to construct a unique Rellich function $E_{n+1}(\theta)$ of $H_{B_{n+1}}(\theta)$ in a certain neighborhood of $E_n(\theta)$. Moreover, one  can   verify that $E_{n+1}(\theta)$ also   satisfies  the Lipschitz monotonicity property. It remains  to  use resolvent identity to  establish ``good''   estimates of $G_\Lambda^{\theta,E}$ for  those  $\Lambda$ which do not contain any  $(n+1)$-scale ``bad''   block $B_{n+1}(x)$ with $|E_{n+1}(\theta+x\cdot\omega)-E|<\delta_{n+1}$.  Since $\varepsilon_0$ is small, the loss of  exponential decay rate in each scale is super-exponentially small, providing a uniformly positive decay rate. 
	
	Finally, the localization results  follow  from the  Green's function estimates  for ``good''  sets, together with the uniform separation of resonant points due to Diophantine property of $\omega$ and Lipschitz monotonicity property of Rellich functions  by  the standard Schnol's  type argument. 
	
	\subsection{Some notation in the paper}\label{notation}
	\begin{itemize}
		\item    	For $L>0$, we denote by $Q_L$ the $L$-size lattice cube centered at the  origin:
		$$Q_L:=\{x\in \Z^d:\ \|x\|_1\leq L\}.$$
		\item For $a\in \C,R>0$, we denote 
		$$D(a,R):=\{z\in \C :\ |z-a|<R\}.$$
		\item  For two sets $X,Y\subset\Z^d$, we  denote $$\operatorname{dist}_1(X,Y)=\min_{\substack{x\in X, y\in Y} }\|x-y\|_1,$$
		$$\partial_Y^-X=\{x\in X:\ \text{there exists }y\in Y\setminus X\text{ with }\|y-x\|_1=1\},$$
		$$\partial_Y^+X=\{y\in Y\setminus X:\ \text{there exists }x\in X\text{ with }\|y-x\|_1=1\},$$
		$$\partial_YX=\{(x,y):\ x\in X,\  y\in Y\setminus X \text{ with }\|y-x\|_1=1 \}.$$
		\item For $\Lambda\subset\Z^d, p\in \Z^d$, we  denote  $$\Lambda+p:=\{x+p:\ x\in \Lambda\}.$$  
		\item We denote by $\|c\|$ the $\ell^2$-norm of a vector $c$ and denote by $\|A\|$ the operator ($\ell^2$) norm of a matrix $A$.
		\item We denote by $A^{\#}$ the adjugate matrix of a matrix $A$.
	\end{itemize}

	\subsection{Organization of the paper}
	Since most of   arguments in the unbounded  potentials case  are similar  to the bounded one  except the analysis of the jump discontinuities of the Rellich function (Proposition \ref{qiang} v.s. Proposition \ref{ubp}), we will mainly focus the multi-scale analysis  on  the  bounded case.   In \S \ref{MSA}, we present  the key arguments of this paper, constructing the Rellich functions
	and establishing Green’s function estimates via  multi-scale  analysis in the bounded case. In \S \ref{ub}, we complete  the analysis of the  unbounded potentials case. In \S \ref{localization},  we employ the results from \S \ref{MSA} and \S \ref{ub} to finish the proof of Theorems  \ref{mainb} and \ref{mainub}.

	\section{Multi-scale analysis}\label{MSA}
	Let $ v$ belong to the BLM class. Let  $\varepsilon_0>0$ be sufficiently small depending on $L,d,\tau ,\gamma$. Let $0\leq \varepsilon\leq \varepsilon_0$.  Define $ \delta_{0}:=\varepsilon_0^{\frac{1}{20}}$. 
	\subsection{The initial scale}
	\subsubsection{Definition of $E_0(\theta)$ and $S_0(\theta,E)$}
	The $0$-generation  Rellich function $E_0(\theta)$ is exactly  $v(\theta)$. Fix $\theta\in \R,E\in \C$. The $0$-resonant points set $S_0(\theta,E)$ is defined as 
	$$S_0(\theta,E):=\left\{p\in \Z^d:\ |v (\theta+p\cdot \omega)-E|<\delta_0\right\}.$$
	\subsubsection{Green's function estimates for $0$-nonresonant sets}
	Fix $\theta,E$. we say that a finite set $\Lambda\subset \Z^d$ is $0$-nonresonant related to  the pair $(\theta,E)$ if $\Lambda\cap S_0(\theta,E)=\emptyset$.
	\begin{prop}\label{0g}
		Let $\gamma_0=\frac{1}{2}|\ln\varepsilon|$.	If $\Lambda\cap S_0(\theta,E)=\emptyset$, then  for any $E^*\in \C$ such that  $|E-E^*|<\delta_0/5$, we have 
		\begin{align*}
			\|G_\Lambda^{\theta,E^*}\|&\leq 10\delta_0^{-1},\\
			|G_\Lambda^{\theta,E^*}(x,y)|&\leq e^{-\gamma_0\|x-y\|_1}, \ \|x-y\|_1\geq 1.
		\end{align*}	
	\end{prop}
	\begin{proof}
		Since $\Lambda\cap S_0(\theta,E)=\emptyset$, it follows that 	for any $x\in \Lambda$ and  $|E-E^*|<\delta_0/5$, we have  
		$$|v (\theta+x\cdot \omega) -E^*|\geq |v (\theta+x\cdot \omega) -E|-|E-E^*|\geq \delta_0-\delta_0/5>\delta_0/2. $$ Thus $V_\Lambda(\theta)$,  the diagonal part of $H_\Lambda (\theta)$, satisfies   $\| (V_\Lambda (\theta)-E^*)^{-1}\|\leq2\delta_0^{-1}$.  Since $\|\Delta\|\leq2d$, it follows that 
		$$\varepsilon\|\left (V_\Lambda(\theta)-E^*\right)^{-1} \Delta\| \leq 4d\varepsilon\delta_0^{-1}\leq\frac{1}{2}.$$ Hence  by Neumann series, we have   
		\begin{align*}
			\left(H_\Lambda (\theta)-E^*\right)^{-1} & =\left (\varepsilon \Delta+V_\Lambda(\theta)-E^*\right)^{-1} \\
			& =\sum_{l=0}^{\infty} (-1)^l \varepsilon^l\left[\left (V_\Lambda(\theta)-E^*\right)^{-1} \Delta\right]^l\left (V_\Lambda(\theta)-E^*\right)^{-1}.
		\end{align*}
		Thus\begin{align*}
			\|\left(H_\Lambda (\theta)-E^*\right)^{-1}\|&\leq \sum_{l=0}^{\infty}\left (\frac{4 d\varepsilon}{\delta_0}\right)^l \|\left (V_\Lambda(\theta)-E^*\right)^{-1}\| \\ &\leq2\|\left (V_\Lambda(\theta)-E^*\right)^{-1}\|\\&\leq4\delta_0^{-1}.
		\end{align*}
		Since $\Delta$ has 
		only nearest neighbor matrix elements, provided  $\|x-y\|_1\geq1$, we have 
		\begin{align*}
			\left|\left(H_\Lambda (\theta)-E^*\right)^{-1} (x, y)\right|  &\leq \sum_{l=\|x-y\|_1}^{\infty}\left (\frac{4 d\varepsilon}{\delta_0}\right)^l \|\left (V_\Lambda(\theta)-E^*\right)^{-1}\| \\&\leq \frac{4}{\delta_0}\left (\frac{4 d\varepsilon}{\delta_0}\right)^{\|x-y\|_1}\\ 
			&\leq e^{-\gamma_0(\frac{3}{2}\|x-y\|_1-\frac{1}{2})}\\
			&\leq e^{-\gamma_0\|x-y\|_1}.
		\end{align*} 
	\end{proof}
	\subsection{The first inductive scale}\label{n=1}
	Let $l_1:=[|\ln\delta_0|^4]$ be the first inductive length and $B_1=Q_{l_1}$ be the first scale block. 
	Recall that $H_{B_1}(\theta)$ is the Dirichlet restriction of $H(\theta)$ on $\ell^2(B_1)$.    Listing the   eigenvalues in non-decreasing order,  $H_{B_1}(\theta)$ has  $|B_1|$ branches of  Rellich functions  $\lambda_i(\theta)$ $(1\leq i\leq |B_1|)$. By the perturbation theory (Min-Max argument) for self-adjoint operators, the Rellich  functions $\lambda_i(\theta)$ are $1$-periodic and continuous on $[0,1)$ except on the finite set of points $0=\alpha_1<\alpha_2<\cdots<\alpha_{|B_1|}<\alpha_{|B_1|+1}=1$, where 
	$$\{\alpha_1,\alpha_2,\cdots,\alpha_{|B_1|}\}=\{\{-x\cdot\omega\}:\ x\in B_1\}$$
	with $\{-x\cdot\omega\}$ being the decimal part of $ -x\cdot\omega$.  
	Define $$\beta_x:=\{-x\cdot\omega\}.$$  
	From the  equation
	$$H_{B_1}(\beta_x)-H_{B_1}(\beta_x-0)=-{\bm e}_x{\bm e}_x^{\rm T},$$ it follows that   all  discontinuities  of $\lambda_i(\theta)$ are caused by  negative rank one perturbations due to the $x$-site diagonal element of  the matrix jumping from  $1$ to $0$, as  the jumps at finite energies will affect all eigenvalues. 
	Lipschitz monotonicity property \eqref{LC} implies the operator inequality 
	\begin{equation}\label{mno}
		H_{B_1}(\theta_2)-H_{B_1}(\theta_1)\geq L(\theta_2-\theta_1)\operatorname{Id}_{B_1}, \ \ \alpha_k\leq \theta_1\leq \theta_2<\alpha_{k+1},
	\end{equation} where $\operatorname{Id}_{B_1}$ is the identity operator on $\ell^2(B_1)$.
	Hence it follows from the  Min-Max argument that  the Rellich functions $\lambda_i(\theta)$ also satisfy the local Lipschitz monotonicity property
	\begin{equation}\label{mne}
		\lambda_i(\theta_2)-\lambda_i(\theta_1)\geq L(\theta_2-\theta_1), \ \ \alpha_k\leq \theta_1\leq \theta_2<\alpha_{k+1}.
	\end{equation}
	The goal of this section is to construct the first generation Rellich function $E_1(\theta)$, which is a  $1$-periodic, single-valued, real-valued Rellich function of $H_{B_1}(\theta)$. The construction is done by locally choosing  the  unique branch of $\lambda_i(\theta)$ which is resonant with the previous scale Rellich function $E_0(\theta)$. We will prove that in our construction, on each small interval $[\alpha_k,\alpha_{k+1})$,  $E_1(\theta)$ coincides with one branch of $\lambda_i(\theta)$, and more importantly, $E_1(\theta)$ has  non-negative jumps at $\alpha_k$ except $\alpha_1(=0)$, that is, 
	$$E_1(\alpha_k)\geq E_1(\alpha_k-0), \ \ 2\leq k\leq |B_1|. $$ 
	It may seem strange at the first sight that the Rellich function $E_1(\theta)$ has   non-negative jumps since  negative rank one perturbations will cause all the eigenvalues to drop down. The reason of these non-negative jumps is that $E_1(\theta)$ coincides with different branches of $\lambda_i(\theta)$ before and after $\alpha_k$. Hence $E_1(\theta)$ maintains the Lipschitz monotonicity property from $v(\theta)$ on the whole interval $[0,1)$. After constructing $E_1(\theta)$ and verifying the properties of it, for fixed $\theta$ and $E$, we relate the first scale resonant points  $p\in S_1(\theta,E)$ to the difference  $|E_1(\theta+p\cdot\omega)-E|$ and prove good Green's function estimates for $1$-good  set.
	\subsubsection{Construction of $E_1(\theta)$}
	\begin{lem}\label{sep}
		For any $\theta\in \R$ and  $x\in B_1$ such that $x\neq o$, we have $$|v(\theta+x\cdot\omega)-v(\theta)|\geq 20\delta_0.$$
	\end{lem}
	\begin{proof}
		By \eqref{DC}, \eqref{LC} and $l_1^{-\tau}\geq  |\ln\delta_0|^{-4\tau}$, provided $\delta_0$ sufficiently small, we have  $$|v(\theta+x\cdot\omega)-v(\theta)|\geq L\|x\cdot\omega\|_\T\geq L\gamma l_1^{-\tau}\geq  L\gamma  |\ln\delta_0|^{-4\tau} \geq 20\delta_0.$$
	\end{proof}
	\begin{prop}\label{715}
		For any $\theta\in \R$, $H_{B_1}(\theta)$ has a unique eigenvalue $E_1(\theta)$ such that $|E_1(\theta)-v(\theta)|\leq \varepsilon$. Moreover, any other eigenvalues of  $H_{B_1}(\theta)$, $\hat E$, except $E_1(\theta)$,  satisfy $|\hat E-v(\theta)|>10\delta_0$.
	\end{prop}
	
	\begin{proof}
		Fix $\theta\in \R$. Consider $d_\theta(z):=\operatorname{det}[z\operatorname{Id}_{B_1}-H_{B_1}(\theta)]$, which is the characteristic polynomial of $H_{B_1}(\theta)$. In the following we restrict $z$ to the complex neighborhood $$D:=D(v(\theta),10\delta_0).$$
		It suffices to prove that $d_\theta(z)$ has a unique zero $E_1(\theta)$ in $\bar{D}$ satisfying $|E_1(\theta)-v(\theta)|\leq \varepsilon$.
		Let us  denote $B_1\setminus \{o\}$ by  $B_1^o$. We can write 
		$$	z\operatorname{Id}_{B_1}-H_{B_1}(\theta)=\begin{pmatrix}
			z-v(\theta) & c^{\text{T}} \\
			c & z\operatorname{Id}_{B_1^o}-H_{B_1^o}(\theta)
		\end{pmatrix}.$$
		By Lemma \ref{sep}, for $x\in B_1^o$ and $z\in \bar{D}$, we have  
		$$|z-v(\theta+x\cdot\omega)|\geq |v(\theta)-v(\theta+x\cdot\omega)|-|z-v(\theta)|\geq20\delta_0-10\delta_0=10\delta_0. $$
		Thus $B_1^o\cap S_0(\theta,z)=\emptyset$ and hence by Proposition \ref{0g}, we have 
		\begin{equation}\label{gb}
			\|G_{B_1^o}^{\theta,z}\|\leq 10\delta_0^{-1}.
		\end{equation}
		Denote the Schur complement by
		\begin{align}\label{923}
			s_\theta(z):&=z-v(\theta) -c^{\text{T}}\left(z\operatorname{Id}_{B_1^o}-H_{B_1^o}(\theta)\right)^{-1}c\nonumber\\
			&=:z-v(\theta)-r_\theta(z).
		\end{align}
		It follows from Schur complement formula that  $d_\theta(z)=0$ if and only if $s_\theta(z)=0$. By \eqref{gb} and $\|c\|\leq 2d\varepsilon$, we have \begin{equation}\label{556}
			|r_\theta(z)|\leq 40d^2\varepsilon^2\delta_0^{-1}\leq \varepsilon.
		\end{equation} Hence, 
		$$|r_\theta(z)|<|z-v(\theta)|=10\delta_0, \ \ z\in \partial D.$$
		It follows that by Rouch\'e theorem, $s_\theta(z)$ has the same number of zeros as $z-v(\theta)$ in $\bar{D}$. Denote by $E_1(\theta)$ the unique zero of $s_\theta(z)$ in $\bar{D}$. Since   $s_\theta(E_1(\theta))=E_1(\theta)-v(\theta)-r_\theta(E_1(\theta))=0$ and because of \eqref{556}, we have 
		\begin{equation}\label{841}
			|E_1(\theta)-v(\theta)|=|r_\theta(E_1(\theta))|\leq \varepsilon.
		\end{equation}
	\end{proof}
	\begin{rem}\label{454}
		Since $v(\theta)$ is  $1$-periodic, $H_{B_1}(\theta)$ is $1$-periodic and self-adjoint, it follows that $E_1(\theta)$ is a $1$-periodic,   real-valued function.
	\end{rem}
	
	\begin{lem}\label{sl}
		With the notation in the proof of Proposition \ref{715}, we have the following approximation, for  $z\in D$,  \begin{equation}\label{app}
			s_\theta(z)=(z-E_1(\theta))(1+O(\delta_0)).
		\end{equation}
	 Define   $$ D':=D(v(\theta),5\delta_0).$$ Then for   $z\in D'$, \begin{equation}\label{De}
			s_\theta'(z)=1+O(\delta_0) 	.
		\end{equation}
		The above  $``O(\delta_0)"$ are understood as  analytic functions of $z$ bounded by  $\delta_0$.
	\end{lem}
	\begin{proof}
		Since $E_1(\theta)$ is the unique zero of $s_\theta(z)$ in $\bar{D}$, the function $$h_\theta(z):=\frac{s_\theta(z)}{z-E_1(\theta)}-1$$
		is analytic in $D$. Moreover, by \eqref{556} and \eqref{841}, one has  for $z\in \partial D$,
		\begin{align*} 
			|h_\theta(z)|=\left|\frac{E_1(\theta)-v(\theta)-r_\theta(z)}{z-E_1(\theta)}\right|&\leq\left|\frac{E_1(\theta)-v(\theta)}{z-E_1(\theta)}\right| +\left|\frac{r_\theta(z)}{z-E_1(\theta)}\right|\\&\leq \frac{2\varepsilon}{10\delta_0-\varepsilon}\\ 
			&<\delta_0.
		\end{align*}
		Hence by maximum principle, we have  
		$$\sup_{z\in D}|h_\theta(z)|\leq \sup_{z\in \partial D}|h_\theta(z)|<\delta_0$$ and \eqref{app} follows. By \eqref{923}, we have  $s_\theta'(z)=1-r_\theta'(z)$. By   \eqref{556} and Cauchy integral estimate, it follows that 
		$$\sup_{z\in D'} |r_\theta'(z)|\leq \delta_0^{-1}\sup_{z\in D} |r_\theta(z)|\leq \delta_0^{-1}\varepsilon\leq \delta_0.$$ Thus we finish the proof of \eqref{De}.
		
	\end{proof}
	\begin{rem}
		Analogues of Proposition \ref{715} and Lemma \ref{sl} hold true for the left limit operator $H_{B_1}(\theta-0)$. We state them  as a proposition:
		\begin{prop}\label{*2}
			For any $\theta\in \R$, $H_{B_1}(\theta-0)$ has a unique eigenvalue $E_1(\theta-0)$ such that $|E_1(\theta-0)-v(\theta-0)|\leq \varepsilon$. Moreover, since $H_{B_1}(\theta)=H_{B_1}(\theta-0)$ for $\theta\notin \left\{\{-x\cdot\omega\}\right\}_{x\in B_1}$, it follows that $E_1(\theta)=E_1(\theta-0)$ for $\theta\notin \left\{\{-x\cdot\omega\}\right\}_{x\in B_1}$.
		\end{prop}
		\begin{lem}\label{H-} For $z\in  D(v(\theta-0),10\delta_0)$, the Schur complement 
			\begin{align}\label{923.}
				s_{\theta-0}(z):&=z-v(\theta-0) -c^{\text{T}}\left(z\operatorname{Id}_{B_1^o}-H_{B_1^o}(\theta-0)\right)^{-1}c\nonumber	\\
				&=z-v(\theta-0)-r_{\theta-0}(z)
			\end{align}
			exists and 	satisfies 
			\begin{align*}
				s_{\theta-0}(z)&=(z-E_1(\theta-0))(1+O(\delta_0)) ,  & z\in  D(v(\theta-0),10\delta_0),\\
				s_{\theta-0}'(z)&=1+O(\delta_0) , & z\in  D(v(\theta-0),5\delta_0).
			\end{align*}
			
		\end{lem}
		
	\end{rem}

	\subsubsection{Verification of Lipschitz monotonicity property for $E_1(\theta)$}
	In this part, we prove that $E_1(\theta)$ inherits the Lipschitz monotonicity property of $v(\theta)$. 
	\begin{prop}\label{814}
		Let $E_1(\theta)$ be the eigenvalue of $H_{B_1}(\theta)$ constructed in Proposition \ref{715}.  Then $E_1(\theta)$ satisfies the following two properties{\rm :} 
		\begin{itemize}	\item On each small interval  $[\alpha_k,\alpha_{k+1})$, $E_1(\theta)$ coincides with exactly one branch of $\lambda_i(\theta)$, and hence is continuous and satisfies the local Lipschitz monotonicity property  on $[\alpha_k,\alpha_{k+1})$ by \eqref{mne}{\rm :}
			\begin{equation*}
				E_{1}(\theta_2)-E_{1}(\theta_1)\geq L(\theta_2-\theta_1), \ \ \alpha_k\leq \theta_1\leq \theta_2<\alpha_{k+1}.
			\end{equation*} 
			\item At  $\alpha_k$ $(2\leq k\leq |B_1|)$, the   discontinuities of $H_{B_1}(\theta)$,  $E_1(\theta)$ has  non-negative ``jumps''{\rm :}
			$$E_1(\alpha_k)\geq E_1(\alpha_k-0). $$ 
		\end{itemize}
		As a consequence, $E_1(\theta)$ satisfies the uniform Lipschitz monotonicity property{\rm :}	$$E_1(\theta_2)-E_1(\theta_1)\geq L(\theta_2-\theta_1), \ \ 0\leq\theta_1\leq \theta_2<1.$$
	\end{prop}
	\begin{proof}
		To prove the first item, we assume that $E_1(\alpha_k)=\lambda_i(\alpha_k)$ for some branch $\lambda_i(\theta)$. Define the set $X=\{\theta\in[\alpha_k,\alpha_{k+1}) :\ E_1(\theta)=\lambda_i(\theta) \}$. We claim that   $X=[\alpha_k,\alpha_{k+1})$. By the  definition,  $\alpha_k\in X$, thus $X\neq\emptyset$. Assume there is  a sequence  $\{\theta_j\}_{j=1}^{\infty}\subset X$ such that $\theta_j\to\theta_\infty\in [\alpha_k,\alpha_{k+1})$. Hence by Proposition \ref{715}, one has 
		\begin{equation}\label{721}
			|\lambda_i(\theta_j)-v(\theta_j)|=|E_1(\theta_j)-v(\theta_j)|\leq \varepsilon.
		\end{equation}
		Since $\lambda_i(\theta)$ and $v(\theta)$ are continuous on $[\alpha_k,\alpha_{k+1})$, taking a limit of \eqref{721} yields $|\lambda_i(\theta_\infty)-v(\theta_\infty)|\leq \varepsilon$. By the definition of $E_1(\theta_\infty)$ and the uniqueness of eigenvalue of $H_{B_1}(\theta_\infty)$ in the energy interval $[v(\theta_\infty)-10\delta_0,v(\theta_\infty)+10\delta_0]$, we get $E_1(\theta_\infty)=\lambda_i(\theta_\infty)$ and hence $X$ is closed in $[\alpha_k,\alpha_{k+1})$.
		Assume that $\theta_0\in X$. Then by Proposition \ref{715},  for any $j\neq i$, one has 
		\begin{equation}\label{858}
			|\lambda_j(\theta_0)-v(\theta_0)|>10\delta_0.
		\end{equation}
		By the continuity of $\lambda_j(\theta)$ and $v(\theta)$ in $[\alpha_k,\alpha_{k+1})$, the inequality \eqref{858} can be extended to  a neighborhood   (denoted by $U$) of $\theta_0$ in $[\alpha_k,\alpha_{k+1})$. Hence by the definition of $E_1(\theta)$, $\lambda_j(\theta)\neq E_1(\theta)$ for any $\theta\in U$ and  $j\neq i$. So it must be that $\lambda_i(\theta)= E_1(\theta)$ in $U$. Hence $X$ is open in $[\alpha_k,\alpha_{k+1})$. Thus we finish the proof of the first item.
		For the second item, 	since $\alpha_k\neq 0$, there exists some $x\in B_1$ with  $x\neq o$ such   that $\alpha_k=\beta_x(=\{-x\cdot\omega\})$. Since $\beta_x\neq 0$, we have  $v(\beta_x)=v(\beta_x-0)$.
		It suffices to prove \begin{equation}\label{1052}
			E_1(\beta_x)\geq E_1(\beta_x-0).
		\end{equation} Recalling the notation $s_\theta(z)$, $r_\theta(z)$, $s_{\theta-0}(z)$, $r_{\theta-0}(z)$ (c.f. \eqref{923}, \eqref{923.}), 
		we will restrict $z$ to the interval $I:=(v(\beta_x)-5\delta_0,v(\beta_x)+5\delta_0)$ in the following discussion  so that \eqref{De} holds and  all the functions of $z$ will be real-valued.  
		If we can  prove \begin{equation}\label{1055}
			s_{\beta_x}(z)\leq s_{\beta_x-0}(z),
		\end{equation}
		then 
		$$s_{\beta_x}(E_1(\beta_x-0))\leq s_{\beta_x-0}(E_1(\beta_x-0))=0=s_{\beta_x}(E_1(\beta_x)),$$   and \eqref{1052} will follow from the above inequality and the monotonicity of  $s_{\beta_x}(z)$ (since by \eqref{De}, $s_{\beta_x}'(z)\approx1$). The remainder of the proof  aims at verifying \eqref{1055}. Recalling \eqref{923}, \eqref{923.}, since $v(\beta_x)=v(\beta_x-0)$,  we compute 
		\begin{align}
			&s_{\beta_x}(z)-s_{\beta_x-0}(z)\nonumber\\ =&r_{\beta_x-0}(z)-r_{\beta_x}(z)\nonumber\\
			=&c^{\text{T}}\left[ \left(z\operatorname{Id}_{B_1^o}-H_{B_1^o}(\beta_x-0)\right)^{-1}- \left(z\operatorname{Id}_{B_1^o}-H_{B_1^o}(\beta_x)\right)^{-1}\right]c\nonumber\\
			=&c^{\text{T}}\left[ \left(z\operatorname{Id}_{B_1^o}-H_{B_1^o}(\beta_x-0)\right)^{-1}{\bm e}_x {\bm e}_x^{\rm T}\left(z\operatorname{Id}_{B_1^o}-H_{B_1^o}(\beta_x)\right)^{-1}\right]c. \label{1210}
		\end{align}
		On the last line of the above equation  we use the resolvent identity and the  equation $$H_{B_1^o}(\beta_x-0)-H_{B_1^o}(\beta_x)={\bm e}_x {\bm e}_x^{\rm T}.$$
	Cramer's rule implies 
		$$\left(z\operatorname{Id}_{B_1^o}-H_{B_1^o}(\beta_x-0)\right)^{-1}=\frac{1}{\operatorname{det}\left(z\operatorname{Id}_{B_1^o}-H_{B_1^o}(\beta_x-0)\right)}\left(z\operatorname{Id}_{B_1^o}-H_{B_1^o}(\beta_x-0)\right)^{\#}$$
		and $$\left(z\operatorname{Id}_{B_1^o}-H_{B_1^o}(\beta_x)\right)^{-1}=\frac{1}{\operatorname{det}\left(z\operatorname{Id}_{B_1^o}-H_{B_1^o}(\beta_x)\right)}\left(z\operatorname{Id}_{B_1^o}-H_{B_1^o}(\beta_x)\right)^{\#}.$$
		A simple computation shows that 
		\begin{equation}\label{1437}
			\left(z\operatorname{Id}_{B_1^o}-H_{B_1^o}(\beta_x-0)\right)^{\#}{\bm e}_x=\left(z\operatorname{Id}_{B_1^o}-H_{B_1^o}(\beta_x)\right)^{\#}{\bm e}_x.
		\end{equation}
		We denote by $b(z)$ the vector in \eqref{1437}.
		Hence 
		\begin{equation}\label{1528}
			\eqref{1210}=\frac{1}{\operatorname{det}\left(z\operatorname{Id}_{B_1^o}-H_{B_1^o}(\beta_x-0)\right)}\frac{1}{\operatorname{det}\left(z\operatorname{Id}_{B_1^o}-H_{B_1^o}(\beta_x)\right)}\left(c^{\rm T} b(z)\right)^2.
		\end{equation}
		From \eqref{1528}, it follows that  the sign of $s_{\beta_x}(z)-s_{\beta_x-0}(z)$ is the same as the sign of  the product $$\operatorname{det}\left(z\operatorname{Id}_{B_1^o}-H_{B_1^o}(\beta_x-0)\right)\operatorname{det}\left(z\operatorname{Id}_{B_1^o}-H_{B_1^o}(\beta_x)\right).$$
		Denote by $B_1^1$ the set $B_1^o\setminus \{x\}$. 
		Since $$\langle H_{B_1^o}(\beta_x-0){\bm e}_x,{\bm e}_x\rangle=v(1-0)=1, \  \ \langle H_{B_1^o}(\beta_x){\bm e}_x,{\bm e}_x\rangle=v(0)=0$$ and  $H_{B_1^1}(\beta_x -0)=H_{B_1^1}(\beta_x )$, we  can write 
		$$z\operatorname{Id}_{B_1^o}-H_{B_1^o}(\beta_x-0)=\begin{pmatrix}
			z-1 & \tilde{c}^{\text{T}} \\
			\tilde{c} & z\operatorname{Id}_{B_1^1}-H_{B_1^1}(\beta_x)
		\end{pmatrix}
		$$
		and 
		$$z\operatorname{Id}_{B_1^o}-H_{B_1^o}(\beta_x)=\begin{pmatrix}
			z-0 & \tilde{c}^{\text{T}} \\
			\tilde{c} & z\operatorname{Id}_{B_1^1}-H_{B_1^1}(\beta_x)
		\end{pmatrix}.
		$$
		By Lemma \ref{sep} and $|z-v(\beta_x)|<5\delta_0$, it follows that  $B_1^1\cap S_0(\beta_x,z)=\emptyset$. By Proposition \ref{0g}, we get the resolvent bound 
		\begin{equation}\label{1623}
			\left \|\left( z\operatorname{Id}_{B_1^1}-H_{B_1^1}(\beta_x)  \right)^{-1}\right\|\leq 10\delta_0^{-1}.
		\end{equation}
		Hence by Schur complement formula, 
		\begin{align}\label{351}
			&\nonumber	\operatorname{det}\left(z\operatorname{Id}_{B_1^o}-H_{B_1^o}(\beta_x-0)\right)\\
			=&
			\left [z-1-\tilde{c}^{\rm T}\left(z\operatorname{Id}_{B_1^1}-H_{B_1^1}(\beta_x)\right)^{-1}\tilde{c}\right ]\operatorname{det}\left(z\operatorname{Id}_{B_1^1}-H_{B_1^1}(\beta_x)\right).
		\end{align}
		By Lemma \ref{sep}, we have $$|v(\beta_x)-v(1-0)|=|v(\beta_x)-1|\geq 20\delta_0.$$ Since $v(\beta_x)\leq v(1-0)=1$, it follows that $v(\beta_x)-1\leq -20\delta_0$.
		Now \eqref{1623} combined  with $\|\tilde{c}\|\leq 2d\varepsilon$ and $|z-v(\beta_x)|\leq 5\delta_0$ implies  
		\begin{equation}\label{352}
			z-1-\tilde{c}^{\rm T}\left(z\operatorname{Id}_{B_1^1}-H_{B_1^1}(\beta_x)\right)^{-1}\tilde{c}\leq v(\beta_x) -1+10\delta_0\leq -20\delta_0+10\delta_0<0. 
		\end{equation}
		Similarly,\begin{align}\label{353}
			\nonumber &	 \operatorname{det}\left(z\operatorname{Id}_{B_1^o}-H_{B_1^o}(\beta_x)\right)\\
			=&\left [z-0-\tilde{c}^{\rm T}\left(z\operatorname{Id}_{B_1^1}-H_{B_1^1}(\beta_x)\right)^{-1}\tilde{c}\right ]\operatorname{det}\left(z\operatorname{Id}                                                           _{B_1^1}-H_{B_1^1}(\beta_x)\right). 
		\end{align}
		By Lemma \ref{sep}, we have $$|v(\beta_x)-v(0)|=|v(\beta_x)-0|\geq 20\delta_0.$$ 
		Since $v(\beta_x)\geq v(0)=0$, it follows that $v(\beta_x)-0\geq 20\delta_0$. By \eqref{1623}, $\|\tilde{c}\|\leq 2d\varepsilon$, $|z-v(\beta_x)|\leq 5\delta_0$, we have 
		\begin{equation}\label{354}
			z-0-\tilde{c}^{\rm T}\left(z\operatorname{Id}_{B_1^1}-H_{B_1^1}(\beta_x)\right)^{-1}\tilde{c}\geq  v(\beta_x) -0-10\delta_0\geq 20\delta_0-10\delta_0>0. 
		\end{equation}
		Combing \eqref{351}, \eqref{352}, \eqref{353} and \eqref{354}, we get 
		$$\operatorname{det}\left(z\operatorname{Id}_{B_1^o}-H_{B_1^o}(\beta_x-0)\right)\operatorname{det}\left(z\operatorname{Id}_{B_1^o}-H_{B_1^o}(\beta_x)\right)<0$$ and hence 
		\eqref{1055} holds true  by \eqref{1528}. Thus we finish the proof of \eqref{1052} and the second item.
	\end{proof}
	\begin{rem}\label{.}
		Another  proof of Proposition \ref{814} depends on eigenvalue perturbation theory. The separation property of the  diagonal  elements  \begin{equation}\label{842}
			\min_{\substack{x,y\in B_1\\ x\neq y}}|v(\theta+x\cdot\omega)-v(\theta+y\cdot\omega)| \geq L\gamma (2dl_1)^{-\tau}>10d\varepsilon
		\end{equation} yields  the separation of the Rellich functions $\lambda_i(\theta)$ after adding an $\varepsilon$-perturbative Laplacian.
		Thus one can verify that $E_1(\theta)$ coincides with $\lambda_k(\theta)$ (the $k$-branch of Rellich functions in non-decreasing order) on $[\alpha_k,\alpha_{k+1})$. The non-negative ``jumps'' at $\alpha_k$ follow from a min-max principle argument for rank one perturbation. 
		However, this easier idea seems difficult to   be  extended to large scale case since \eqref{842} will be violated as the size of the block increases.   
	\end{rem}
	\subsubsection{Green's function estimates for $1$-good sets}
	In this part, we relate the resonances to $E_1(\theta)$ and establish Green's function estimates for $1$-good sets. 
	For $p\in \Z^d$, we denote  $$B_1(p):=B_1+p.$$  
	\begin{defn}
		Fix $\theta\in \R$, $E\in \C$. Let $\delta_1:=e^{-l_1^\frac{2}{3}}$. Define the first scale resonant points set (related to $(\theta,E)$) as  $$S_1(\theta,E):=\left\{p\in \Z^d:\ |E_1 (\theta+p\cdot \omega)-E|<\delta_1\right\}.$$
		Let $\Lambda\subset\Z^d$ be a finite set. Related to $(\theta,E)$,  we say
		\begin{itemize}
			\item $\Lambda$ is $1$-nonresonant  if $\Lambda\cap S_1(\theta,E)=\emptyset$.
			\item $\Lambda$ is $1$-regular if $p\in \Lambda\cap S_0(\theta,E)$ $\Rightarrow$ $B_1(p)\subset\Lambda$.
			\item  $\Lambda$ is $1$-good if it is both $1$-nonresonant and $1$-regular.
		\end{itemize}
		
	\end{defn}
	\begin{lem}\label{1727}
		Fix $\theta\in \R$, $E\in \C$. If $p\in S_0(\theta,E)$ and $p \notin S_1(\theta,E)$, then for any $E^*\in \C$ such that $|E-E^*|<\delta_1/2$, we have 
		\begin{align}\label{1l2}
			\|G_{B_1(p)} ^{\theta,E^*}\|&\leq10\delta_1^{-1}, \\
			\label{1ex}
			|G_{B_1(p)}^{\theta,E^*}(x,y)|&\leq e^{-\gamma_0(1-l_1^{-\frac{1}{50}})\|x-y\|_1}, \ \|x-y\|_1\geq l_1^{\frac{4}{5}}.
		\end{align}
	\end{lem}
	\begin{proof}
		Since $p\in S_0(\theta,E)$, by definition we have $|v(\theta+p\cdot \omega )-E|\leq \delta_0$. Thus $$ |v(\theta+p\cdot \omega )-E^*|\leq |v(\theta+p\cdot \omega )-E| +|E-E^*|\leq \delta_0+\delta_1/2.$$ By Proposition \ref{715} and translation property, $E_1(\theta+p\cdot \theta )$ is the unique eigenvalue of $H_{B_1(p)}(\theta)$ in $[v(\theta+p\cdot \omega )-10\delta_0,v(\theta+p\cdot \omega )+10\delta_0]$ and moreover 
		\begin{align*}\
			|E_1(\theta+p\cdot \omega )-E^*|&\leq  |v(\theta+p\cdot \omega )-E^*|+|E_1(\theta+p\cdot \omega )-v(\theta+p\cdot \omega )|\\&\leq \delta_0+\delta_1/2+\varepsilon\\&<2\delta_0.
		\end{align*}
		Hence we get 
		$$\operatorname{dist}\left (\operatorname{Spec}(H_{B_1(p)}(\theta )),E^* \right )=|E_1(\theta+p\cdot \omega )-E^*|.$$
		Since $p \notin S_1(\theta,E)$, we have $$|E_1(\theta+p\cdot \omega )-E^*|\geq |E_1(\theta+p\cdot \omega )-E|-|E-E^*|\geq \delta_1-\delta_1/2\geq \delta_1/2.$$
		Since $H_{B_1(p)}(\theta )$ is self-adjoint,  we get  $$\|G_{B_1(p)} ^{\theta,E^*}\|=\frac{1}{\operatorname{dist}\left (\operatorname{Spec}(H_{B_1(p)}(\theta )),E^* \right )}=\frac{1}{|E_1(\theta+p\cdot \omega )-E^*|}\leq 2\delta_1^{-1}.$$ Thus we finish the proof of \eqref{1l2}.
		Denote  $B_1^o(p):=B_1(p)\setminus\{p\}$. By Lemma \ref{sep}, it follows that for any $x\in B_1^o(p)$, we have\begin{align*}
			|v(\theta+x\cdot\omega)-E|&\geq|v(\theta+x\cdot\omega)-v(\theta+p\cdot\omega)| -|v(\theta+p\cdot\omega)-E|\\&\geq 20\delta_0-\delta_0\\&>\delta_0.
		\end{align*} Thus $B_1^o(p)\cap S_0(\theta,E)=\emptyset$. Hence by Proposition \ref{0g}, we have 
		\begin{align}
			\|G_{B_1^o(p)} ^{\theta,E^*}\|&\leq10\delta_0^{-1}, \nonumber\\
			\label{1200}
			|G_{B_1^o(p)}^{\theta,E^*}(x,y)|&\leq e^{-\gamma_0\|x-y\|_1}, \ \|x-y\|_1\geq 1.
		\end{align}  Let $x,y\in B_1(p)$ such that $\|x-y\|_1\geq l_1^{\frac{4}{5}}$. By self-adjointness, we may assume $\|x-p\|_1\geq \frac{1}{2}l_1^\frac{4}{5}$.
	Recall the definition of $\partial$ and $\partial^\pm$. 	By the resolvent identity, one has 
		\begin{align}\label{1156}
			G_{B_1(p)}^{\theta,E^*}(x,y)=\chi_{B_1^o(p)}(y)G_{B_1^o(p)}^{\theta,E^*}(x,y)-\varepsilon\sum_{(w,w')\in \partial_{B_1(p)}B_1^o(p)} G_{B_1^o(p)}^{\theta,E^*}(x,w)G_{B_1(p)}^{\theta,E^*}(w',y).
		\end{align}
			Since $w \in  \partial^-_{B_1(p)}B_1^o(p)$, $w' \in  \partial^+_{B_1(p)}B_1^o(p)=\{p\}$, one has 
			$$\|x-w\|_1\geq \|x-p\|_1-\|w-p\|_1\geq \|x-p\|_1-1.$$ Hence by  \eqref{1200} and  \eqref{1156}, we have 
			\begin{equation}\label{1205}
				|G_{B_1(p)}^{\theta,E^*}(x,y)|\leq e^{-\gamma_0\|x-y\|_1}+e^{-\gamma_0(\|x-p\|_1-1)}|G_{B_1(p)}^{\theta,E^*}(p,y)|.
			\end{equation}
			\begin{itemize}
				\item 
				If $\|y-p\|_1\leq l_1^\frac{3}{4}$, then    \begin{align*}
					\|x-p\|_1-1&\geq \|x-y\|_1-\|y-p\|_1-1\\
					&\geq \left (1-(\l_1^\frac{3}{4}+1) l_1^{-\frac{4}{5}}\right) \|x-y\|_1\\
					&\geq \left (1- 2l_1^{-\frac{1}{20}}\right) \|x-y\|_1.
				\end{align*}
				Bounding the term $|G_{B_1(p)}^{\theta,E^*}(p,y)|$ in \eqref{1205} by $\|G_{B_1(p)}^{\theta,E^*}\|\leq 10\delta_1^{-1}=10e^{l_1^\frac{2}{3}}$, we get \begin{align*}
					\eqref{1205}&\leq e^{-\gamma_0\|x-y\|_1}+10e^{-\gamma_0(1- 2l_1^{-\frac{1}{20}})\|x-y\|_1}e^{l_1^\frac{2}{3}}  \\
					&\leq e^{-\gamma_0(1- 2l_1^{-\frac{1}{20}}-(l_1^\frac{2}{3}+10) l_1^{-\frac{4}{5}})\|x-y\|_1}\\
					&\leq e^{-\gamma_0(1-l_1^{-\frac{1}{50}})\|x-y\|_1}.
				\end{align*}
				\item  If $\|y-p\|_1> l_1^\frac{3}{4}$, we can use  \eqref{1l2}, \eqref{1200} and the resolvent identity  to bound the term   $|G_{B_1(p)}^{\theta,E^*}(p,y)|$ in \eqref{1205} by 
				\begin{align*}
					|G_{B_1(p)}^{\theta,E^*}(p,y)|=|G_{B_1(p)}^{\theta,E^*}(y,p)|&\leq \varepsilon\sum_{(w,w')\in \partial_{B_1(p)}B_1^o(p)} |G_{B_1^o(p)}^{\theta,E^*}(y,w)G_{B_1(p)}^{\theta,E^*}(w',p)|\\
					&\leq e^{-\gamma_0(\|y-p\|_1-1)} \delta_1^{-1}.
				\end{align*}
				Thus 
				\begin{align*}
					\eqref{1205} &\leq e^{-\gamma_0\|x-y\|_1}+e^{-\gamma_0(\|x-p\|_1-1)}e^{-\gamma_0(\|y-p\|_1-1)} \delta_1^{-1}\\&\leq e^{-\gamma_0(1-(l_1^\frac{2}{3} +10) l_1^{-\frac{4}{5}})\|x-y\|_1}\\
					&\leq e^{-\gamma_0(1-l_1^{-\frac{1}{50}})\|x-y\|_1}.
				\end{align*}
			\end{itemize}
			Thus we finish the proof of \eqref{1ex}.
		\end{proof}
		\begin{prop}\label{1g}
			Fix $\theta\in \R$, $E\in \C$. Assume that  a finite set $\Lambda$ is $1$-good related to $(\theta,E)$,  then for any $E^*\in \C$ such that $|E-E^*|<\delta_1/5$, we have 
			\begin{align}\label{1l}
				\|G_{\Lambda} ^{\theta,E^*}\|&\leq10\delta_1^{-1}, 
				\\
				|G_{\Lambda} ^{\theta,E^*}(x,y)|&\leq e^{-\gamma_1\|x-y\|_1}, \ \|x-y\|_1\geq l_1^{\frac{5}{6}},\label{1e}
			\end{align}
			where $\gamma_1=\gamma_0(1-l_1^{-\frac{1}{80}})$.
		\end{prop}
		\begin{proof}
			Define the set $S_0^\Lambda:=S_0(\theta,E)\cap\Lambda$. If $S_0^\Lambda=\emptyset$, then \eqref{1l} and \eqref{1e} follows from  Proposition \ref{0g} immediately. Thus we assume that $S_0^\Lambda\neq \emptyset$. 
			To prove \eqref{1l}, by self-adjointness it suffices to prove 
			\begin{equation}\label{sd}
				\operatorname{dist}\left (\operatorname{Spec}(H_{\Lambda}(\theta )),E \right )\geq  \delta_1 /2.
			\end{equation}
			Fix $E'$ such that $|E'-E|<\delta_1/2$ and $E'\notin\operatorname{Spec}(H_{\Lambda}(\theta ))$.
			For any $p\in S_0^\Lambda$, since $\Lambda$ is $1$-regular, it follows that  $B_1(p)\subset\Lambda$. Since $\Lambda$ is $1$-nonresonant, it follows that  $p\notin S_1(\theta,E)$. Hence Lemma \ref{1727} holds true for such $p$.
			Denote $\Lambda_0:=\Lambda\setminus S_0^\Lambda$. Then $\Lambda_0\cap S_0(\theta,E)=\emptyset$ and hence Proposition \ref{0g} holds true for $\Lambda_0$.
			\begin{itemize}
				\item  Let $x\in \Lambda$ such that $\operatorname{dist}_1(x,S_0^\Lambda)\geq l_1^\frac{2}{3}$. By the resolvent identity, one has 
				$$G_\Lambda^{\theta,E'}(x,y)=\chi_{\Lambda_0}(y)G_{\Lambda_0}^{\theta,E'}(x,y)-\varepsilon\sum_{(w,w')\in \partial_{\Lambda}\Lambda_0} G_{\Lambda_0}^{\theta,E'}(x,w)G_{\Lambda}^{\theta,E'}(w',y).$$
				Thus \begin{align}\label{1835}
					\sum_{y\in \Lambda}|G_\Lambda^{\theta,E'}(x,y)|&\leq \sum_{y\in \Lambda_0} |G_{\Lambda_0}^{\theta,E'}(x,y)|+\varepsilon\sum_{\substack{(w,w')\in \partial_{\Lambda}\Lambda_0\\y\in \Lambda}} |G_{\Lambda_0}^{\theta,E'}(x,w)||G_{\Lambda}^{\theta,E'}(w',y)|.
				\end{align}
				Since  $|E-E'|<\delta_1/2$, by Proposition \ref{0g}, the first summation on the right-hand side of \eqref{1835} satisfies
				\begin{align*}
					\sum_{y\in \Lambda_0} |G_{\Lambda_0}^{\theta,E'}(x,y)|&\leq \|G_{\Lambda_0}^{\theta,E'}\|+\sum_{y\in\Lambda_0:\|x-y\|_1\geq1 }|G_{\Lambda_0}^{\theta,E'}(x,y)|\\
					&\leq 10\delta_0^{-1}+\sum_{y:\|x-y\|_1\geq1 } e^{-\gamma_0\|x-y\|_1}\\
					&\leq 20\delta_0^{-1}.
				\end{align*}
				In the second summation on the right-hand side of \eqref{1835}, since $w\in \partial_{\Lambda}^-\Lambda_0,$ we have 
				$$\|x-w\|_1\geq\operatorname{dist}_1(x,S_0^\Lambda)-\operatorname{dist}_1(w,S_0^\Lambda)\geq l_1^\frac{2}{3}-1.$$ It follows that  
				\begin{align*}
					&\varepsilon\sum_{\substack{(w,w')\in \partial_{\Lambda}\Lambda_0\\y\in \Lambda}} |G_{\Lambda_0}^{\theta,E'}(x,w)||G_{\Lambda}^{\theta,E'}(w',y)|\\ \leq&2d\varepsilon \sum_{w\in \partial_{\Lambda}^-\Lambda_0} e^{-\gamma_0\|x-w\|_1}	\sup_{w'\in \partial_{\Lambda}^+\Lambda_0}\sum_y|G_\Lambda^{\theta,E'}(w',y)|\\
					\leq & 2d\varepsilon \sum_{{w:\|x-w\|_1\geq l_1^\frac{2}{3}-1}} e^{-\gamma_0\|x-w\|_1}	\sup_{w'\in \partial_{\Lambda}^+\Lambda_0}\sum_y|G_\Lambda^{\theta,E'}(w',y)| 
					\\ \leq &\frac{1}{4}	\sup_{w'\in \Lambda}\sum_y|G_\Lambda^{\theta,E'}(w',y)|.
				\end{align*}
				Hence, \begin{equation}\label{1906}
					\sup_{x: \operatorname{dist}_1(x,S_0^\Lambda)\geq l_1^\frac{2}{3}}\sum_{y\in \Lambda}|G_\Lambda^{\theta,E'}(x,y)|\leq 20\delta_0^{-1}+ \frac{1}{4}	\sup_{w'\in \Lambda}\sum_{y\in \Lambda}|G_\Lambda^{\theta,E'}(w',y)|.
				\end{equation}
				\item Let $x\in \Lambda$ such that $\operatorname{dist}_1(x,S_0^\Lambda)< l_1^\frac{2}{3}$. Then $\|x-p\|_1<l_1^\frac23$ and  $x\in B_1(p)$ for some $p\in S_0^\Lambda$.
				By the  resolvent identity, one has 
				\begin{equation*}
					G_\Lambda^{\theta,E'}(x,y)=\chi_{B_1(p)}(y)G_{B_1(p)}^{\theta,E'}(x,y)-\varepsilon\sum_{(w,w')\in \partial_\Lambda B_1(p)} G_{B_1(p)}^{\theta,E'}(x,w)G_{\Lambda}^{\theta,E'}(w',y).
				\end{equation*}
				Thus \begin{align}\label{1911}
					\sum_{y\in \Lambda}|G_\Lambda^{\theta,E'}(x,y)|\leq \sum_{y\in B_1(p)} |G_{B_1(p)}^{\theta,E'}(x,y)|+\varepsilon\sum_{\substack{(w,w')\in \partial_\Lambda B_1(p)\\ y\in \Lambda}} |G_{B_1(p)}^{\theta,E'}(x,w)||G_{\Lambda}^{\theta,E'}(w',y)|.
				\end{align}
				By Lemma \ref{1727},  the first summation on the right-hand side of \eqref{1911} satisfies
				\begin{align*}
					\sum_{y\in B_1(p)} |G_{B_1(p)}^{\theta,E'}(x,y)|&\leq (3l_1)^d\|G_{B_1(p)}^{\theta,E'}\|\leq \delta_1^{-2}.
				\end{align*}
				In the second summation on the right-hand side of \eqref{1911}, since $w\in \partial_\Lambda^-B_1(p)$, we have  $\|w-p\|_1\geq l_1$. It follows that  $$\|x-w\|_1\geq \|w-p\|_1- \|x-p\|_1\geq  l_1-l_1^\frac{2}{3}\geq l_1^\frac{4}{5}.$$ Since $|E'-E|<\delta_1/2$, by Lemma \ref{1727}, one has 
				\begin{align*}
					&\varepsilon\sum_{\substack{(w,w')\in \partial_\Lambda B_1(p)\\ y\in \Lambda}} |G_{B_1(p)}^{\theta,E'}(x,w)||G_{\Lambda}^{\theta,E'}(w',y)|\\	\leq& 2d\varepsilon \sum_{w\in \partial_\Lambda^-B_1(p)} |G_{B_1(p)}^{\theta,E'}(x,w)|	\sup_{w'\in \partial_\Lambda^+B_1(p)}\sum_{y\in \Lambda}|G_\Lambda^{\theta,E'}(w',y)|\\ \leq &
					2d\varepsilon \sum_{w:\|x-w\|_1\geq l_1^\frac{4}{5}} e^{-\gamma_0(1-l_1^{-\frac{1}{50}})\|x-w\|_1}	\sup_{w'\in \Lambda}\sum_{y\in \Lambda}|G_\Lambda^{\theta,E'}(w',y)|\\ \leq& \frac{1}{4}	\sup_{w'\in \Lambda}\sum_{y\in \Lambda}|G_\Lambda^{\theta,E'}(w',y)|.
				\end{align*}
				Hence, \begin{equation}\label{1940}
					\sup_{x: \operatorname{dist}_1(x,S_0^\Lambda)< l_1^\frac{2}{3}}\sum_{y\in \Lambda}|G_\Lambda^{\theta,E'}(x,y)|\leq \delta_1^{-2}+ \frac{1}{4}	\sup_{w'\in \Lambda}\sum_{y\in \Lambda}|G_\Lambda^{\theta,E'}(w',y)|.
				\end{equation}
			\end{itemize}
			By \eqref{1906} and \eqref{1940}, we get 
			$$	\sup_{x\in \Lambda}\sum_{y\in \Lambda}|G_\Lambda^{\theta,E'}(x,y)|\leq \delta_1^{-2}+ \frac{1}{4}	\sup_{w'\in \Lambda}\sum_{y\in \Lambda}|G_\Lambda^{\theta,E'}(w',y)|$$ and hence 
			$$	\sup_{x\in \Lambda}\sum_{y\in \Lambda}|G_\Lambda^{\theta,E'}(x,y)|\leq 2\delta_1^{-2}.$$
			By Schur's test and self-adjointness, we get 
			\begin{equation}\label{1947}
				\| G_\Lambda^{\theta,E'}\|\leq	\sup_{x\in \Lambda}\sum_{y\in \Lambda}|G_\Lambda^{\theta,E'}(x,y)|\leq 2\delta_1^{-2} .
			\end{equation}
			Since the uniform bound \eqref{1947} holds true for all $E'$ in  $|E'-E|<\delta_1/2$  except a finite subset of $\operatorname{Spec}(H_{\Lambda}(\theta ))$,   \eqref{1947} must  hold true for all $E'$ in  $|E'-E|<\delta_1/2$ as the resolvent is unbounded near the spectrum. Thus we finish the proof of \eqref{sd}. To prove \eqref{1e}, we let $x,y\in \Lambda$ such that $\|x-y\|_1\geq l_1^{\frac{5}{6}}$. For any $q\in \Lambda$, we introduce a subset $U(q)\subset\Lambda$ with $q\in U(q)$ to iterate the resolvent identity.
			\begin{itemize}
				\item If $\operatorname{dist}_1(q,S_0^\Lambda)\geq l_1^\frac{2}{3}$, we define 
				\begin{equation}\label{sf1}
					U(q)=\{u\in \Lambda:\ \|q-u\|_1\leq l_1^\frac{1}{2}\}.
				\end{equation}
				By the resolvent identity,  one has 
				$$G_\Lambda^{\theta,E^*}(q,y)=\chi_{U(q)}(y)G_{U(q)}^{\theta,E^*}(q,y)-\varepsilon\sum_{(w,w')\in \partial_\Lambda U(q)} G_{U(q)}^{\theta,E^*}(q,w)G_{\Lambda}^{\theta,E^*}(w',y).$$
				Since $\operatorname{dist}_1(q,S_0^\Lambda)\geq l_1^\frac{2}{3}>l_1^\frac{1}{2}$, it follows that $U(q)\cap S_0(\theta,E)=\emptyset.$ By Proposition \ref{0g} and $\|q-w\|_1\geq l_1^\frac{1}{2}$ for $w\in \partial_\Lambda^- U(q)$, we have 
				\begin{align}\label{fenzi}
					\nonumber	|G_\Lambda^{\theta,E^*}(q,y)|&\leq\chi_{U(q)}(y)|G_{U(q)}^{\theta,E^*}(q,y)|+\varepsilon\sum_{(w,w')\in \partial_\Lambda U(q)} |G_{U(q)}^{\theta,E^*}(q,w)||G_{\Lambda}^{\theta,E^*}(w',y)|\\
					\nonumber	&\leq\chi_{U(q)}(y)|G_{U(q)}^{\theta,E^*}(q,y)|+2dl_1^d\varepsilon\sup_{(w,w')\in \partial_\Lambda U(q)} |G_{U(q)}^{\theta,E^*}(q,w)||G_{\Lambda}^{\theta,E^*}(w',y)|\\
					\nonumber	&\leq\chi_{U(q)}(y)|G_{U(q)}^{\theta,E^*}(q,y)|+l_1^d\sup_{(w,w')\in \partial_\Lambda U(q)} e^{-\gamma_0\|q-w\|_1}|G_{\Lambda}^{\theta,E^*}(w',y)|\\ 
					\nonumber		&\leq\chi_{U(q)}(y)|G_{U(q)}^{\theta,E^*}(q,y)|+\sup_{(w,w')\in \partial_\Lambda U(q)} e^{-\gamma_0(\|q-w\|_1-d\ln l_1)}|G_{\Lambda}^{\theta,E^*}(w',y)|\\ 
					&\leq\chi_{U(q)}(y)|G_{U(q)}^{\theta,E^*}(q,y)|+\sup_{w'\in \partial_\Lambda^+ U(q)} e^{-\gamma_0(1-l_1^{-\frac13})\|q-w'\|_1}|G_{\Lambda}^{\theta,E^*}(w',y)|.
				\end{align}
				For $y\in U(q)$, since by Proposition \ref{0g}, 
				$$|G_{U(q)}^{\theta,E^*}(q,y)|\leq  e^{-\gamma_0\|q-y\|_1},\ \ \|q-y\|_1\geq 1,$$ 
				$$|G_{U(q)}^{\theta,E^*}(q,y)|\leq\|G_{U(q)}^{\theta,E^*}\| \leq 10\delta_0^{-1} ,\ \ \|q-y\|_1< 1,$$
				we have 
				\begin{equation}\label{f1}
					|G_{U(q)}^{\theta,E^*}(q,y)|\leq10\delta_0^{-1} e^{-\gamma_0(\|q-y\|_1-1)}.
				\end{equation}
				\item If $\operatorname{dist}_1(q,S_0^\Lambda)< l_1^\frac{2}{3}$, then $\|q-p\|_1\leq l_1^{\frac{2}{3}}$ and $q\in B_1(p)$ for some $p\in S_0^\Lambda$. We define 
				\begin{equation}\label{sf2}
					U(q)=\{B_1(p):\ q\in B_1(p)\text{ with }p\in S_0^\Lambda\}.
				\end{equation}
				By the resolvent identity,  one has 
				$$G_\Lambda^{\theta,E^*}(q,y)=\chi_{U(q)}(y)G_{U(q)}^{\theta,E^*}(q,y)-\varepsilon\sum_{(w,w')\in \partial_\Lambda U(q)} G_{U(q)}^{\theta,E^*}(q,w)G_{\Lambda}^{\theta,E^*}(w',y).$$
				Since  $w\in \partial_\Lambda^- U(q)=\partial_\Lambda^- B_1(p)$ and $\|q-p\|_1\leq l_1^{\frac{2}{3}}$, it follows that 
				\begin{equation}\label{yya}
					\|q-w\|_1\geq \|w-p\|_1-\|p-q\|_1\geq l_1-l_1^{\frac{2}{3}}\geq l_1^{\frac45}. 
				\end{equation}
				Since $p\notin S_1(\theta,E)$, by Lemma \ref{1727} and \eqref{yya}, one can get a similar estimate as \eqref{fenzi}:
				\begin{align}\label{fenzi2}
					\nonumber|G_\Lambda^{\theta,E^*}(q,y)|&\leq\chi_{U(q)}(y)|G_{U(q)}^{\theta,E^*}(q,y)|+\varepsilon\sum_{(w,w')\in \partial_\Lambda U(q)} |G_{U(q)}^{\theta,E^*}(q,w)||G_{\Lambda}^{\theta,E^*}(w',y)|\\
					\nonumber	&\leq\chi_{U(q)}(y)|G_{U(q)}^{\theta,E^*}(q,y)|+2dl_1^d\varepsilon\sup_{(w,w')\in \partial_\Lambda U(q)} |G_{U(q)}^{\theta,E^*}(q,w)||G_{\Lambda}^{\theta,E^*}(w',y)|\\
					\nonumber	&\leq\chi_{U(q)}(y)|G_{U(q)}^{\theta,E^*}(q,y)|+l_1^d\sup_{(w,w')\in \partial_\Lambda U(q)} e^{-\gamma_0(1-l_1^{-\frac{1}{50}})\|q-w\|_1}|G_{\Lambda}^{\theta,E^*}(w',y)|\\ 
					&\leq\chi_{U(q)}(y)|G_{U(q)}^{\theta,E^*}(q,y)|+\sup_{w'\in \partial_\Lambda^+ U(q)} e^{-\gamma_0(1-l_1^{-\frac{1}{60}})\|q-w'\|_1}|G_{\Lambda}^{\theta,E^*}(w',y)|.
				\end{align}
				For $y\in U(q)$, since by Lemma \ref{1727},
				$$|G_{U(q)}^{\theta,E^*}(q,y)|\leq  e^{-\gamma_0(1-l_1^{-\frac{1}{50}})\|q-y\|_1},\ \ \|q-y\|_1\geq l_1^\frac{4}{5},$$ 
				$$|G_{U(q)}^{\theta,E^*}(q,y)|\leq\|G_{U(q)}^{\theta,E^*}\| \leq 10\delta_1^{-1} ,\ \ \|q-y\|_1< l_1^\frac{4}{5},$$
				we have 
				\begin{equation}\label{f2}
					|G_{U(q)}^{\theta,E^*}(q,y)|\leq10\delta_1^{-1} e^{-\gamma_0(1-l_1^{-\frac{1}{50}})(\|q-y\|_1-l_1^\frac{4}{5})}.
				\end{equation}
			\end{itemize}
			Let $q_0=x$. Applying the estimates \eqref{fenzi} and \eqref{fenzi2} $K$ times yields 
			\begin{align}\label{1217}
				|G_\Lambda^{\theta,E^*}(x,y)|\leq \sum_{k=0}^{K-1}\chi_{U(q_k)}(y)|G_{U(q_k)}^{\theta,E^*}(q_k,y)|+
				\prod_{k=0}^{K-1}e^{-\gamma_0(1-l_1^{-\frac{1}{60}})\|q_k-q_{k+1}\|_1}|G_{\Lambda}^{\theta,E^*}(q_K,y)|,
			\end{align}
			where $q_{k+1} \in \partial_\Lambda^+U(q_k)$ takes the supremum in the last line of \eqref{fenzi} or \eqref{fenzi2}. 
			We choose $K$ such that $y\notin U(q_k)$ for $k\leq K-1$ and $y\in U(q_K)$, thus the first summation in \eqref{1217} vanishes. Hence, 
			\begin{align}\label{2221}
				|G_\Lambda^{\theta,E^*}(x,y)|\leq
				\prod_{k=0}^{K-1}e^{-\gamma_0(1-l_1^{-\frac{1}{60}})\|q_k-q_{k+1}\|_1}|G_{\Lambda}^{\theta,E^*}(q_K,y)|.
			\end{align}
			By \eqref{fenzi} and    \eqref{fenzi2}, we have 
			\begin{align}\label{2204}
				|G_{\Lambda}^{\theta,E^*}(q_K,y)|\leq |G_{U(q_K)}^{\theta,E^*}(q_K,y)|+\sup_{w'\in \partial_\Lambda^+ U(q_K)} e^{-\gamma_0(1-l_1^{-\frac{1}{60}})\|q_K-w'\|_1}|G_{\Lambda}^{\theta,E^*}(w',y)|.
			\end{align}
			Since
			\begin{align*}
				e^{-\gamma_0\|q_K-y\|}_1&\leq e^{-\gamma_0(1-l_1^{-\frac{1}{50}})\|q_K-y\|_1},\\
				e^{\gamma_0}&\leq e^{\frac{1}{2}\gamma_0l_1^\frac{4}{5}} \leq e^{\gamma_0(1-l_1^{-\frac{1}{50}})l_1^\frac{4}{5}}
			\end{align*}
			it follows from \eqref{f1} and \eqref{f2} that the first term on the right-hand side of \eqref{2204} satisfies
			\begin{equation}\label{25}
				|G_{U(q_K)}^{\theta,E^*}(q_K,y)|\leq10\delta_1^{-1} e^{-\gamma_0(1-l_1^{-\frac{1}{50}})(\|q_K-y\|_1-l_1^\frac{4}{5})}.
			\end{equation}
			For $y\in U(q_K)$, $ w'\in \partial_\Lambda^+ U(q_K)$, we have:
			\begin{itemize}
				\item If $U(q_K)$ is in the  form of \eqref{sf1}, then $$\|q_K-w'\|_1\geq \|q_K-y\|_1.$$
				\item If $U(q_K)$ is in the  form of \eqref{sf2} with $U(q_K)=B_1(p)$, then  
				\begin{align*}
					\|q_K-w'\|_1&\geq \|p-w'\|_1-\|p-q_K\|_1\\&\geq \|p-y\|_1-\|p-q_K\|_1\\&\geq \|q_K-y\|_1-\|p-q_K\|-\|p-q_K\|_1\\
					&\geq \|q_K-y\|_1-2l_1^\frac{2}{3} .
				\end{align*}
			\end{itemize}   
			It follows that the second term on the right-hand side of \eqref{2204} satisfies
			\begin{align}\label{2223}
				\nonumber&	\sup_{w'\in \partial_\Lambda^+ U(q_K)} e^{-\gamma_0(1-l_1^{-\frac{1}{60}})\|q_K-w'\|_1}|G_{\Lambda}^{\theta,E^*}(w',y)|\\ 
				\nonumber\leq &	\sup_{w'\in \partial_\Lambda^+ U(q_K)} e^{-\gamma_0(1-l_1^{-\frac{1}{60}})(\|q_K-y\|_1-2l_1^\frac{2}{3})}|G_{\Lambda}^{\theta,E^*}(w',y)|\\
				\nonumber \leq & e^{-\gamma_0(1-l_1^{-\frac{1}{60}})(\|q_K-y\|_1-2l_1^\frac{2}{3})}\|G_{\Lambda}^{\theta,E^*}\|\\
				\leq& 10\delta_1^{-1}e^{-\gamma_0(1-l_1^{-\frac{1}{60}})(\|q_K-y\|_1-2l_1^\frac23)}.
			\end{align}
			By \eqref{2221}, \eqref{2204}, \eqref{25}, \eqref{2223} and $\delta_1^{-1}=e^{l_1^\frac{2}{3}}$, provided $\|x-y\|\geq l_1^\frac56 $, we have 
			\begin{align*}
				|G_\Lambda^{\theta,E^*}(x,y)|&\leq
				20\delta_1^{-1}	e^{-\gamma_0(1-l_1^{-\frac{1}{60}})(\sum\limits_{k=0}^{K-1}\|q_k-q_{k+1}\|_1+\|q_K-y\|_1-3l_1^\frac45)}\\
				&\leq e^{-\gamma_0(1-l_1^{-\frac{1}{60}})(\|x-y\|_1-3l_1^\frac{4}{5}-2l_1^\frac{2}{3})}\\
				&\leq e^{-\gamma_0(1-l_1^{-\frac{1}{60}})(1-(3l_1^\frac{4}{5}+2l_1^\frac{2}{3})l_1^{-\frac{5}{6}})\|x-y\|_1}\\
				&\leq e^{-\gamma_0(1-l_1^{-\frac{1}{80}})\|x-y\|_1}.
			\end{align*}
		\end{proof}
		\begin{rem}\label{lh}
			Replacing $\theta$ by $\theta-0$, we can estimate the Green's function $G_{\Lambda} ^{\theta-0,E^*}$ of the left limit operator  $H_\Lambda(\theta-0)$       for  $1$-good set $\Lambda$ related to $(\theta-0,E)$ with   $|E-E^*|<\delta_1/5$, by analogous definitions\begin{align*}
				S_0(\theta-0,E)&:=\left\{p\in \Z^d:\ |v (\theta+p\cdot \omega-0)-E|<\delta_0\right\},\\
				S_1(\theta-0,E)&:=\left\{p\in \Z^d:\ |E_1 (\theta+p\cdot \omega-0)-E|<\delta_1\right\}.
			\end{align*}

		\end{rem}
		\subsection{Inductive hypotheses}
		In this section, we state the inductive hypotheses of each scale. First, we introduce some inductive definitions.
		\begin{defn}
			Recalling that $l_1=[|\ln\delta_0|^4],\delta_0=\varepsilon_0^{\frac{1}{20}},\gamma_0=\frac{1}{2}|\ln\varepsilon|$ defined in the previous section, for $n\geq 1$, we define the inductive parameters:
			\begin{itemize}
				\item $l_n:=(l_1)^{2^{n-1}}$  \hfill	(the $n$-scale length)
				\item $\delta_n:=e^{-l_n^\frac{2}{3}}$\hfill (the $n$-scale resonance parameter)
				\item $\gamma_n:=\prod\limits_{k=1}^{n}(1-l_k^{-\frac{1}{80}})\gamma_0$ \hfill (the Green's function decay rate)
			\end{itemize}    
			We let $\varepsilon$ be sufficiently small so that  \begin{equation}\label{rate}
				\gamma_n\geq \gamma_\infty:=\prod\limits_{k=1}^{\infty}(1-l_k^{-\frac{1}{80}})\gamma_0\geq \frac{1}{2}\gamma_0\geq 10.
			\end{equation}
		\end{defn}
		
		\begin{flushleft}
			\textbf{Inductive hypotheses (scale $n$).} The following inductive  hypotheses hold true for $0\leq m\leq n $: 
		\end{flushleft}
		\begin{hp}\label{h1}
			There exists a block $B_m$ independent of $\theta$ satisfying  \begin{equation}\label{h11}
				Q_{l_m}\subset B_m\subset Q_{l_m+50l_{m-1}}.
			\end{equation}
			The Dirichlet restriction  $H_{B_m}(\theta)$ has a $1$-periodic, single-valued, real-valued  Rellich function $E_m(\theta)$ satisfying the following properties{\rm :}
			\begin{enumerate}
				\item For $m=0$, $B_0=\{o\}$  and $E_0(\theta)=v(\theta)$.
				\item For $m\geq1$, $E_m(\theta)$ is well-approximated by the previous generation Rellich function $E_{m-1}(\theta)${\rm :}  
				\begin{equation}\label{h12}
					|E_m(\theta)- E_{m-1}(\theta)|\leq e^{-l_{m-1}},
				\end{equation} where $e^{-l_0}$ is understood as $\varepsilon$. 
				\item $E_m(\theta)$ is the unique eigenvalue of $H_{B_m}(\theta)$ in $D(E_{m-1}(\theta),10\delta_{m-1})$.
				\item Denote by  $0=\alpha_1<\alpha_2<\cdots<\alpha_{|B_m|}<\alpha_{|B_m|+1}=1$ the discontinuities of $H_{B_m}(\theta)$, where 
				$$\{\alpha_1,\alpha_2,\cdots,\alpha_{|B_m|}\}=\left\{\{-x\cdot\omega\}\right\}_{x\in B_m}.$$
				Then $E_m(\theta)$ is continuous on each interval $[\alpha_k,\alpha_{k+1})$  with  the local Lipschitz monotonicity property{\rm :}  \begin{equation}\label{1639}
					E_m(\theta_2)-E_m(\theta_1)\geq L(\theta_2-\theta_1), \ \ \alpha_k\leq \theta_1\leq \theta_2<\alpha_{k+1},
				\end{equation}
				and moreover, $E_m(\theta)$ has  non-negative jumps at $\alpha_k$  {\rm ($2\leq k\leq |B_m|$):}
				\begin{equation}\label{1640}
					E_m(\alpha_k)\geq E_m(\alpha_k-0).
				\end{equation}
			\end{enumerate}
		\end{hp}
		\begin{rem}
			By \eqref{1639} and \eqref{1640}, $E_m(\theta)$ in fact satisfies the uniform Lipschitz monotonicity property{\rm :} 
			\begin{equation}\label{Lc}
				E_m(\theta_2)-E_m(\theta_1)\geq L(\theta_2-\theta_1), \ \ 0\leq \theta_1 \leq \theta_2<1.
			\end{equation}
		\end{rem}
		
		\begin{defn}\label{A}
			To simplify the notation, for an $m$-generation  Rellich function $E_m(\theta)$ described in Hypothesis  \ref{h1},  $\theta\in \R\setminus\Z$ and  $R>0$, we denote the region 
			\begin{equation}\label{AR}
				A_m(\theta,R):=\bigcup_{a\in [E_m(\theta-0),E_m(\theta)]}D(a,R).  
			\end{equation}
			For the case $E_m(\theta-0)=E_m(\theta)$, the above  $[E_m(\theta-0),E_m(\theta)]$ is understood as the single point  $E_m(\theta)$ and thus $A_m(\theta,R)=D(E_m(\theta),R)$.\\
			
			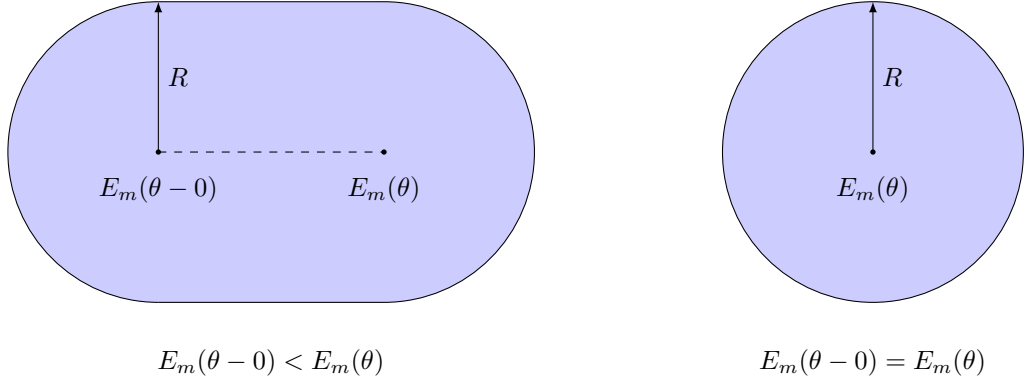
\begin{figure}[htp]
				\centering

				\begin{tikzpicture}[>=latex]
					\fill[blue!20!white] (2,0) circle (2);	\fill[blue!20!white] (5,0) circle (2); 	\fill[blue!20!white] (2,2) rectangle (5,-2); 
					\draw (2,2) arc (90:270:2); 
					\draw (5,2) arc (90:-90:2);
					\draw(2,2)-- (5,2);		 	\draw(2,-2)-- (5,-2);
					\draw (2,-0.2)node[below]{$E_m(\theta-0)$};  	\draw (5,-0.2)node[below]{$E_m(\theta)$};	\fill(2,0)circle(1pt);	\fill(5,0)circle(1pt);
					\draw[->](2,0)--(2,2);  	\draw (2,1)node[right]{$R$};
					\draw[dashed](2,0)--(5,0);  
					\fill[blue!20!white] (11.5,0) circle (2);	\draw (11.5,-0.2)node[below]{$E_m(\theta)$};\fill(11.5,0)circle(1pt); \draw (13.5,0) arc (0:360:2);	\draw (11.5,1)node[right]{$R$};	\draw[->](11.5,0)--(11.5,2);  \draw(3.5,-2.5)node[below]{$E_m(\theta-0)<E_m(\theta)$}; \draw(11.5,-2.5)node[below]{$E_m(\theta-0)=E_m(\theta)$};
				\end{tikzpicture}
				\caption{A picture for $A_m(\theta,R)$}
				
			\end{figure}
		\end{defn}
		\begin{hp}\label{h2}
			Denote  $B_m^o:=B_m\setminus\{o\}$. 
			Write	$$	z\operatorname{Id}_{B_m}-H_{B_m}(\theta)=\begin{pmatrix}
				z-v(\theta) & c^{\text{T}} \\
				c & z\operatorname{Id}_{B_m^o}-H_{B_m^o}(\theta)
			\end{pmatrix}.$$
			If  the resolvent $(z\operatorname{Id}_{B_m^o}-H_{B_m^o}(\theta))^{-1}$ exists, we 	denote the Schur complement by $$s_{\theta,m}(z):=z-v(\theta) -c^{\text{T}}\left(z\operatorname{Id}_{B_m^o}-H_{B_m^o}(\theta)\right)^{-1}c.$$
			\begin{itemize}
				\item  For any $\theta\in \R\setminus\Z$, $z\in A_{m-1}(\theta,10\delta_{m-1})$ and  $\xi\in \{\theta,\theta-0\}$, the resolvent $(z\operatorname{Id}_{B_m^o}-H_{B_m^o}(\xi))^{-1}$ exists and  we have approximations for the Schur complement{\rm :}
				\begin{align}\label{sa1}
					s_{\xi,m}(z)&=(z-E_m(\xi))\left (1+O(\sum_{k=0}^{m-1}\delta_k) \right),& z\in A_{m-1}(\theta,10\delta_{m-1}),\\
					\label{sa2}
					s_{\xi,m}'(z)&=1+O(\sum_{k=0}^{m-1}\delta_k), & z\in  A_{m-1}(\theta,5\delta_{m-1}).
				\end{align}
				\item For any $z\in D(E_{m-1}(0),10\delta_{m-1})$,   the resolvent $(z\operatorname{Id}_{B_m^o}-H_{B_m^o}(0))^{-1}$ exists  and 
				\begin{align}\label{sa10}
					s_{0,m}(z)&=(z-E_m(0))\left (1+O(\sum_{k=0}^{m-1}\delta_k) \right),&  z\in D(E_{m-1}(0),10\delta_{m-1}),\\
					\label{sa20}
					s_{0,m}'(z)&=1+O(\sum_{k=0}^{m-1}\delta_k), & z\in  D(E_{m-1}(0),5\delta_{m-1}).
				\end{align}
				\item For any $z\in D(E_{m-1}(1-0),10\delta_{m-1})$, the resolvent $(z\operatorname{Id}_{B_m^o}-H_{B_m^o}(1-0))^{-1}$ exists  and 
				\begin{align*}
					s_{1-0,m}(z)&=(z-E_m(1-0))\left (1+O(\sum_{k=0}^{m-1}\delta_k) \right), &  z\in D(E_{m-1}(1-0),10\delta_{m-1}),\\
					s_{1-0,m}'(z)&=1+O(\sum_{k=0}^{m-1}\delta_k),& z\in  D(E_{m-1}(1-0),5\delta_{m-1}).
				\end{align*}
			\end{itemize}
			The above  $``O(\sum_{k=0}^{m-1}\delta_k)"$ are understood as analytic functions of $z$ bounded by  $\sum_{k=0}^{m-1}\delta_k$.
		\end{hp}
		\begin{defn}
			For $p\in \Z^d$, we denote  $$B_m(p):=B_m+p.$$  Fix $\theta\in \R$, $E\in \C$.  Define the $m$-scale resonant points set as  \begin{equation}\label{sm}
				S_m(\theta,E):=\left\{p\in \Z^d:\ |E_m (\theta+p\cdot \omega)-E|<\delta_m\right\}.
			\end{equation}
			Let $\Lambda\subset\Z^d$ be a finite set. Related to $(\theta,E)$,  we say
			\begin{itemize}
				\item $\Lambda$ is $m$-nonresonant  if $\Lambda\cap S_m(\theta,E)=\emptyset$.
				\item  $\Lambda$ is $m$-regular if $p\in \Lambda\cap S_k(\theta,E)$ $\Rightarrow$  $B_{k+1}(p)\subset\Lambda$ ($0\leq k\leq m-1$).
				\item $\Lambda$ is $m$-good if it is both $m$-nonresonant and $m$-regular.
			\end{itemize}
		\end{defn}
		\begin{hp}\label{h3}
			Fix $\theta\in \R$, $E\in \C$. Assume that  a finite set $\Lambda$ is $m$-good related to $(\theta,E)$,  then for any $E^*\in C$ such that $|E-E^*|<\delta_m/5$, we have 
			\begin{align*}
				\|G_{\Lambda} ^{\theta,E^*}\|&\leq10\delta_m^{-1}, \\
				|G_{\Lambda} ^{\theta,E^*}(x,y)|&\leq e^{-\gamma_m\|x-y\|_1}, \ \|x-y\|_1\geq l_m^{\frac{5}{6}}.
			\end{align*}
			Moreover, as mentioned in Remark \ref{lh},	analogous estimates of  $G_{\Lambda} ^{\theta-0,E^*}$ hold true  for $m$-good set  $\Lambda $ related to $(\theta-0,E)$ with  $|E-E^*|<\delta_m/5$. 
		\end{hp}
		
		\begin{hp}\label{h4}$B_m$ satisfies  the following properties{\rm :}
			\begin{itemize}
				\item For any $E\in D(E_{m-1}(0),10\delta_{m-1})$, the set  $B_m$ is 
				$(m-1)$-regular related to $(0,E)$.
				\item  For any $E\in D(E_{m-1}(1-0),10\delta_{m-1})$, the set  $B_m$ is 
				$(m-1)$-regular related to $(1-0,E)$.
				\item For any $\theta\in\R\setminus\Z$, $E\in A_{m-1}(\theta,10\delta_{m-1})$, the set  $B_m$ is 
				$(m-1)$-regular related to $(\theta,E)$ and $(\theta-0,E)$. 
			\end{itemize}
		\end{hp}

		\subsection{Verification of the hypotheses for $n=1$}
		
		By the analysis in section \ref{n=1}, we can prove
		\begin{prop}\label{1832}There exists $\varepsilon_0=\varepsilon_0(L,d,\tau,\gamma)>0$ such that for all $0\leq \varepsilon\leq \varepsilon_0$, the inductive hypotheses hold true at the scale $n=1$.
		\end{prop}
		\begin{proof}
			We check the hypotheses one by one. \begin{itemize}
				\item   The $1$-periodic, single-valued, real-valued  properties of $E_1(\theta)$ follows  from the construction (c.f. Remark \ref{454}).  Item (2) and  (3)  follow from Proposition \ref{715}.  
				\eqref{1639} and \eqref{1640} follow from  Proposition \ref{814}.
				\item
				Lemma \ref{sl} and \ref{H-} yield Hypothesis \ref{h2}.
				\item The good Green's function estimates for $0$-nonresonant and $1$-good sets are verified in  Proposition \ref{0g} and \ref{1g}.
				\item There is nothing to check with  Hypothesis \ref{h4} for $n=1$ since we make a convention  that any set is 
				$0$-regular.
			\end{itemize}
		\end{proof}
		\subsection{Verification of inductive hypotheses for the general inductive scale} In this section, we will prove the following theorem:
		\begin{thm}\label{key2}
			There exists $\varepsilon_0=\varepsilon_0(L,d,\tau,\gamma)>0$ such that for all $0\leq \varepsilon\leq \varepsilon_0$ and $n\geq 1$, if the  inductive hypotheses hold true at the scale $n$, then they  remain valid  at the  scale $n+1$.
		\end{thm}
	\begin{rem}
		Here the choice of $\varepsilon_0$ can be traced from  the first inductive step, which is dictated by finitely many conditions.
	\end{rem}
		The proof of Theorem \ref{key2} is divided into the following  several parts. First,  we list  two lemmas  that will be used  later. We omit the proofs since they are obvious. 
		\begin{lem}\label{con}
			For $0\leq k\leq m\leq n$, with the notation from Definition \ref{A},  we have the inclusion relations{\rm :}
			\begin{itemize}
				\item for  $\theta\in \R,$
				$$D(E_m(\theta),10\delta_m)\subset D(E_k(\theta),10\delta_k),$$  $$D(E_m(\theta-0),10\delta_m)\subset D(E_k(\theta-0),10\delta_k).$$
				\item  for $\theta\in \R\setminus\T,$
				$$A_m(\theta,10\delta_m)\subset A_k(\theta,10\delta_k). $$  \end{itemize}
			Moreover, they remain valid if all the $``10"$ above replaced by $``5"$.
		\end{lem}
		\begin{lem}\label{no}
			For $0\leq k\leq m\leq n$, with the notation from Definition \ref{A}, for any $\theta\in \R\setminus\Z$ and $\tilde\theta\in \R$, we have 
			\begin{align*}
				\min(|E_m(\tilde\theta)-E_m(\theta)|,|E_m(\tilde\theta)-E_m(\theta-0)|)& \leq  \operatorname{dist}(E_m(\tilde\theta),A_m(\theta,R))+R,\\
				\min(|E_m(\tilde\theta-0)-E_m(\theta)|,|E_m(\tilde\theta-0)-E_m(\theta-0)|)& \leq  \operatorname{dist}(E_m(\tilde\theta-0),A_m(\theta,R))+R.
			\end{align*}
		\end{lem}
		
		\subsubsection{Construction of the block $B_{n+1}$}
		In this section, we construct the  block $B_{n+1}$ satisfying hypothesis \eqref{h11} and Hypothesis \ref{h4}. 
		\begin{lem}\label{1222}
			Let $0\leq m\leq n$.  
			Define the sets   $$X_m^{(1)}=\{ x\in \Z^d : \text{ there exists  $E\in D(E_{n}(0),10\delta_{n})$  with $x\in S_m(0,E)$}\},$$  
			$$X_m^{(2)}=\{ x\in \Z^d : \text{ there exists  $E\in D(E_{n}(1-0),10\delta_{n})$  with $x\in S_m(1-0,E)$}\},$$
			$$X_m^{(3)}=\{ x\in \Z^d : \text{ there exist    $\theta\in \R\setminus\Z$, $E\in A_{n}(\theta,10\delta_{n})$  with $x\in S_m(\theta,E)\cup S_m(\theta-0,E)$}\}$$
			
			and $$X_m=X_m^{(1)}\cup X_m^{(2)}\cup X_m^{(3)}.$$  Then we have 
			$$\inf_{\substack{x,y\in X_m\\x\neq y}}\|x-y\|_1\geq 100l_{m+1}.$$ 
		\end{lem}
		\begin{proof}  	 We claim that $	L\|x\cdot \omega \|_\T\leq 11\delta_m$ for any $x\in X_m$.
					 Assume $x\in X_m^{(3)}$. (The other two cases are analogous.) Then  by the definition of $S_m(\theta,E)$, $S_m(\theta-0,E)$ and  $ X_m^{(3)}$, there exist $\theta\in \R\setminus\Z$, $\xi_1\in \{\theta,\theta-0\}$,  $E\in A_{n}(\theta,10\delta_{n})\subset  A_{m}(\theta,10\delta_{m})$   such that 
				$$|E_m(\xi_1+x\cdot \omega)-E|<\delta_m  .$$ It follows from Lemma \ref{no} that 
				\begin{align*}
					&	\min \left(|E_m(\theta)-E_m(\xi_1+x\cdot \omega)|,|E_m(\theta-0)-E_m(\xi_1+x\cdot \omega)|\right)\\
					\leq& \operatorname{dist}\left(A_{m}(\theta,10\delta_{m}),E_m(\xi_1+x\cdot \omega)\right)+10\delta_m\\
					\leq &|E_m(\xi_1+x\cdot \omega)-E|+10\delta_m\\\leq& 11\delta_m .
				\end{align*}
				Hence there exists  some  $\xi_2\in \{\theta,\theta-0\},$   such that 
				$$|E_m(\xi_1+x\cdot \omega)-E_m(\xi_2)|\leq 11\delta_m.$$
				Thus by the Lipschitz monotonicity property of $E_m(\theta)$, we have 
				\begin{equation*}
					L\|x\cdot \omega \|_\T	\leq |E_m(\xi_1+x\cdot \omega)-E_m(\xi_2)|\leq 11\delta_m.
				\end{equation*}
	 Let $y\in  X_m$ with $x\neq y$. Similarly, we have  $	L\|y\cdot \omega \|_\T\leq 11\delta_m.$
			By  \eqref{DC} and $x\neq y$, it follows that 
			$$L\gamma \|x-y\|_1^{-\tau}\leq  L\|(x-y)\cdot \omega \|_\T\leq L\|x\cdot \omega \|_\T+L\|y\cdot \omega \|_\T\leq 22\delta_m.$$
			Since $\delta_m= e^{-l_m^\frac{2}{3}}$, it follows from  the above inequality that   
			$$100l_{m+1}=100l_{m}^2\leq e^{\frac{1}{\tau}l_m^\frac{2}{3}}(\frac{L\gamma}{22})^{\frac{1}{\tau}}=(\frac{L\gamma}{22\delta_m})^{\frac{1}{\tau}}\leq \|x-y\|_1$$
			provided $\varepsilon_0$ is sufficiently small.
		\end{proof}
		\begin{lem}\label{1227}
			Let $L\geq 1$ and  $X\subset\Z^d$ satisfy the  condition{\rm :} 
			\begin{equation}\label{844}
				\inf_{\substack{x,y\in X\\x\neq y}}\|x-y\|_1\geq 10L.
			\end{equation}
			Then for any set $B\subset \Z^d$, there exists another set $\tilde{B}$  satisfying the following two  properties{\rm :} 
			\begin{itemize}
				\item $B\subset \tilde{B}\subset \{x\in \Z^d: \ \operatorname{dist}_1(x,B)\leq 2 L\}$.
				\item   $(Q_L+x)\subset \tilde{B}$ for any $x\in X$ such that  $(Q_L+x)\cap \tilde{B}\neq \emptyset$.
			\end{itemize}
		\end{lem}
		\begin{proof}
			We define 
			\begin{equation}\label{839}
				\tilde{B}:=B\bigcup\left(\bigcup_{x\in X:(Q_L+x)\cap B\neq \emptyset }(Q_L+x)\right)
			\end{equation}
			and claim that $\tilde{B}$ satisfies the desired properties. For any $y\in \tilde{B}$, it follows that  $y\in B$ or $ y\in (Q_L+x)$ for some $x\in X$ with  $(Q_L+x)\cap B\neq \emptyset$. Hence, 
			$$\operatorname{dist}_1(y,B)\leq \|y-x\|_1+\operatorname{dist}_1(x,B)\leq L+L\leq 2L.$$ Thus we finish the proof of the first property. For the second,
			assume  there exists $x\in X$ such that $(Q_L+x)\cap \tilde{B}\neq \emptyset$. Then  by \eqref{839}, if \begin{equation}\label{849}
				(Q_L+x)\cap B= \emptyset,
			\end{equation} then there exists  $y\in X$ with \begin{equation}\label{851}
				(Q_L+y)\cap B\neq \emptyset
			\end{equation} such that  $(Q_L+x)\cap (Q_L+y)\neq \emptyset$. Thus $\|x-y\|_1\leq 2L$. By \eqref{844}, it follows that $x=y$, which contradicts with \eqref{849} and \eqref{851}. Thus \eqref{849} does not hold true. Then by \eqref{839}, it follows that  $(Q_L+x)\subset \tilde{B}$.
		\end{proof}
		\begin{prop}\label{m+1r}
			There exists a block $B_{n+1}$ satisfying the following  properties{\rm :}
			\begin{itemize}
				\item $	Q_{l_{n+1}}\subset B_{n+1}\subset Q_{l_{n+1}+50l_{n}}$.
				%
				\item    $x\in B_{n+1}\cap S_k(0,E)$  for some    $E\in D(E_n(0) ,10\delta_n)$ $\Rightarrow$ $B_{k+1}(x)\subset B_{n+1}$  {\rm ($0\leq k\leq n-1$)}.
				\item    $x\in B_{n+1}\cap S_k(1-0,E)$  for some    $E\in D(E_n(1-0) ,10\delta_n)$ $\Rightarrow$ $B_{k+1}(x)\subset B_{n+1}$  {\rm ($0\leq k\leq n-1$)}.
				\item    $x\in B_{n+1}\cap (S_k(\theta,E)\cup S_k(\theta-0,E) )$  for some $\theta\in \R\setminus\Z$,   $E\in A_n(\theta ,10\delta_n)$ $\Rightarrow$ $B_{k+1}(x)\subset B_{n+1}$  {\rm ($0\leq k\leq n-1$)}.
			\end{itemize}
		\end{prop}
		\begin{proof}
			Define the sets   $$X_k^{(1)}=\{ x\in \Z^d : \text{ there exists  $E\in D(E_{n}(0),10\delta_{n})$  with $x\in S_k(0,E)$}\},$$  
			$$X_k^{(2)}=\{ x\in \Z^d : \text{ there exists  $E\in D(E_{n}(1-0),10\delta_{n})$  with $x\in S_k(1-0,E)$}\},$$
			$$X_k^{(3)}=\{ x\in \Z^d : \text{ there exist    $\theta\in \R\setminus\Z$, $E\in A_{n}(\theta,10\delta_{n})$  with $x\in S_k(\theta,E)\cup S_k(\theta-0,E)$}\}$$
			
			and $$X_k=X_k^{(1)}\cup X_k^{(2)}\cup X_k^{(3)}.$$
			By  Lemma \ref{1222}, we have 
			$$\inf_{\substack{x,y\in X_k\\x\neq y}}\|x-y\|_1\geq 100l_{k+1}.$$
			We start with $B^{(0)}:=Q_{l_{n+1}}$.	Setting $L=10l_{n}$, $X=X_{n-1}$ in  Lemma \ref{1227} yields  that there exists $B^{(1)}$ satisfying 
			\begin{itemize}
				\item  $B^{(0)}\subset B^{(1)}\subset \{x\in \Z^d: \ \operatorname{dist}_1(x,B^{(0)})\leq 20l_n\} .$
				\item   $(Q_{10l_n}+x)\subset B^{(1)}$ for any $x\in X_{n-1}$ such that  $(Q_{10l_n}+x)\cap B^{(1)}\neq \emptyset$.
			\end{itemize}
			By setting $L=10l_{n-k}$, $X=X_{n-k-1}$ in Lemma \ref{1227},  we inductively	define  $B^{(k+1)}$ ($1\leq k\leq n-1$), such that
			\begin{equation}\label{1321}
				B^{(k)}\subset B^{(k+1)}\subset \{x\in \Z^d: \ \operatorname{dist}_1(x,B^{(k)})\leq 20l_{n-k}\},
			\end{equation}
			\begin{equation}\label{1916}
				\text{ 	 $(Q_{10l_{n-k}}+x)\subset B^{(k+1)}$ for $x\in X_{n-k-1}$ with $(Q_{10l_{n-k}}+x)\cap B^{(k+1)}\neq \emptyset$.}
			\end{equation}
			We claim that $B^{(n)}$ is the desired block. Since $$\sum_{k=m}^{n-1}20l_{n-k}\leq 30l_{n-m},$$ it follows from \eqref{1321}  that 
			\begin{equation}\label{1316}
				B^{(n)}\subset \{x\in \Z^d :\ \operatorname{dist}_1(x,B^{(m)})\leq 30l_{n-m}\}.
			\end{equation}
			Thus  $$B^{(n)}\subset \{x\in \Z^d :\ \operatorname{dist}_1(x,B^{(0)})\leq  30l_n\}\subset  Q_{l_{n+1}+50l_{n}}.$$
			Let $x$ be such that   \begin{equation}\label{1429}
				x\in B^{(n)}\cap X_k
			\end{equation} for some $0\leq k\leq n-1$.  By \eqref{1316} with $m=n-k$, we have
			\begin{equation}\label{1858}
				B^{(n)}\subset \{x\in \Z^d :\ \operatorname{dist}_1(x,B^{(n-k)})\leq 30l_{k}\}.
			\end{equation}
			By \eqref{1429}, \eqref{1858} and $30l_k\leq l_{k+1}$, it follows that 
			$$(Q_{10l_{k+1}}+x)\cap B^{(n-k)}\neq \emptyset.$$
			By \eqref{1916}, it follows that $(Q_{10l_{k+1}}+x)\subset  B^{(n-k)}.$
			Hence, by the inductive hypothesis  \eqref{h11} and $l_{k+1}+50l_k\leq 10l_{k+1}$, it follows that 
			$$B_{k+1}(x)\subset (Q_{10l_{k+1}}+x)\subset  B^{(n-k)}\subset B^{(n)}.$$
			Hence the set $B_{n+1}:=B^{(n)}$,  with $	Q_{l_{n+1}}\subset B_{n+1}\subset Q_{l_{n+1}+50l_{n}}$ satisfies 
			\begin{itemize}
				\item For any $E\in D(E_{n}(0),10\delta_{n})$, the set    $B_{n+1}$ is 
				$n$-regular related to $(0,E)$.
				\item  For any $E\in D(E_{n}(1-0),10\delta_{n})$, the set   $B_{n+1}$ is 
				$n$-regular related to $(1-0,E)$.
				\item For any $\theta\in\R\setminus\Z$, the set  $E\in A_{n}(\theta,10\delta_{n})$,  $B_{n+1}$ is 
				$n$-regular related to $(\theta,E)$ and $(\theta-0,E)$. 
			\end{itemize}
		\end{proof}
		\subsubsection{Construction of $E_{n+1}(\theta)$}
		Having constructed the $(n+1)$-scale  block $B_{n+1}$ in the previous section,  we construct the $(n+1)$-scale Rellich function $E_{n+1}(\theta )$ of the Dirichlet restriction $H_{B_{n+1}}(\theta)$ satisfying $$|E_{n+1}(\theta)-E_{n}(\theta)|\leq e^{-l_n}$$ in this section. To this end, we utilize the Schur complement perturbation argument to show the existence and uniqueness of such $E_{n+1}(\theta)$. 
		\begin{lem}\label{sepn}
			Let $m$ be an integer such that  $0\leq m\leq n$. For any $\theta\in\R$,  $\xi_1,\xi_2\in \{\theta,\theta-0\}$  and $x\in B_{m+1}$ such that $x\neq o$, we have $$|E_{m}(\xi_1+x\cdot \omega )-E_m(\xi_2)|\geq 20\delta_m. $$
		\end{lem}
		\begin{proof}
			Let $x\in B_{m+1}$ with $x\neq o$. Since  $B_{m+1}\subset Q_{l_{m+1}+50l_m}\subset Q_{2l_{m+1}}$, it follows that  $\|x\|_1\leq 2l_{m+1}$. Thus by \eqref{DC}, inductive hypothesis \eqref{Lc} and $l_{m+1}^{-\tau}=l_m^{-2\tau}=|\ln\delta_m|^{-3\tau}$, it follows that 
			$$|E_{m}(\xi_1+x\cdot\omega)-E_m(\xi_2)|\geq L\|x\cdot\omega\|_\T \geq L\gamma( 2l_{m+1})^{-\tau}\geq L\gamma2^{-\tau}|\ln\delta_m|^{-3\tau}\geq 20\delta_m.$$
		\end{proof}
		\begin{lem}\label{ag}
			Let $m$ be an integer of  $0\leq m\leq n$. Denote  the set  $B_{m+1}^o:=B_{m+1}\setminus \{o\}$. We  have the followings{\rm :}
			\begin{itemize}
				\item  For any $E\in D(E_m(0),10\delta_m)$, the set  $B_{m+1}^o$ is $m$-good related to $(0,E)$.
				\item  For any $E\in D(E_m(1-0),10\delta_m)$, the set  $B_{m+1}^o$ is $m$-good related to $(1-0,E)$.
				\item  For any $\theta\in \R\setminus\Z $ and $E\in A_m(\theta ,10\delta_m)$, the set  $B_{m+1}^o$ is $m$-good related to $(\theta,E)$ and $(\theta-0,E)$.
			\end{itemize}
		\end{lem}
		\begin{proof}
			We give the proof of  the third item since the others are analogous. Fix $\theta\in \R\setminus\Z$ and $E\in A_m(\theta ,10\delta_m)$. 	First, we prove that  $B_{m+1}^o$ is $m$-nonresonant related to $(\theta,E)$ and $(\theta-0,E)$. It follows from Lemma \ref{sepn} and Lemma \ref{no}  that  for  $x\in B_{m+1}^o$, $\xi\in\{\theta,\theta-0\}$ we have \begin{align*}
				|E_m(\xi+x\cdot\omega ) -E|&\geq \operatorname{dist}(E_m(\xi+x\cdot\omega ),A_m(\theta,10\delta_m))\\&\geq \min_{\xi'\in\{\theta,\theta-0\}} |E_m(\xi+x\cdot\omega ) -E_m(\xi')|-10\delta_m\\
				&\geq 20\delta_m-10\delta_m\\
				&\geq 10\delta_m.
			\end{align*}
			Thus $B_{m+1}^o\cap( S_m(\theta,E)\cup S_m(\theta-0,E))=\emptyset$.  It remains to prove that  $B_{m+1}^o$ is $m$-regular related to $(\theta,E)$ and  $(\theta-0,E)$. By Hypothesis \ref{h4} and  Proposition \ref{m+1r}, it follows that  $B_{m+1}$ is $m$-regular related to $(\theta,E)$ and $(\theta-0,E)$.  Let $0\leq k\leq m-1$ and $x\in B_{m+1}^o\cap( S_{k}(\theta,E) \cup S_{k}(\theta-0,E))$. Since $B_{m+1}$ is $m$-regular and $x\in B_{m+1}^o\cap( S_{k}(\theta,E) \cup S_{k}(\theta-0,E))\subset  B_{m+1}\cap( S_{k}(\theta,E) \cup S_{k}(\theta-0,E))$, it follows that \begin{equation}\label{1545}
				B_{k+1}(x)\subset B_{m+1}.
			\end{equation}
			By the definition of  $S_{k}(\theta,E)$ and $S_{k}(\theta-0,E)$, there exists some  $\xi_1\in \{\theta,\theta-0\}$ such that $|E_k(\xi_1+x\cdot\omega)-E|<\delta_k.$ Since $E\in A_m(\theta ,10\delta_m)\subset A_k(\theta ,10\delta_k)$ (c.f. Lemma \ref{con}),  it follows from Lemma \ref{no}  that 
			\begin{align*}
				&	\min \left(|E_k(\theta)-E_k(\xi_1+x\cdot \omega)|,|E_k(\theta-0)-E_k(\xi_1+x\cdot \omega)|\right)\\
				\leq& \operatorname{dist}\left(A_{k}(\theta,10\delta_{k}),E_k(\xi_1+x\cdot \omega)\right)+10\delta_k\\
				\leq &|E_k(\xi_1+x\cdot \omega)-E|+10\delta_k\\
				\leq &11\delta_k .
			\end{align*}
			Hence there exists  some  $\xi_2\in \{\theta,\theta-0\},$   such that 
			$$|E_k(\xi_1+x\cdot \omega)-E_k(\xi_2)|\leq 11\delta_k.$$
			Thus by $x\neq o$, \eqref{DC} and  the inductive hypothesis \eqref{Lc},
			\begin{align*}
				L\gamma\|x\|_1^{-\tau}\leq L \|x\cdot \omega\|_\T	\leq |E_k(\xi_1+x\cdot\omega)-E_k(\xi_2)|
				&\leq 11\delta_k.
			\end{align*}
			It follows that  
			\begin{equation}\label{1606}
				\|x\|_1\geq (L\gamma/10)^{\frac{1}{\tau}}\delta_k^{-\frac{1}{\tau}}= (L\gamma/10)^{\frac{1}{\tau}}e^{\frac{1}{\tau}l_k^{\frac{2}{3}}}\geq 10l_{k+1}.
			\end{equation}
			By \eqref{1606} and $B_{k+1}(x)\subset   (Q_{l_{k+1}+50l_{k}}+x)$, it follows that 
			\begin{equation}\label{1610}
				o\notin B_{k+1}(x).
			\end{equation}
			By \eqref{1545} and \eqref{1610}, we get $B_{k+1}(x)\subset B_{m+1}^o$.
		\end{proof}
		By translation argument,  Lemma \ref{ag} yields the following corollary
		\begin{cor} \label{cor}	Let $m$ be an integer of  $0\leq m\leq n$ and $x\in \Z^d$. Denote $B_{m+1}^o(x):=B^o_{m+1}+x$. Then we have the followings{\rm :}
			\begin{itemize}
				\item For any $E\in D(E_m(0),10\delta_m)$, the set  $B_{m+1}^o(x)$ is $m$-good related to $(\{-x\cdot\omega\},E)$.
				\item  For any $E\in D(E_m(1-0),10\delta_m)$, the set  $B_{m+1}^o(x)$ is $m$-good related to $(\{-x\cdot\omega\}-0,E)$.
				\item For any  $\theta\in \R$ with $\theta+x\cdot\omega\notin \Z$ and   $E\in A_m(\theta+x\cdot\omega ,10\delta_m)$, the set  $B^o_{m+1}(x)$ is $m$-good related to $(\theta,E)$ and $(\theta-0,E)$.
			\end{itemize}
		\end{cor}
		\begin{prop}\label{528}
			For any $\theta\in \R$, $H_{B_{n+1}}(\theta)$ has a unique eigenvalue $E_{n+1}(\theta)$ such that $|E_{n+1}(\theta)-E_{n}(\theta)|\leq e^{-l_n}.$ Moreover, any other eigenvalues of  $H_{B_{n+1}}(\theta)$, $\hat E$ except $E_{n+1}(\theta)$  satisfy $|\hat E- E_n(\theta)|>10\delta_n$.
		\end{prop}
		\begin{proof}
			Fix $\theta\in \R$. Consider $d_{\theta,n+1}(z):=\operatorname{det}[z\operatorname{Id}_{B_{n+1}}-H_{B_{n+1}}(\theta)]$, which is the characteristic polynomial of $H_{B_{n+1}}(\theta)$. In the following, we restrict $z$ in the complex neighborhood $$D:=\begin{cases}
				D(E_n(0),10\delta_n),\ \ &\theta\in \Z,\\
				A_n(\theta,10\delta_n),\ \ &\theta\notin\Z.
			\end{cases}$$
			It suffices to prove that $d_{\theta,n+1}(z)$ has a unique zero $E_{n+1}(\theta)$ in $\bar{D}$ satisfying $|E_{n+1}(\theta)-E_n(\theta)|\leq e^{-l_n}$. Denote $B_{n+1}^o:=B_{n+1}\setminus \{o\}$.
			We can write 
			$$	z\operatorname{Id}_{B_{n+1}}-H_{B_{n+1}}(\theta)=\begin{pmatrix}
				z-v(\theta) & c^{\text{T}} \\
				c & z\operatorname{Id}_{B_{n+1}^o}-H_{B_{n+1}^o}(\theta)
			\end{pmatrix}.$$
			By Lemma \ref{ag}  and  Lemma \ref{con}, $B_{n+1}^o$ (resp. $B_n^o$) is $n$-good (resp. $(n-1)$-good) related to $(\theta,z)$. Hence by Hypothesis \ref{h3}, we have 
			\begin{align}\label{bn+1}
				\|G_{B_{n+1}^o}^{\theta,z}\|&\leq10\delta_n^{-1},  \\
				\label{bn+}
				|G_{B_{n}^o}^{\theta,z}(x,y)|&\leq e^{-\gamma_{n-1}\|x-y\|_1}, \ \|x-y\|_1\geq l_{n-1}^{\frac{5}{6}}.
			\end{align}
			Denote the Schur complement by \begin{align}\label{923+}
				s_{\theta,n+1}(z):&=z-v(\theta) -c^{\text{T}}\left(z\operatorname{Id}_{B_{n+1}^o}-H_{B_{n+1}^o}(\theta)\right)^{-1}c \nonumber \\
				&=z-v(\theta)-r_{\theta,n+1}(z).
			\end{align}
			It follows from Schur complement formula, $d_{\theta,n+1}(z)=0$ if and only if $s_{\theta,n+1}(z)=0$. 
			Since the vector $c$ has only nonzero elements on the nearest neighbor of $o$, by the  resolvent identity, it follows that 
			\begin{align}
				\nonumber s_{\theta,n+1}(z)&=s_{\theta,n}(z)-\varepsilon \sum_{\substack{\|x\|_1,\|y\|_1=1\\(w,w')\in \partial_{B_{n+1}^o}B_{n}^o}}c(x)G_{B_{n}^o}^{\theta,z}(x,w)G_{B_{n+1}^o}^{\theta,z}(w',y)c(y)\\
				&=s_{\theta,n}(z)-g_{\theta,n+1}(z), \label{1945}
			\end{align}
			where $$s_{\theta,n}(z)= z-v(\theta) -\tilde{c}^{\text{T}}\left(z\operatorname{Id}_{B_{n}^o}-H_{B_{n}^o}(\theta)\right)^{-1}\tilde{c}$$
			with $\tilde{c}$ being  the restriction of $c$ on $B_{n}^o$. 
			Since $	Q_{l_n}\subset B_n$, it follows that  for $\|x\|_1=1$ and  $w\in \partial^-_{B_{n+1}^o}B_{n}^o$, we have 
			$$\|x-w\|_1\geq l_n-1\geq l_{n-1}^\frac{5}{6}.$$ 
			By \eqref{bn+1} and \eqref{bn+} and $\gamma_{n-1}\geq 10$ (c.f. \eqref{rate}), it follows that 
			\begin{align}\label{res}
				|	g_{\theta,n+1}(z)|\leq l_n^d e^{-\gamma_{n-1}(l_n-1)}\delta_n^{-1}\leq e^{-\gamma_{n-1}(l_n-2l_n^{\frac{2}{3}})}\leq e^{-3l_n}.
			\end{align}
			By the inductive  hypotheses \eqref{sa1}, \eqref{sa10} and Lemma \ref{con}, we have 
			\begin{align}
				s_{\theta,n+1}(z)&=	(z-E_n(\theta))\left (1+O(\sum_{k=0}^{n-1}\delta_k)\right )-	g_{\theta,n+1}(z)\label{1924}\\
				&=\left (1+O(\sum_{k=0}^{n-1}\delta_k)\right )\left(z-E_n(\theta)-\tilde{g}_{\theta,n+1}(z)\right ) \label{1845}
			\end{align}
			with \begin{equation}\label{1853}
				|\tilde{g}_{\theta,n+1}(z)|\leq 2|g_{\theta,n+1}(z)|\leq  e^{-2l_n}.
			\end{equation} Hence, 
			\begin{equation}\label{15}
				|\tilde{g}_{\theta,n+1}(z)|<10\delta_n\leq |z-E_n(\theta)|, \ \ z\in \partial D.
			\end{equation}
			By \eqref{1845}, it follows that  $s_{\theta,n+1}(z)$ and  $z-E_n(\theta)-\tilde{g}_{\theta,n+1}(z)$ have the same zero(s) in $\bar{D}$. And by \eqref{15}  together with Rouch\'e theorem, $z-E_n(\theta)-\tilde{g}_{\theta,n+1}(z)$ has the same number of zero(s) as $z-E_{n}(\theta)$ in $\bar{D}$. Denote by $E_{n+1}(\theta)$ the unique zero of $s_{\theta,n+1}(z)$ in $\bar{D}$. Since   $s_{\theta,n+1}(E_{n+1}(\theta))=0$, by \eqref{1845}, it follows that $E_{n+1}(\theta)-E_n(\theta)-\tilde{g}_{\theta,n+1}(E_{n+1}(\theta))=0$. Hence by  \eqref{1853}, we have 
			\begin{equation}\label{1844}
				|E_{n+1}(\theta)-E_n(\theta)|=|\tilde{g}_{\theta,n+1}(E_{n+1}(\theta))|\leq e^{-2l_n}.
			\end{equation}
		\end{proof}
		
		\begin{rem}\label{454n}
			Since $E_n(\theta)$ is  $1$-periodic  and   $H_{B_{n+1}}(\theta)$ is  $1$-periodic and self-adjoint, it follows that $E_{n+1}(\theta)$ is a $1$-periodic,   real-valued function.
		\end{rem}
		
		\begin{lem}\label{sln}
			With the notation in the proof of Proposition \ref{528}, we have the following approximation, for  $z\in D$,  \begin{equation}\label{appn}
				s_{\theta,n+1}(z)=(z-E_{n+1}(\theta))\left (1+O(\sum_{k=0}^{n}\delta_k)\right ).
			\end{equation}
			Moreover, for $z\in D'$ with  $$ D':=\begin{cases}
				D(E_n(0),5\delta_n),\ \ &\theta\in \Z,\\
				A_n(\theta,5\delta_n),\ \ &\theta\notin\Z,
			\end{cases}$$
			we have \begin{equation}\label{Den}
				s_{\theta,n+1}'(z)=1+O(\sum_{k=0}^{n}\delta_k)	.
			\end{equation}
			The above  $``O(\sum_{k=0}^{n}\delta_k)"$ are understood as analytic functions of $z$ bounded by  $\sum_{k=0}^{n}\delta_k$.
		\end{lem}
		\begin{proof}
			Since $E_{n+1}(\theta)$ is the unique zero of $s_{\theta,n+1}(z)$ in $\bar{D}$, the function $$h_\theta(z):=\frac{s_{\theta,n+1}(z)}{z-E_{n+1}(\theta)}-1$$
			is analytic in $D$. Moreover, by \eqref{1924}, \eqref{1853} and \eqref{1844}, one has for $z\in \partial D$
			\begin{align*}
				|h_\theta(z)|&=\left|\frac{E_{n+1}(\theta)-E_n(\theta)-g_{\theta,n+1}(z)}{z-E_{n+1}(\theta)}+\frac{z-E_{n}(\theta)}{z-E_{n+1}(\theta)}O(\sum_{k=0}^{n-1}\delta_k)\right|\\
				&\leq\left|\frac{E_{n+1}(\theta)-E_n(\theta)}{z-E_{n+1}(\theta)}\right| +\left|\frac{g_{\theta,n+1}(z)}{z-E_{n+1}(\theta)}\right| +\left|\frac{z-E_{n}(\theta)}{z-E_{n+1}(\theta)}\right|\sum_{k=0}^{n-1}\delta_k\\
				&\leq \frac{2e^{-2l_n}}{10\delta_n-e^{-2l_n}}+\frac{10\delta_n}{10\delta_n-e^{-2l_n}}\sum_{k=0}^{n-1}\delta_k\\
				&<\delta_n^2+(1+\delta_n^2)\sum_{k=0}^{n-1}\delta_k\\&<\sum_{k=0}^{n}\delta_k.
			\end{align*}	
			Hence by maximum principle, we have  
			$$\sup_{z\in D}|h_\theta(z)|\leq \sup_{z\in\partial  D}|h_\theta(z)|<\sum_{k=0}^{n}\delta_k$$ and \eqref{appn} follows. By \eqref{1945}, we have  $s_{\theta,n+1}'(z)=s_{\theta,n}'(z)-g_{\theta,n+1}'(z)$. By   \eqref{1853} and Cauchy integral estimate, it follows that 
			\begin{equation}\label{1955}
				\sup_{z\in D'} |g_{\theta,n+1}'(z)|\leq \delta_n^{-1}\sup_{z\in D} |g_{\theta,n+1}(z)|\leq \delta_n^{-1} e^{-2l_n}\leq \delta_n
			\end{equation}
			By the inductive hypotheses \eqref{sa2}, \eqref{sa20} and Lemma \ref{con}, we have \begin{equation}\label{1954}
				s_{\theta,n}'(z)=1+O(\sum_{k=0}^{n-1}\delta_k). 
			\end{equation} Then \eqref{Den} follows from   \eqref{1955} and \eqref{1954}.
		\end{proof}
		\begin{rem}
			Analogues of Proposition \ref{528} and Lemma \ref{sln} hold true for the left limit operator $H_{B_{n+1}}(\theta-0)$. We state them as a proposition:
			\begin{prop}\label{**2}
				For any $\theta\in \R$, $H_{B_{n+1}}(\theta-0)$ has a unique eigenvalue $E_{n+1}(\theta-0)$ such that $|E_{n+1}(\theta-0)-E_n(\theta-0)|\leq e^{-l_n}$. Moreover, since $H_{B_{n+1}}(\theta)=H_{B_{n+1}}(\theta-0)$ and $E_n(\theta )=E_n(\theta-0)$ for $\theta\notin \left\{\{-x\cdot\omega\}\right\}_{x\in B_{n+1}}$, it follows that $E_{n+1}(\theta)=E_{n+1}(\theta-0)$ for $\theta\notin \left\{\{-x\cdot\omega\}\right\}_{x\in B_{n+1}}$.
			\end{prop}
			\begin{lem}\label{H-*} Denote $$\tilde D:=\begin{cases}
					D(E_n(1-0),10\delta_n),\ \ &\theta\in \Z,\\
					A_n(\theta,10\delta_n),\ \ &\theta\notin\Z,
				\end{cases}$$
				and $$\tilde D':=\begin{cases}
					D(E_n(1-0),5\delta_n),\ \ &\theta\in \Z,\\
					A_n(\theta,5\delta_n),\ \ &\theta\notin\Z.
				\end{cases}$$
				For  $z\in\tilde  D$, the Schur complement 
				\begin{align}\label{923.*}
					s_{\theta-0,n+1}(z):&=z-v(\theta-0) -c^{\text{T}}\left(z\operatorname{Id}_{B_{n+1}^o}-H_{B_{n+1}^o}(\theta-0)\right)^{-1}c\nonumber	\\
					&=z-v(\theta-0)-r_{\theta-0,n+1}(z)
				\end{align}
				exists and 	satisfies 
				\begin{align*}
					s_{\theta-0,n+1}(z)&=(z-E_{n+1}(\theta-0))\left(1+O(\sum_{k=0}^{n}\delta_k)\right) , & z\in  \tilde D,\\
					s_{\theta-0,n+1}'(z)&=1+O(\sum_{k=0}^{n}\delta_k) ,& z\in  \tilde D'.
				\end{align*}
				
			\end{lem}
			
		\end{rem}

		\subsubsection{Verification of the Lipschitz monotonicity property for $E_{n+1}(\theta)$} 
		Denote $\{-x\cdot\omega\}$ by $\beta_x$ and denote by  $0=\alpha_1<\alpha_2<\cdots<\alpha_{|B_{n+1}|}<\alpha_{|B_{n+1}|+1}=1$ the discontinuities of $H_{B_{n+1}}(\theta)$, where 
		$$\{\alpha_1,\alpha_2,\cdots,\alpha_{|B_{n+1}|}\}=\left\{\beta_x\right\}_{x\in B_{n+1}}.$$
		Also denote  the Rellich functions of $H_{B_{n+1}}(\theta)$ in non-decreasing order by $\lambda_i(\theta)$ ($1\leq i\leq |B_{n+1}|$) which are $1$-periodic and continuous on $[0,1)$ except on the set $\{\alpha_1,\alpha_2,\cdots,\alpha_{|B_{n+1}|}\}$ with the  local Lipschitz monotonicity property as  mentioned at the beginning of section \ref{n=1}.
		\begin{prop}\label{qiang}
			Let $E_{n+1}(\theta)$ be the eigenvalue of $H_{B_{n+1}}(\theta)$ constructed in Proposition \ref{528}. Then $E_{n+1}(\theta)$ satisfies the following two properties{\rm :}
			\begin{itemize}	\item On each small interval $[\alpha_k,\alpha_{k+1})$, $E_{n+1}(\theta)$ coincides with exactly one branch of $\lambda_i(\theta)$, and hence is continuous and satisfies the local Lipschitz monotonicity property on $[\alpha_k,\alpha_{k+1})${\rm :}
				\begin{equation*}\label{400}
					E_{n+1}(\theta_2)-E_{n+1}(\theta_1)\geq L(\theta_2-\theta_1), \ \ \alpha_k\leq \theta_1\leq \theta_2<\alpha_{k+1}.
				\end{equation*}
				\item $E_1(\theta)$ has  non-negative ``jumps'' at  $\alpha_k$ $(2\leq k\leq |B_{n+1}|)${\rm :}  
				$$E_{n+1}(\alpha_k)\geq E_{n+1}(\alpha_k-0). $$ 
			\end{itemize}
		\end{prop}
		\begin{proof}
			The proof of the first property is similar to the corresponding proof at the first scale (c.f. Proposition \ref{814}).
			For the second property, 	since $\alpha_k\neq 0$, there exists some $x\in B_{n+1}$ with  $x_1\neq o$ such that  $\alpha_k=\beta_x$. 
			Recall the notation $s_{\theta,n+1}(z)$, $r_{\theta,n+1}(z)$, $s_{\theta-0,n+1}(z)$, $r_{\theta-0,n+1}(z)$  (c.f. \eqref{923+}, \eqref{923.*}).
			We will restrict $z$ to the interval \begin{equation*}
				I:=(E_n(\beta_x-0)-5\delta_n,E_n(\beta_x)+5\delta_n)
			\end{equation*} in the following discussion  so that all the functions of $z$ will be real-valued. Since $\beta_x\in \R\setminus\Z$ and $I\subset A_n(\beta_x,5\delta_n)$,  it follows that  \eqref{Den} holds true.   
			If we can  prove \begin{equation}\label{1055+}
				s_{\beta_x,n+1}(z)\leq s_{\beta_x-0,n+1}(z),\ \ z\in I,
			\end{equation}
			then 
			\begin{equation}\label{ky}
				s_{\beta_x,n+1}(E_{n+1}(\beta_x-0))\leq s_{\beta_x-0,n+1}(E_{n+1}(\beta_x-0))=0=s_{\beta_x,n+1}(E_{n+1}(\beta_x)).
			\end{equation} The desired conclusion   \begin{equation}\label{key}
				E_{n+1}(\beta_x)\geq  E_{n+1}(\beta_x-0)
			\end{equation} will follow from \eqref{ky} and the monotonicity of  $s_{\beta_x,n+1}(z)$ (since by \eqref{Den}, $s_{\beta_x,n+1}'(z)\approx1$). The remainder of the proof  aims at verifying \eqref{1055+}. Recalling \eqref{923+}, we compute 
			\begin{align}
				&s_{\beta_x,n+1}(z)-s_{\beta_x-0,n+1}(z)\nonumber\\ =&r_{\beta_x-0,n+1}(z)-r_{\beta_x,n+1}(z)\nonumber\\
				=&c^{\text{T}}\left[ \left(z\operatorname{Id}_{B_{n+1}^o}-H_{B_{n+1}^o}(\beta_x-0)\right)^{-1}- \left(z\operatorname{Id}_{B_{n+1}^o}-H_{B_{n+1}^o}(\beta_x)\right)^{-1}\right]c\nonumber\\
				=&c^{\text{T}}\left[ \left(z\operatorname{Id}_{B_{n+1}^o}-H_{B_{n+1}^o}(\beta_x-0)\right)^{-1}{\bm e}_x {\bm e}_x^{\rm T}\left(z\operatorname{Id}_{B_{n+1}^o}-H_{B_{n+1}^o}(\beta_x)\right)^{-1}\right]c. \label{1210+}
			\end{align}
			On the last line of the above equation  we use the resolvent identity and the  equation $$H_{B_{n+1}^o}(\beta_x-0)-H_{B_{n+1}^o}(\beta_x)={\bm e}_x {\bm e}_x^{\rm T}.$$
		Cramer's rule implies 
			$$\left(z\operatorname{Id}_{B_{n+1}^o}-H_{B_{n+1}^o}(\beta_x-0)\right)^{-1}=\frac{\left(z\operatorname{Id}_{B_{n+1}^o}-H_{B_{n+1}^o}(\beta_x-0)\right)^{\#}}{\operatorname{det}\left(z\operatorname{Id}_{B_{n+1}^o}-H_{B_{n+1}^o}(\beta_x-0)\right)}$$
			and $$\left(z\operatorname{Id}_{B_{n+1}^o}-H_{B_{n+1}^o}(\beta_x)\right)^{-1}=\frac{\left(z\operatorname{Id}_{B_{n+1}^o}-H_{B_{n+1}^o}(\beta_x)\right)^{\#}}{\operatorname{det}\left(z\operatorname{Id}_{B_{n+1}^o}-H_{B_{n+1}^o}(\beta_x)\right)}.$$
			A simple computation shows that 
			\begin{equation}\label{1437+}
				\left(z\operatorname{Id}_{B_{n+1}^o}-H_{B_{n+1}^o}(\beta_x-0)\right)^{\#}{\bm e}_x=\left(z\operatorname{Id}_{B_{n+1}^o}-H_{B_{n+1}^o}(\beta_x)\right)^{\#}{\bm e}_x.
			\end{equation}
			We denote by $b(z)$ the vector in \eqref{1437+}.
			Hence 
			\begin{equation}\label{1528+}
				\eqref{1210+}=\frac{1}{\operatorname{det}\left(z\operatorname{Id}_{B_{n+1}^o}-H_{B_{n+1}^o}(\beta_x-0)\right)}\frac{1}{\operatorname{det}\left(z\operatorname{Id}_{B_{n+1}^o}-H_{B_{n+1}^o}(\beta_x)\right)}\left(c^{\rm T} b(z)\right)^2.
			\end{equation}
			From \eqref{1528+}, it follows  that the sign of $s_{\beta_x,n+1}(z)-s_{\beta_x-0,n+1}(z)$ is the same as the sign of  the product \begin{equation}\label{532}
				\operatorname{det}\left(z\operatorname{Id}_{B_{n+1}^o}-H_{B_{n+1}^o}(\beta_x-0)\right)\operatorname{det}\left(z\operatorname{Id}_{B_{n+1}^o}-H_{B_{n+1}^o}(\beta_x)\right).
			\end{equation}
			The following technical Lemma \ref{1448} shows that the sign of \eqref{532} is negative and hence by  \eqref{1210+} and  \eqref{1528+}, it follows that \eqref{1055+} holds true. Thus we finish the proof of \eqref{key} and the second property.
		\end{proof}
		\begin{lem}\label{1448}
			Let $\beta_x=\{-x\cdot\omega\}$ with $x\neq o$. Then for   \begin{equation}\label{I}
				z\in I:=(E_n(\beta_x-0)-5\delta_n,E_n(\beta_x)+5\delta_n),
			\end{equation} we have 
			\begin{equation}\label{key1}
				\operatorname{det}\left(z\operatorname{Id}_{B_{n+1}^o}-H_{B_{n+1}^o}(\beta_x-0)\right)\operatorname{det}\left(z\operatorname{Id}_{B_{n+1}^o}-H_{B_{n+1}^o}(\beta_x)\right)<0.
			\end{equation}
		\end{lem}
		\begin{proof}
			Fix $z\in I$. For $0\leq m\leq n$, we define the set 
			$$S_m:=S_m(\beta_x,z)\cup S_m(\beta_x-0,z).$$
		In the  case   \begin{equation}\label{1509}
				\text{$B_{n+1}^o\cap S_m=\emptyset$ ($0\leq m\leq n$). }
			\end{equation}
		The proof is similar to the corresponding proof at the first scale  (c.f. Proposition \ref{814}).

		We consider the  negation case of \eqref{1509}.
			Let  $m$ be the  largest integer of $0\leq m\leq n$ such that $B_{n+1}^o\cap S_m\neq \emptyset$ and let  $x_1\in S_m\cap B_{n+1}^o$.	Without loss of generality,  we assume $x_1\in S_m(\beta_x,z)\cap B_{n+1}^o$ since the case $x_1\in S_m(\beta_x-0,z)\cap B_{n+1}^o$ is analogous. Since $z\in I$, it follows from Lemma \ref{con} that 
			$$z\in I\subset A_n(\beta_x,5\delta_n)\subset A_m(\beta_x,5\delta_m).$$
			And since $x_1\in S_m(\beta_x,z)\cap B_{n+1}^o$, we have $0<\|x_1\|_1\leq 2l_{n+1}$ and $|E_m(\beta_x+x_1\cdot\omega)-z|\leq\delta_m$.
			Then  by \eqref{Lc} and Lemma \ref{no} , it follows that 
			\begin{align*}
				L\|x_1\cdot \omega\|_\T	\leq &\min (|E_m(\beta_x+x_1\cdot \omega )-E_m(\beta_x)|,|E_m(\beta_x+x_1\cdot \omega )-E_m(\beta_x-0)|)
				\\ \leq &\operatorname{dist}(E_m(\beta_x+x_1\cdot\omega),A_m(\beta_x,5\delta_m))+5\delta_m\\
				\leq & |E_m(\beta_x+x_1\cdot\omega)-z|+5\delta_m\\
				\leq& 6\delta_m. 
			\end{align*}
			Hence by \eqref{DC} and $0<\|x_1\|_1\leq 2l_{n+1}$, we have 
			$$L\gamma(2l_{n+1})^{-\tau}\leq L\|x_1\cdot \omega\|_\T 	\leq 6 \delta_m=6e^{-l_m^{\frac{2}{3}}}\leq  L\gamma(2l_m^{16})^{-\tau}.$$
			It follows from the above inequality that $l_{n+1}\geq l_m^{16}$ and hence $m\leq n-3$.\\
			The proof proceeds by case analysis:
			\begin{itemize}
				\item[\textbf{Case 1:}] Assume that 
				\begin{equation}\label{case1}
					x_1\in S_m(\beta_x,z)\cap B_{n+1}^o,\ \  x_1=x.
				\end{equation}
				We denote  $B_{n+1}^1:=B_{n+1}^o\setminus\{x\}$ as before. Then $H_{B_{n+1}^1}(\beta_x)=H_{B_{n+1}^1}(\beta_x-0)$. Let  $p\in S_k(\beta_x,z)\cap B_{n+1}^1$.  (By the maximal property of $m$, it follows that $k\leq m$.) Then $B_{k+1}(p)\subset B_{n+1}^o$ since $B_{n+1}^o$ is $n$-regular related to $(\beta_x,z)$ (c.f. Lemma \ref{ag}).
				Since   $x_1\in S_m(\beta_x,z)$ and $p\in S_k(\beta_x,z)$, it follows that 
				$$|E_m(\beta_x+x\cdot\omega)-z|\leq \delta_m, \ \ |E_k(\beta_x+p\cdot\omega)-z|\leq \delta_k. $$
				Since $p\neq x$, estimating as before by \eqref{h12}, \eqref{Lc} and \eqref{DC}, we have\begin{align*}
					L\gamma\|p-x\|_1^{-\tau} &\leq  L\gamma\|(p-x)\cdot\omega\|_\T\\ &\leq |E_k(\beta_x+p\cdot\omega)-E_k(\beta_x+x\cdot \omega )|\\&\leq |E_k(\beta_x+p\cdot\omega)-z|+|z-E_m(\beta_x+x\cdot\omega)|\\ &+|E_m(\beta_x+x\cdot\omega)-E_k(\beta_x+x\cdot\omega)|
					\\&\leq \delta_k+\delta_m +\sum_{j=k}^{m-1}e^{-l_j}\leq 10\delta_k,
				\end{align*} yielding $\|p-x\|_1\geq (L\gamma/10)^\frac{1}{\tau}\delta_k^{-\tau}\geq 10l_{k+1}$. Since $B_{k+1}\subset Q_{2l_{k+1}}$,  it follows that  $x\notin B_{k+1}(p)$ and hence $B_{k+1}(p)\subset B_{n+1}^1$. Thus $B_{n+1}^1$ is $(m+1)$-regular related to $(\beta_x,z)$. Since  $S_{m+1}(\beta_x,z)\cap B_{n+1}^1=\emptyset$ by the maximal property of $m$, it follows that $B_{n+1}^1$ is $(m+1)$-good related to $(\beta_x,z)$.
				Hence by Hypothesis \ref{h3}, we get the resolvent bound 
				\begin{equation}\label{1623++}
					\left \|\left( z\operatorname{Id}_{B_{n+1}^1}-H_{B_{n+1}^1}(\beta_x)  \right)^{-1}\right\|\leq 10\delta_{m+1}^{-1}.
				\end{equation}
				Since $$\langle H_{B_{n+1}^o}(\beta_x-0){\bm e}_x,{\bm e}_x\rangle=v(1-0)=1, \  \ \langle H_{B_{n+1}^o}(\beta_x){\bm e}_x,{\bm e}_x\rangle=v(0)=0$$ and  $H_{B_{n+1}^1}(\beta_x -0)=H_{B_{n+1}^1}(\beta_x )$, 
				we can write 
				$$z\operatorname{Id}_{B_{n+1}^o}-H_{B_{n+1}^o}(\beta_x-0)=\begin{pmatrix}
					z-1 & \tilde{c}^{\text{T}} \\
					\tilde{c} & z\operatorname{Id}_{B_{n+1}^1}-H_{B_{n+1}^1}(\beta_x)
				\end{pmatrix}
				$$
				and 
				$$z\operatorname{Id}_{B_{n+1}^o}-H_{B_{n+1}^o}(\beta_x)=\begin{pmatrix}
					z-0 & \tilde{c}^{\text{T}} \\
					\tilde{c} & z\operatorname{Id}_{B_{n+1}^1}-H_{B_{n+1}^1}(\beta_x)
				\end{pmatrix}.
				$$
				Hence by Schur complement formula, 
				\begin{align}\label{351++}
					&\nonumber	\operatorname{det}\left(z\operatorname{Id}_{B_{n+1}^o}-H_{B_{n+1}^o}(\beta_x-0)\right)\\=&\left [z-1-\tilde{c}^{\rm T}\left(z\operatorname{Id}_{B_{n+1}^1}-H_{B_{n+1}^1}(\beta_x)\right)^{-1}\tilde{c}\right ]\operatorname{det}\left(z\operatorname{Id}_{B_{n+1}^1}-H_{B_{n+1}^1}(\beta_x)\right)
				\end{align} and 
				\begin{align}\label{353++}
					&\nonumber	\operatorname{det}\left(z\operatorname{Id}_{B_{n+1}^o}-H_{B_{n+1}^o}(\beta_x)\right)\\=&\left [z-0-\tilde{c}^{\rm T}\left(z\operatorname{Id}_{B_{n+1}^1}-H_{B_{n+1}^1}(\beta_x)\right)^{-1}\tilde{c}\right ]\operatorname{det}\left(z\operatorname{Id}                                                           _{B_{n+1}^1}-H_{B_{n+1}^1}(\beta_x)\right) .
				\end{align}
				Denote $B_{m+1}^o(x):=B_{m+1}(x)\setminus\{x\}$. Since $x\in S_m(\beta_x,z)\cap B_{n+1}^o$ and $B_{n+1}^o$ is $n$-regular related to $(\beta_x,z)$ (c.f. Lemma  \ref{ag}), it follows that $B_{m+1}(x)\subset B_{n+1}^o$ and hence $B_{m+1}^o(x)\subset B_{n+1}^1 $.   Since $\beta_x+x\cdot\omega\in \Z $ and  $x\in S_m(\beta_x,z)$, it follows that $|E_m(\beta_x+x\cdot\omega)-z|=|E_m(0)-z|\leq \delta_m$, yielding $z\in D(E_m(0),10\delta_m)$. It  follows from Corollary \ref{cor}  that $B_{m+1}^o(x)$ is $m$-good related to $(\beta_x,z)$. Thus  
				\begin{equation}\label{bm+}
					|G_{B_{m+1}^o(x)}^{\beta_x,z}(\tilde y,y)|\leq e^{-\gamma_{m}\|\tilde y-y\|_1}, \ \|\tilde y-y\|_1\geq l_{m}^{\frac{5}{6}}.
				\end{equation}
				By the resolvent identity, \eqref{1623++}, \eqref{bm+} and \eqref{sa10}, estimating as in \eqref{res}, we have 
				\begin{align}
					\nonumber	&z-0-\tilde{c}^{\rm T}\left(z\operatorname{Id}_{B_{n+1}^1}-H_{B_{n+1}^1}(\beta_x)\right)^{-1}\tilde{c}\\
					\nonumber	=&z-0-\tilde{c}^{\rm T}\left(z\operatorname{Id}_{B_{m+1}^o(x)}-H_{B_{m+1}^o(x)}(\beta_x)\right)^{-1}\tilde{c}\\-&\nonumber\varepsilon \sum_{\substack{\|\tilde{y}-x\|_1,\|y-x\|_1=1\\(w,w')\in \partial_{B_{n+1}^1}B_{m+1}^o(x)}}\tilde c(\tilde y)G_{B_{m+1}^o(x)}^{\beta_x,z}(\tilde y,w)G_{B_{n+1}^1}^{\beta_x,z}(w',y)\tilde c(y)\\
					\nonumber=&z-0-c^{\rm T}\left(z\operatorname{Id}_{B_{m+1}^o}-H_{B_{m+1}^o}(\beta_x+x\cdot\omega)\right)^{-1}c-O(e^{-3l_{m+1}})\\
					\nonumber	=&\left(z-E_{m+1}(\beta_x+x\cdot \omega) -O(e^{-2l_{m+1}})\right)\left(1+O(\sum_{k=0}^m\delta_k)\right)\\
					=&\left(z-E_{m+1}(0) -O(e^{-2l_{m+1}})\right)\left(1+O(\sum_{k=0}^m\delta_k)\right).\label{2036}
				\end{align}
				The above $O(\star)$ are understood as  numbers bounded by $\star$.
				Since  $x\notin S_{m+1}(\beta_x,z)$ by the maximal property of $m$, it follows that  
				\begin{equation}\label{2038}
					|z-E_{m+1}(\beta_x+x\cdot\omega)|=|z-E_{m+1}(0)|\geq \delta_{m+1}.
				\end{equation}
				Since $z\geq E_n(\beta_x-0)-5\delta_n$,     it follows from   inductive hypotheses \eqref{h12}, \eqref{Lc} that 
				\begin{align}\label{2}
					\nonumber	z-E_{m+1}(0)&\geq z-E_{m+1}(\beta_x-0)\\
					\nonumber	&\geq z-E_{n}(\beta_x-0)-|E_n(\beta_x-0)-E_{m+1}(\beta_x-0)|\\
					\nonumber	&\geq -5\delta_n-\sum_{k=m+1}^{n-1}e^{-l_k}\\
					&>-\delta_{m+1}/2.
				\end{align}
				By \eqref{2038} and \eqref{2}, we have 	$z-E_{m+1}(0)\geq \delta_{m+1}>e^{-2l_{m+1}}$. Hence by \eqref{2036}, it follows that 
				\begin{equation}\label{0}
					z-0-\tilde{c}^{\rm T}\left(z\operatorname{Id}_{B_{n+1}^1}-H_{B_{n+1}^1}(\beta_x)\right)^{-1}\tilde{c}>0.
				\end{equation}
				Since $x\in S_m(\beta_x,z)$, we have $|z-E_m(0)|<\delta_m$. Hence by \eqref{h12} and  \eqref{2036}, it follows that 
				\begin{align}
					\nonumber	&z-1-\tilde{c}^{\rm T}\left(z\operatorname{Id}_{B_{n+1}^1}-H_{B_{n+1}^1}(\beta_x)\right)^{-1}\tilde{c}\\
					\nonumber	=&\left(z-E_{m+1}(0) -O(e^{-2l_{m+1}})\right)\left(1+O(\sum_{k=0}^m\delta_k)\right)-1\\
					\nonumber	\leq &2(|z-E_m(0)|+|E_m(0)-E_{m+1}(0)|+e^{-2l_{m+1}})-1\\
					\leq &2(\delta_m+\delta_m)-1<0.\label{1}
				\end{align}
				Then \eqref{key1} follows from \eqref{351++}, \eqref{353++}, \eqref{0} and \eqref{1}.
				\item[\textbf{Case 2:}] Assume that\begin{equation}\label{case2}
					x_1\in S_m(\beta_x,z)\cap B_{n+1}^o,\ \  x_1\neq x.
				\end{equation}
				We denote $B_{n+1}^{(1)}=B_{n+1}^o\setminus \{x_1\}$.
				Let  $p\in S_k(\beta_x,z)\cap B_{n+1}^{(1)}$.  (By the maximal property of $m$, it follows that $k\leq m$.) Then $B_{k+1}(p)\subset B_{n+1}^o$ since $B_{n+1}^o$ is $n$-regular related to $(\beta_x,z)$ (c.f. Lemma  \ref{ag}).
				Since $p\neq x_1$, estimating as before  by \eqref{h12}, \eqref{Lc} and \eqref{DC}, we have
				\begin{align*}
					L\gamma\|p-x_1\|_1^{-\tau} &\leq  L\gamma\|(x_1-p)\cdot\omega\|_\T\\ &\leq |E_k(\beta_x+p\cdot\omega)-E_k(\beta_x+x_1\cdot \omega )|\\&\leq |E_k(\beta_x+p\cdot\omega)-z|+|z-E_m(\beta_x+x_1\cdot\omega)|\\ &+|E_m(\beta_x+x_1\cdot\omega)-E_k(\beta_x+x_1\cdot\omega)|
					\\&\leq \delta_k+\delta_m +\sum_{j=k}^{m-1}e^{-l_j}\leq 10\delta_k,
				\end{align*} yielding $\|p-x_1\|_1\geq (L\gamma/10)^\frac{1}{\tau}\delta_k^{-\tau}\geq 10l_{k+1}$. It follows that  $x_1\notin  B_{k+1}(p)$ and hence $B_{k+1}(p)\subset B_{n+1}^{(1)}$. Thus $B_{n+1}^{(1)}$ is $(m+1)$-regular related to $(\beta_x,z)$. Since  $S_{m+1}(\beta_x,z)\cap B_{n+1}^{(1)}=\emptyset$ by the maximal property of $m$, it follows that $B_{n+1}^{(1)}$ is $(m+1)$-good related to $(\beta_x,z)$.
				Hence by Hypothesis \ref{h3}, we get the resolvent bound 
				\begin{equation}\label{1623,}
					\left \|\left( z\operatorname{Id}_{B_{n+1}^{(1)}}-H_{B_{n+1}^{(1)}}(\beta_x)  \right)^{-1}\right\|\leq 10\delta_{m+1}^{-1}.
				\end{equation}
				Let  $p\in S_k(\beta_x-0,z)\cap B_{n+1}^{(1)}$.  (By the maximal property of $m$, it follows that $k\leq m$.) Then $B_{k+1}(p)\subset B_{n+1}^o$ since $B_{n+1}^o$ is $n$-regular related to $(\beta_x-0,z)$ (c.f. Lemma  \ref{ag}).
				Since $p\neq x_1$, estimating as before  by \eqref{h12}, \eqref{Lc} and \eqref{DC}, we have 
			\begin{align*}
				L\gamma\|p-x_1\|_1^{-\tau} &\leq  L\gamma\|(x_1-p)\cdot\omega\|_\T\\ &\leq |E_k(\beta_x+p\cdot\omega-0)-E_k(\beta_x+x_1\cdot \omega )|\\&\leq |E_k(\beta_x+p\cdot\omega-0)-z|+|z-E_m(\beta_x+x_1\cdot\omega)|\\ &+|E_m(\beta_x+x_1\cdot\omega)-E_k(\beta_x+x_1\cdot\omega)|
				\\&\leq \delta_k+\delta_m +\sum_{j=k}^{m-1}e^{-l_j}\leq 10\delta_k,
			\end{align*} yielding $\|p-x_1\|_1\geq (L\gamma/10)^\frac{1}{\tau}\delta_k^{-\tau}\geq 10l_{k+1}$. It follows that  $x_1\notin B_{k+1}(p)$ and hence $B_{k+1}(p)\subset B_{n+1}^{(1)}$. Thus $B_{n+1}^{(1)}$ is $(m+1)$-regular related to $(\beta_x-0,z)$. Since  $S_{m+1}(\beta_x-0,z)\cap B_{n+1}^{(1)}=\emptyset$ by the maximal property of $m$, it follows that $B_{n+1}^{(1)}$ is $(m+1)$-good related to $(\beta_x-0,z)$.
			Hence by Hypothesis \ref{h3}, we get the resolvent bound 
			\begin{equation}\label{1623,,}
				\left \|\left( z\operatorname{Id}_{B_{n+1}^{(1)}}-H_{B_{n+1}^{(1)}}(\beta_x-0)  \right)^{-1}\right\|\leq 10\delta_{m+1}^{-1}.
			\end{equation}
			Since $x_1\neq x$, it follows that $v(\beta_x-0+x_1\cdot\omega)=v(\beta_x+x_1\cdot\omega)$ and hence  $$\langle H_{B_{n+1}^o}(\beta_x-0){\bm e}_{x_1},{\bm e}_{x_1}\rangle =\langle H_{B_{n+1}^o}(\beta_x){\bm e}_{x_1},{\bm e}_{x_1}\rangle=v(\beta_x+x_1\cdot\omega).$$ Thus we  can write  write 
			$$z\operatorname{Id}_{B_{n+1}^o}-H_{B_{n+1}^o}(\beta_x-0)=\begin{pmatrix}
				z-v(\beta_x+x_1\cdot\omega) & \tilde{c}^{\text{T}} \\
				\tilde{c} & z\operatorname{Id}_{B_{n+1}^{(1)}}-H_{B_{n+1}^{(1)}}(\beta_x-0)
			\end{pmatrix}
			$$
			and 
			$$z\operatorname{Id}_{B_{n+1}^o}-H_{B_{n+1}^o}(\beta_x)=\begin{pmatrix}
				z-v(\beta_x+x_1\cdot\omega) & \tilde{c}^{\text{T}} \\
				\tilde{c} & z\operatorname{Id}_{B_{n+1}^{(1)}}-H_{B_{n+1}^{(1)}}(\beta_x)
			\end{pmatrix}.
			$$
			Hence by Schur complement formula, 
			\begin{align}\label{351,}
				&\nonumber	\operatorname{det}\left(z\operatorname{Id}_{B_{n+1}^o}-H_{B_{n+1}^o}(\beta_x-0)\right)\\=&\left [z-v(\beta_x+x_1\cdot\omega)-\tilde{c}^{\rm T}\left(z\operatorname{Id}_{B_{n+1}^{(1)}}-H_{B_{n+1}^{(1)}}(\beta_x-0)\right)^{-1}\tilde{c}\right ]\operatorname{det}\left(z\operatorname{Id}_{B_{n+1}^{(1)}}-H_{B_{n+1}^{(1)}}(\beta_x-0)\right)
			\end{align} and 
			\begin{align}\label{353,}
				&\nonumber	\operatorname{det}\left(z\operatorname{Id}_{B_{n+1}^o}-H_{B_{n+1}^o}(\beta_x)\right)\\=&\left [z-v(\beta_x+x_1\cdot\omega)-\tilde{c}^{\rm T}\left(z\operatorname{Id}_{B_{n+1}^{(1)}}-H_{B_{n+1}^{(1)}}(\beta_x)\right)^{-1}\tilde{c}\right ]\operatorname{det}\left(z\operatorname{Id}_{B_{n+1}^{(1)}}-H_{B_{n+1}^{(1)}}(\beta_x)\right).
			\end{align}
			Denote $B_{m+1}^o(x_1):=B_{m+1}(x_1)\setminus\{x_1\}$.  Since $x_1\in S_m(\beta_x,z)\cap B_{n+1}^o$ and $B_{n+1}^o$ is $n$-regular related to $(\beta_x,z)$ (c.f. Lemma  \ref{ag}), it follows that $B_{m+1}(x_1)\subset B_{n+1}^o$ and hence $B_{m+1}^o(x_1)\subset B_{n+1}^1 $.
			Since $x_1\in S_m(\beta_x,z)$ yields $|E_m(\beta_x+x_1\cdot\omega)-z|\leq \delta_m$, it follows that    $z\in A_m(\beta_x+x_1\cdot\omega,10\delta_m)$ (by $x_1\neq x$, we have $\beta_x+x_1\cdot\omega\notin\Z$).   It follows from Corollary \ref{cor} that $B_{m+1}^o(x_1)$ is $m$-good related to $(\beta_x,z)$ and $(\beta_x-0,z)$. Thus  by Hypothesis \ref{h3}, we have  
			\begin{equation}\label{bm,-}
				|G_{B_{m+1}^o(x_1)}^{\beta_x-0,z}(\tilde y,y)|\leq e^{-\gamma_{m}\|\tilde y-y\|_1}, \ \|\tilde y-y\|_1\geq l_{m}^{\frac{5}{6}},
			\end{equation}
			and 
			\begin{equation}\label{bm,}
				|G_{B_{m+1}^o(x_1)}^{\beta_x,z}(\tilde y,y)|\leq e^{-\gamma_{m}\|\tilde y-y\|_1}, \ \|\tilde y-y\|_1\geq l_{m}^{\frac{5}{6}}.
			\end{equation}
			By the resolvent identity, \eqref{1623,}, \eqref{1623,,}, \eqref{bm,-}, \eqref{bm,} and \eqref{sa1}, estimating as in \eqref{res}, we have 
			\begin{align}
				\nonumber	&z-v(\beta_x+x_1\cdot\omega)-\tilde{c}^{\rm T}\left(z\operatorname{Id}_{B_{n+1}^{(1)}}-H_{B_{n+1}^{(1)}}(\beta_x-0)\right)^{-1}\tilde{c}\\
				\nonumber	=&z-v(\beta_x+x_1\cdot\omega)-\tilde{c}^{\rm T}\left(z\operatorname{Id}_{B_{m+1}^o(x_1)}-H_{B_{m+1}^o(x_1)}(\beta_x-0)\right)^{-1}\tilde{c}\\-&\nonumber\varepsilon \sum_{\substack{\|\tilde{y}-x_1\|_1,\|y-x_1\|_1=1\\(w,w')\in \partial_{B_{n+1}^{(1)}}B_{m+1}^o(x_1)}}\tilde c(\tilde y)G_{B_{m+1}^o(x_1)}^{\beta_x-0,z}(\tilde y,w)G_{B_{n+1}^{(1)}}^{\beta_x-0,z}(w',y)\tilde c(y)\\
				\nonumber	=&z-v(\beta_x+x_1\cdot\omega)-c^{\rm T}\left(z\operatorname{Id}_{B_{m+1}^o}-H_{B_{m+1}^o}(\beta_x+x_1\cdot \omega -0)\right)^{-1}c-O(e^{-3l_{m+1}})\\
				=&\left(z-E_{m+1}(\beta_x+x_1\cdot \omega-0) -O(e^{-2l_{m+1}})\right)\left(1+O(\sum_{k=0}^m\delta_k)\right) \label{2036,-}
			\end{align}
			and 
			\begin{align}
				\nonumber	&z-v(\beta_x+x_1\cdot\omega)-\tilde{c}^{\rm T}\left(z\operatorname{Id}_{B_{n+1}^{(1)}}-H_{B_{n+1}^{(1)}}(\beta_x)\right)^{-1}\tilde{c}\\
				\nonumber	=&z-v(\beta_x+x_1\cdot\omega)-\tilde{c}^{\rm T}\left(z\operatorname{Id}_{B_{m+1}^o(x_1)}-H_{B_{m+1}^o(x_1)}(\beta_x)\right)^{-1}\tilde{c}\\-&\nonumber\varepsilon \sum_{\substack{\|\tilde{y}-x_1\|_1,\|y-x_1\|_1=1\\(w,w')\in \partial_{B_{n+1}^{(1)}}B_{m+1}^o(x_1)}}\tilde c(\tilde y)G_{B_{m+1}^o(x_1)}^{\beta_x,z}(\tilde y,w)G_{B_{n+1}^{(1)}}^{\beta_x,z}(w',y)\tilde c(y)\\
				\nonumber	=&z-v(\beta_x+x_1\cdot\omega)-c^{\rm T}\left(z\operatorname{Id}_{B_{m+1}^o}-H_{B_{m+1}^o}(\beta_x+x_1\cdot\omega)\right)^{-1} c-O(e^{-3l_{m+1}})\\
				=&\left(z-E_{m+1}(\beta_x+x_1\cdot \omega) -O(e^{-2l_{m+1}})\right)\left(1+O(\sum_{k=0}^m\delta_k)\right) \label{2036,}.
			\end{align}
			\begin{claim}\label{6}
				The product of \eqref{2036,-} and \eqref{2036,} is positive.
			\end{claim}
			Since $x_1\neq x$, then $\{\beta_x+x_1\cdot \omega \}\neq 0$. Hence by the inductive hypothesis \eqref{1640}, we have $E_{m+1}(\beta_x+x_1\cdot \omega)\geq E_{m+1}(\beta_x+x_1\cdot \omega-0)$. And since $x_1\notin S_{m+1}$ (by the maximum property of $m$), it follows that \begin{equation}\label{845}
				|E_{m+1}(\beta_x+x_1\cdot \omega)-z|\geq \delta_{m+1}\end{equation}
			and\begin{equation}\label{846}
				|E_{m+1}(\beta_x+x_1\cdot \omega-0)-z|\geq \delta_{m+1}.
			\end{equation}
			Since $x_1\neq o$, we have $\beta_x\neq \{\beta_x+x_1\cdot\omega\}$.\\  If $\beta_x<\{\beta_x+x_1\cdot\omega\}$, then by \eqref{Lc}, \eqref{h12} and  $z\leq E_n(\beta_x)+5\delta_n$, we have
			\begin{align}\label{>}
				\nonumber	&E_{m+1}(\beta_x+x_1\cdot \omega-0)-z\\ \nonumber\geq& E_{m+1}(\beta_x+x_1\cdot \omega-0)-E_{m+1}(\beta_x) -|E_{m+1}(\beta_x)-E_n(\beta_x)|-5\delta_n \\
				\geq &0-\sum_{k=m+1}^{n-1}e^{-l_k}-5\delta_n\nonumber\\
				>&-\delta_{m+1}/2.
			\end{align}
			Then by \eqref{>}, \eqref{846}, we have 
			$$E_{m+1}(\beta_x+x_1\cdot \omega)-z\geq E_{m+1}(\beta_x+x_1\cdot \omega-0)-z\geq \delta_{m+1} >e^{-2l_{m+1}},$$ yielding the first factors of \eqref{2036,-} and \eqref{2036,} are both negative.\\
			If $\beta_x>\{\beta_x+x_1\cdot\omega\}$, then by \eqref{Lc}, \eqref{h12} and  $z\geq E_n(\beta_x-0)- 5\delta_n$, we have
			\begin{align}\label{<}
				\nonumber	&E_{m+1}(\beta_x+x_1\cdot \omega)-z\\ \nonumber\leq& E_{m+1}(\beta_x+x_1\cdot \omega)-	E_{m+1}(\beta_x-0)+|E_{m+1}(\beta_x-0)-E_n(\beta_x-0)| +5\delta_n\\
				\leq &0+\sum_{k=m+1}^{n-1}e^{-l_k}+5\delta_n\nonumber\\
				<&\delta_{m+1}/2.
			\end{align}
			Then by \eqref{<}, \eqref{845}, we have 
			$$E_{m+1}(\beta_x+x_1\cdot \omega-0)-z\leq E_{m+1}(\beta_x+x_1\cdot \omega)-z\leq -\delta_{m+1} <-e^{-2l_{m+1}},$$ yielding the first factors of \eqref{2036,-} and \eqref{2036,} are both positive.  Hence the claim is proved. 
			
			By \eqref{351,}, \eqref{353,},   \eqref{2036,-}, \eqref{2036,} and Claim \ref{6},  we get the following proposition:
			\begin{prop}\label{953}Let  $m$ be the  largest integer of $0\leq m\leq n$ such that $B_{n+1}^o\cap S_m\neq \emptyset$ and let  $x_1\in S_m\cap B_{n+1}^o$ with $x_1\neq x$.  Then the sign of 
				\begin{equation}\label{956}
					\operatorname{det}\left(z\operatorname{Id}_{B_{n+1}^o}-H_{B_{n+1}^o}(\beta_x-0)\right)\operatorname{det}\left(z\operatorname{Id}_{B_{n+1}^o}-H_{B_{n+1}^o}(\beta_x)\right)
				\end{equation}
				is the same as the sign of  	$$	\operatorname{det}\left(z\operatorname{Id}_{B_{n+1}^{(1)}}-H_{B_{n+1}^{(1)}}(\beta_x-0)\right)\operatorname{det}\left(z\operatorname{Id}_{B_{n+1}^{(1)}}-H_{B_{n+1}^{(1)}}(\beta_x)\right)$$
				with $B_{n+1}^{(1)}:=B_{n+1}^o\setminus\{x_1\}$.
			\end{prop} 
		\end{itemize}
		Proposition \ref{953} means that deleting the largest order resonant point $x_1(\neq x)$ of $B_{n+1}^o$ does not change the sign of the  determinant product \eqref{956}. Moreover,  we   can iterate it into the following form:
		\begin{prop}\label{it}Let $\{B_{n+1}^{(j)}\}_{j=0}^J$ be a set sequence such that $B_{n+1}^{(0)}=B_{n+1}^o$ and $B_{n+1}^{(j)}=B_{n+1}^{(j-1)}\setminus\{x_j\}$ where $x_j(\neq x)$ is the largest order resonant point in $B_{n+1}^{(j-1)}$ (that is, $x_j\in S_{m_j}\cap B_{n+1}^{(j-1)}$ with $m_j$ being the largest integer such that $B_{n+1}^{(j-1)}\cap S_{m_j}\neq \emptyset$).
			Then the sign of 
			\begin{equation*}
				\operatorname{det}\left(z\operatorname{Id}_{B_{n+1}^{(j-1)}}-H_{B_{n+1}^{(j-1)}}(\beta_x-0)\right)\operatorname{det}\left(z\operatorname{Id}_{B_{n+1}^{(j-1)}}-H_{B_{n+1}^{(j-1)}}(\beta_x)\right)
			\end{equation*}
			is the same as the sign of  	$$	\operatorname{det}\left(z\operatorname{Id}_{B_{n+1}^{(j)}}-H_{B_{n+1}^{(j)}}(\beta_x-0)\right)\operatorname{det}\left(z\operatorname{Id}_{B_{n+1}^{(j)}}-H_{B_{n+1}^{(j)}}(\beta_x)\right).$$
		\end{prop} 
		The proof of Proposition \ref{it} is the same as in Case 2 (c.f. \eqref{case2}), with $B_{n+1}^o$ (resp. $B_{n+1}^{(1)}$, $x_1$) replaced by $B_{n+1}^{(j-1)}$ (resp. $B_{n+1}^{(j)}$, $x_j$). One issue is to verify that   $B_{n+1}^{(j-1)}$ is $(m_j+1)$-regular. This  follows from the $n$-regular property of $B_{n+1}^o$ and the sufficiently large separation of the   resonant points from the reason below.
		
		Let  $p_1\in S_k$ with $k\leq m_j$.  Then  for $p_2\in \{x_1,\cdots,x_{j-1}\}$, there exist   $\xi_1,\xi_2\in \{\beta_x,\beta_x-0\}$ such that  
		\begin{align*}
			&L\gamma \|p_1-p_2\|^{-\tau}\\ \leq& L\gamma \|(p_1-p_2)\cdot\omega\|_\T\\
			\leq& |E_k(\xi_1+p_1\cdot\omega)-E_k(\xi_2+p_2\cdot\omega)|\\
			\leq & |E_k(\xi_1+p_1\cdot\omega)-z|+|z-E_{m_j}(\xi_2+p_2\cdot\omega)|+|E_{m_j}(\xi_2+p_2\cdot\omega)-E_k(\xi_2+p_2\cdot\omega)|\\
			\leq&  \delta_k+2\delta_{m_j}+\sum_{s=k}^{m_j-1}e^{-l_s}\leq 10\delta_k,
		\end{align*}
		yielding   $\|p_1-p_2\|_1\geq 10l_{k+1}$ and hence  $p_2\notin B_{k+1}(p_1)$. Thus the whole blocks $B_{k+1}(p)$ ($p\in S_k\cap B_{n+1}^{(j-1)} $, $k\leq m_j$) are contained in $B_{n+1}^{(j-1)}$.
		
		By Proposition \ref{it}, we can find a set $B_{n+1}^{(J)}=B_{n+1}^o\setminus\{x_1,\cdots,x_J\}$ with $x_j\neq x$ such that 
		\begin{itemize}
			\item either  $B_{n+1}^{(J)}\cap S_m=\emptyset$ ($0\leq m\leq n$) or $x\in B_{n+1}^{(J)}$ is the largest order resonant point in $B_{n+1}^{(J)}$.
			\item The sign of 	$$\operatorname{det}\left(z\operatorname{Id}_{B_{n+1}^o}-H_{B_{n+1}^o}(\beta_x-0)\right)\operatorname{det}\left(z\operatorname{Id}_{B_{n+1}^o}-H_{B_{n+1}^o}(\beta_x)\right)$$ is the same as the sign of 
			$$	\operatorname{det}\left(z\operatorname{Id}_{B_{n+1}^{(J)}}-H_{B_{n+1}^{(J)}}(\beta_x-0)\right)\operatorname{det}\left(z\operatorname{Id}_{B_{n+1}^{(J)}}-H_{B_{n+1}^{(J)}}(\beta_x)\right).$$
		\end{itemize}
		Replacing $B_{n+1}^o$ by $B_{n+1}^{(J)}$ in  \eqref{1509} or Case 1  (c.f. \ref{case1})  yields 
		$$	\operatorname{det}\left(z\operatorname{Id}_{B_{n+1}^{(J)}}-H_{B_{n+1}^{(J)}}(\beta_x-0)\right)\operatorname{det}\left(z\operatorname{Id}_{B_{n+1}^{(J)}}-H_{B_{n+1}^{(J)}}(\beta_x)\right)<0$$ 
		and hence \eqref{key1}.
	\end{proof}
	
	\subsubsection{Green's function estimates for $(n+1)$-good sets}
	In this part, we relate the resonances to $E_{n+1}(\theta)$ and establish Green's function estimates for $(n+1)$-good sets. 
	For $p\in \Z^d$, we denote  $$B_{n+1}(p):=B_{n+1}+p.$$  
	\begin{defn}
		Fix $\theta\in \R$, $E\in \C$. Let $\delta_{n+1}:=e^{-l_{n+1}^\frac{2}{3}}$. Define the $(n+1)$-scale resonant points set (related to $(\theta,E)$) as  $$S_{n+1}(\theta,E):=\left\{p\in \Z^d:\ |E_{n+1} (\theta+p\cdot \omega)-E|<\delta_{n+1}\right\}.$$
		Let $\Lambda\subset\Z^d$ be a finite set. Related to $(\theta,E)$,  recall that 
		\begin{itemize}
			\item $\Lambda$ is $(n+1)$-nonresonant  if $\Lambda\cap S_{n+1}(\theta,E)=\emptyset$.
			\item  $\Lambda$ is $(n+1)$-regular if $p\in \Lambda\cap S_k(\theta,E)$ $\Rightarrow$  $B_{k+1}(p)\subset\Lambda$ ($0\leq k\leq n$).
			\item $\Lambda$ is $(n+1)$-good if it is both $(n+1)$-nonresonant and $(n+1)$-regular.
		\end{itemize}
		
	\end{defn}
	\begin{lem}\label{1727`}
		Let $m$ be an integer of $0\leq m\leq n$.	Fix $\theta\in \R$, $E\in \C$. If $p\in S_m(\theta,E)$ and $p \notin S_{m+1}(\theta,E)$, then for any $E^*\in\C$ such that $|E-E^*|<\delta_{m+1}/2$, we have 
		\begin{align}\label{1l2`}
			\|G_{B_{m+1}(p)} ^{\theta,E^*}\|&\leq10\delta_{m+1}^{-1}, 
			\\ 
			\label{1ex`}
			|G_{B_{m+1}(p)}^{\theta,E^*}(x,y)|&\leq e^{-\gamma_m(1-l_{m+1}^{-\frac{1}{50}})\|x-y\|_1}, \ \|x-y\|_1\geq l_{m+1}^{\frac{4}{5}}.
		\end{align}
	\end{lem}
	\begin{proof}
		Since $p\in S_m(\theta,E)$, by definition we have $|E_m(\theta+p\cdot \omega )-E|\leq \delta_m$. Thus $$ |E_m(\theta+p\cdot \omega )-E^*|\leq |E_m(\theta+p\cdot \omega )-E| +|E-E^*|\leq \delta_m+\delta_{m+1}/2.$$ By item ($3$) of Hypothesis \ref{h1},  Proposition \ref{528} and translation property, $E_{m+1}(\theta+p\cdot \theta )$ is the unique eigenvalue of $H_{B_{m+1}(p)}(\theta)$ in $[E_m(\theta+p\cdot \omega )-10\delta_m,E_m(\theta+p\cdot \omega )+10\delta_m]$ and moreover 
		\begin{align*}\
			|E_{m+1}(\theta+p\cdot \omega )-E^*|&\leq  |E_m(\theta+p\cdot \omega )-E^*|+|E_{m+1}(\theta+p\cdot \omega )-E_m(\theta+p\cdot \omega )|\\&\leq \delta_m+\delta_{m+1}/2+e^{-l_m}\\&<2\delta_m.
		\end{align*}
		Hence we get 
		$$\operatorname{dist}\left (\operatorname{Spec}(H_{B_{m+1}(p)}(\theta )),E^* \right )=|E_{m+1}(\theta+p\cdot \omega )-E^*|.$$
		Since $p \notin S_{m+1}(\theta,E)$, we have $$|E_{m+1}(\theta+p\cdot \omega )-E^*|\geq |E_{m+1}(\theta+p\cdot \omega )-E|-|E-E^*|\geq \delta_{m+1}-\delta_{m+1}/2\geq \delta_{m+1}/2.$$
		Since $H_{B_{m+1}(p)}(\theta )$ is self-adjoint,  we get  $$\|G_{B_{m+1}(p)} ^{\theta,E^*}\|=\frac{1}{\operatorname{dist}\left (\operatorname{Spec}(H_{B_{m+1}(p)}(\theta )),E^* \right )}=\frac{1}{|E_{m+1}(\theta+p\cdot \omega )-E^*|}\leq 2\delta_{m+1}^{-1}.$$ Thus we finish the proof of \eqref{1l2`}.
		Denote  $B_{m+1}^o(p):=B_{m+1}(p)\setminus\{p\}$. Since $|E_m(\theta+p\cdot \omega )-E|\leq \delta_m$ it follows that $E\in D(E_m(\theta+p\cdot\omega),\delta_m)$ and hence by Corollary \ref{cor},   $B_{m+1}^o(p)$ is $(m+1)$-good related to $(\theta,E)$. Hence by Hypothesis \ref{h3}, we have 
		\begin{align}\label{1200`}
			\|G_{B_{m+1}^o(p)} ^{\theta,E^*}\|&\leq10\delta_m^{-1},\nonumber\\
			|G_{B_{m+1}^o(p)}^{\theta,E^*}(x,y)|&\leq e^{-\gamma_m\|x-y\|_1}, \ \|x-y\|_1\geq l_m^\frac{5}{6}.
		\end{align}  Let $x,y\in B_{m+1}(p)$ such that $\|x-y\|_1\geq l_{m+1}^{\frac{4}{5}}$. By self-adjointness, we may assume $\|x-p\|_1\geq \frac{1}{2}l_{m+1}^\frac{4}{5}$.
		By the resolvent identity, one has 
		\begin{align}
			\nonumber		G_{B_{m+1}(p)}^{\theta,E^*}(x,y)&=\chi_{B_{m+1}^o(p)}(y)G_{B_{m+1}^o(p)}^{\theta,E^*}(x,y)\\
			-&\varepsilon\sum_{(w,w')\in \partial_{B_{m+1}(p)}B_{m+1}^o(p)} G_{B_{m+1}^o(p)}^{\theta,E^*}(x,w)G_{B_{m+1}(p)}^{\theta,E^*}(w',y),\label{1156`}
		\end{align}
			Since $w \in  \partial^-_{B_{m+1}(p)}B_{m+1}^o(p)$, $w' \in  \partial^+_{B_{m+1}(p)}B_{m+1}^o(p)=\{p\}$, one has 
			$$\|x-w\|_1\geq \|x-p\|_1-\|w-p\|_1\geq \|x-p\|_1-1>l_m^\frac{5}{6}.$$ Hence by  \eqref{1200`} and  \eqref{1156`}, we have 
			\begin{equation}\label{1205`}
				|G_{B_{m+1}(p)}^{\theta,E^*}(x,y)|\leq e^{-\gamma_m\|x-y\|_1}+e^{-\gamma_m(\|x-p\|_1-1)}|G_{B_{m+1}(p)}^{\theta,E^*}(p,y)|.
			\end{equation}
			\begin{itemize}
				\item 
				If $\|y-p\|_1\leq l_{m+1}^\frac{3}{4}$, then    \begin{align*}
					\|x-p\|_1-1&\geq \|x-y\|_1-\|y-p\|_1-1\\
					&\geq \left (1-(l_{m+1}^\frac{3}{4}+1) l_{m+1}^{-\frac{4}{5}}\right) \|x-y\|_1\\
					&\geq \left (1- 2l_{m+1}^{-\frac{1}{20}}\right) \|x-y\|_1.
				\end{align*}
				Bounding the term $|G_{B_{m+1}(p)}^{\theta,E^*}(p,y)|$ in \eqref{1205`} by $\|G_{B_{m+1}(p)}^{\theta,E^*}\|\leq 10\delta_{m+1}^{-1}=10e^{l_{m+1}^\frac{2}{3}}$, we get \begin{align*}
					\eqref{1205`}&\leq e^{-\gamma_m\|x-y\|_1}+10e^{-\gamma_m(1- 2l_{m+1}^{-\frac{1}{20}})\|x-y\|_1}e^{l_{m+1}^\frac{2}{3}}  \\
					&\leq e^{-\gamma_m(1- 2l_{m+1}^{-\frac{1}{20}}-(l_{m+1}^\frac{2}{3}+10) l_{m+1}^{-\frac{4}{5}})\|x-y\|_1}\\
					&\leq e^{-\gamma_m(1-l_{m+1}^{-\frac{1}{50}})\|x-y\|_1}.
				\end{align*}
				\item  If $\|y-p\|_1> l_{m+1}^\frac{3}{4}$, we can use  \eqref{1l2`}, \eqref{1200`} and the resolvent identity  to bound the term   $|G_{B_{m+1}(p)}^{\theta,E^*}(p,y)|$ in \eqref{1205`} by 
				\begin{align*}
					|G_{B_{m+1}(p)}^{\theta,E^*}(p,y)|&=|G_{B_{m+1}(p)}^{\theta,E^*}(y,p)|\\&\leq \varepsilon\sum_{(w,w')\in \partial_{B_{m+1}(p)}B_{m+1}^o(p)} |G_{B_{m+1}^o(p)}^{\theta,E^*}(y,w)G_{B_{m+1}(p)}^{\theta,E^*}(w',p)|\\
					&\leq e^{-\gamma_m(\|y-p\|_1-1)} \delta_{m+1}^{-1}.
				\end{align*}
				Thus 
				\begin{align*}
					\eqref{1205} &\leq e^{-\gamma_m\|x-y\|_1}+e^{-\gamma_m(\|x-p\|_1-1)}e^{-\gamma_m(\|y-p\|_1-1)} \delta_{m+1}^{-1}\\&\leq e^{-\gamma_m(1-(l_{m+1}^\frac{2}{3} +10) l_{m+1}^{-\frac{4}{5}})\|x-y\|_1}\\
					&\leq e^{-\gamma_m(1-l_{m+1}^{-\frac{1}{50}})\|x-y\|_1}.
				\end{align*}
			\end{itemize}
			Thus we finish the proof of \eqref{1ex}.
		\end{proof}
		
		\begin{prop}\label{1g`}
			Fix $\theta\in \R$, $E\in \C$. Assume that  a finite set $\Lambda$ is $(n+1)$-good related to $(\theta,E)$,  then for any $E^*\in \C$ such that $|E-E^*|<\delta_{n+1}/5$, we have 
			\begin{align}\label{1l`}
				\|G_{\Lambda} ^{\theta,E^*}\|&\leq10\delta_{n+1}^{-1},\\ 
				\label{1e`}
				|G_{\Lambda} ^{\theta,E^*}(x,y)|&\leq e^{-\gamma_{n+1}\|x-y\|_1}, \ \|x-y\|_1\geq l_{n+1}^{\frac{5}{6}},
			\end{align}
			where $\gamma_{n+1}=\gamma_n(1-l_{n+1}^{-\frac{1}{80}})$.
		\end{prop}
		\begin{proof}
			Define the set $S_n^\Lambda:=S_n(\theta,E)\cap\Lambda$. If $S_n^\Lambda=\emptyset$, then \eqref{1l`} and \eqref{1e`} follows from  Hypothesis \ref{h3} immediately. Thus we assume that $S_n^\Lambda\neq \emptyset$. 
			To prove \eqref{1l`}, by self-adjointness it suffices to prove 
			\begin{equation}\label{sd`}
				\operatorname{dist}\left (\operatorname{Spec}(H_{\Lambda}(\theta )),E \right )\geq  \delta_{n+1} /2.
			\end{equation}
			Fix $E'$ such that $|E'-E|<\delta_{n+1}/2$ and $E'\notin\operatorname{Spec}(H_{\Lambda}(\theta ))$.
			For any $p\in S_n^\Lambda$, since $\Lambda$ is $(n+1)$-regular, it follows that  $B_{n+1}(p)\subset\Lambda$. Since $\Lambda$ is $(n+1)$-nonresonant, it follows that  $p\notin S_{n+1}(\theta,E)$. Hence Lemma \ref{1727`} holds true for such $p$.
			Denote $\Lambda_0:=\Lambda\setminus S_n^\Lambda$. Then $\Lambda_0\cap S_n(\theta,E)=\emptyset$ and hence $\Lambda_0$ is $n$-nonresonant. Moreover, for $p\in S_n^\Lambda$ and  $q\in S_k(\theta,E)\cap \Lambda$ with  $0\leq k\leq n-1$, it follows from \eqref{Lc}, \eqref{DC}, \eqref{h12} that  
			\begin{align*}
				&L\gamma\|p-q\|_1^{-\tau}\leq	 L\|(p-q)\cdot\omega\|_\T\\
				\leq& |E_k(\theta+p\cdot\omega)-E_k(\theta+q\cdot\omega)|\\
				\leq  &  |E_k(\theta+p\cdot\omega)-E_n(\theta+p\cdot\omega)|+|E_n(\theta+p\cdot\omega)-E|+|E_k(\theta+q\cdot\omega)-E|\\
				\leq	& \delta_k+\delta_n+\delta_k\\
				\leq &3\delta_k.
			\end{align*}
			Hence $\|p-q\|_1\geq 100l_{k+1}$, yielding 	$p\notin B_{k+1}(q)$ and hence \begin{equation}\label{re}
				S_n^\Lambda\cap B_{k+1}(q)=\emptyset.
			\end{equation}
			Since $\Lambda$ is $(n+1)$-regular, it follows from  $q\in S_k(\theta,E)\cap \Lambda$ that 
			\begin{equation}\label{re2}
				B_{k+1}(q)\subset \Lambda.
			\end{equation}
			\eqref{re} together with  \eqref{re2} yields $B_{k+1}(q)\subset\Lambda_0$. It follows that $\Lambda_0$ is $n$-regular and hence $n$-good related to $(\theta,E)$. Thus by Hypothesis \ref{h3}, we have 	\begin{align}\label{ml`}
				\|G_{\Lambda_0} ^{\theta,E'}\|&\leq10\delta_n^{-1}, \\
				\label{me`}
				|G_{\Lambda_0} ^{\theta,E'}(x,y)|&\leq e^{-\gamma_n\|x-y\|_1}, \ \|x-y\|_1\geq l_n^{\frac{5}{6}}.
			\end{align}
			\begin{itemize}
				\item  Let $x\in \Lambda$ such that $\operatorname{dist}_1(x,S_n^\Lambda)\geq l_{n+1}^\frac{2}{3}$. By the resolvent identity, one has 
				$$G_\Lambda^{\theta,E'}(x,y)=\chi_{\Lambda_0}(y)G_{\Lambda_0}^{\theta,E'}(x,y)-\varepsilon\sum_{(w,w')\in \partial_{\Lambda}\Lambda_0} G_{\Lambda_0}^{\theta,E'}(x,w)G_{\Lambda}^{\theta,E'}(w',y).$$
				Thus \begin{align}\label{1835`}
					\sum_{y\in \Lambda}|G_\Lambda^{\theta,E'}(x,y)|&\leq \sum_{y\in \Lambda_0} |G_{\Lambda_0}^{\theta,E'}(x,y)|+\varepsilon\sum_{\substack{(w,w')\in \partial_{\Lambda}\Lambda_0\\y\in \Lambda}} |G_{\Lambda_0}^{\theta,E'}(x,w)||G_{\Lambda}^{\theta,E'}(w',y)|.
				\end{align}
				By \eqref{ml`} and \eqref{me`}, the first summation on the right-hand side of \eqref{1835`} satisfies
				\begin{align*}
					\sum_{y\in \Lambda_0} |G_{\Lambda_0}^{\theta,E'}(x,y)|&\leq (3l_n)^d\|G_{\Lambda_0}^{\theta,E'}\|+\sum_{y\in\Lambda_0:\|x-y\|_1\geq l_n }|G_{\Lambda_0}^{\theta,E'}(x,y)|\\
					&\leq 3^d10l_n^d\delta_n^{-1}+\sum_{y:\|x-y\|_1\geq l_n } e^{-\gamma_n\|x-y\|_1}\\
					&\leq \delta_n^{-2}.
				\end{align*}
				In the second summation on the right-hand side of \eqref{1835`}, since $w\in \partial_{\Lambda}^-\Lambda_0,$ we have 
				$$\|x-w\|_1\geq\operatorname{dist}_1(x,S_n^\Lambda)-\operatorname{dist}_1(w,S_n^\Lambda)\geq l_{n+1}^\frac{2}{3}-1>l_n^\frac{5}{6}.$$ It follows that  
				\begin{align*}
					&\varepsilon\sum_{\substack{(w,w')\in \partial_{\Lambda}\Lambda_0\\y\in \Lambda}} |G_{\Lambda_0}^{\theta,E'}(x,w)||G_{\Lambda}^{\theta,E'}(w',y)|\\ \leq&2d\varepsilon \sum_{w\in \partial_{\Lambda}^-\Lambda_0} e^{-\gamma_n\|x-w\|_1}	\sup_{w'\in \partial_{\Lambda}^+\Lambda_0}\sum_y|G_\Lambda^{\theta,E'}(w',y)|\\
					\leq & 2d\varepsilon \sum_{{w:\|x-w\|_1\geq l_{n+1}^\frac{2}{3}-1}} e^{-\gamma_n\|x-w\|_1}	\sup_{w'\in \partial_{\Lambda}^+\Lambda_0}\sum_y|G_\Lambda^{\theta,E'}(w',y)| 
					\\ \leq &\frac{1}{4}	\sup_{w'\in \Lambda}\sum_y|G_\Lambda^{\theta,E'}(w',y)|.
				\end{align*}
				Hence, \begin{equation}\label{1906`}
					\sup_{x: \operatorname{dist}_1(x,S_n^\Lambda)\geq l_{n+1}^\frac{2}{3}}\sum_{y\in \Lambda}|G_\Lambda^{\theta,E'}(x,y)|\leq \delta_n^{-2}+ \frac{1}{4}	\sup_{w'\in \Lambda}\sum_{y\in \Lambda}|G_\Lambda^{\theta,E'}(w',y)|.
				\end{equation}
				\item Let $x\in \Lambda$ such that $\operatorname{dist}_1(x,S_n^\Lambda)< l_{n+1}^\frac{2}{3}$. Then $\|x-p\|_1<l_{n+1}^\frac23$ and  $x\in B_{n+1}(p)$ for some $p\in S_n^\Lambda$.
				By the  resolvent identity, one has 
				\begin{equation*}
					G_\Lambda^{\theta,E'}(x,y)=\chi_{B_{n+1}(p)}(y)G_{B_{n+1}(p)}^{\theta,E'}(x,y)-\varepsilon\sum_{(w,w')\in \partial_\Lambda B_{n+1}(p)} G_{B_{n+1}(p)}^{\theta,E'}(x,w)G_{\Lambda}^{\theta,E'}(w',y).
				\end{equation*}
				Thus \begin{align}
					\nonumber	\sum_{y\in \Lambda}|G_\Lambda^{\theta,E'}(x,y)|&\leq \sum_{y\in B_{n+1}(p)} |G_{B_{n+1}(p)}^{\theta,E'}(x,y)|\\+&\varepsilon\sum_{\substack{(w,w')\in \partial_\Lambda B_{n+1}(p)\\ y\in \Lambda}} |G_{B_{n+1}(p)}^{\theta,E'}(x,w)||G_{\Lambda}^{\theta,E'}(w',y)|.\label{1911`}
				\end{align}
				By Lemma \ref{1727`},  the first summation on the right-hand side of \eqref{1911`} satisfies
				\begin{align*}
					\sum_{y\in B_{n+1}(p)} |G_{B_{n+1}(p)}^{\theta,E'}(x,y)|&\leq (3l_{n+1})^d\|G_{B_{n+1}(p)}^{\theta,E'}\|\leq \delta_{n+1}^{-2}.
				\end{align*}
				In the second summation on the right-hand side of \eqref{1911`}, since $w\in \partial_\Lambda^-B_{n+1}(p)$, we have  $\|w-p\|_1\geq l_{n+1}$. It follows that  $$\|x-w\|_1\geq \|w-p\|_1- \|x-p\|_1\geq  l_{n+1}-l_{n+1}^\frac{2}{3}\geq l_{n+1}^\frac{4}{5}.$$ Since $|E'-E|<\delta_1/2$, by Lemma \ref{1727`}, one has 
				\begin{align*}
					&\varepsilon\sum_{\substack{(w,w')\in \partial_\Lambda B_{n+1}(p)\\ y\in \Lambda}} |G_{B_{n+1}(p)}^{\theta,E'}(x,w)||G_{\Lambda}^{\theta,E'}(w',y)|\\	\leq& 2d\varepsilon \sum_{w\in \partial_\Lambda^-B_{n+1}(p)} |G_{B_{n+1}(p)}^{\theta,E'}(x,w)|	\sup_{w'\in \partial_\Lambda^+B_{n+1}(p)}\sum_{y\in \Lambda}|G_\Lambda^{\theta,E'}(w',y)|\\ \leq &
					2d\varepsilon \sum_{w:\|x-w\|_1\geq l_{n+1}^\frac{4}{5}} e^{-\gamma_n(1-l_{n+1}^{-\frac{1}{50}})\|x-w\|_1}	\sup_{w'\in \Lambda}\sum_{y\in \Lambda}|G_\Lambda^{\theta,E'}(w',y)|\\ \leq& \frac{1}{4}	\sup_{w'\in \Lambda}\sum_{y\in \Lambda}|G_\Lambda^{\theta,E'}(w',y)|.
				\end{align*}
				Hence, \begin{equation}\label{1940`}
					\sup_{x: \operatorname{dist}_1(x,S_n^\Lambda)< l_{n+1}^\frac{2}{3}}\sum_{y\in \Lambda}|G_\Lambda^{\theta,E'}(x,y)|\leq \delta_{n+1}^{-2}+ \frac{1}{4}	\sup_{w'\in \Lambda}\sum_{y\in \Lambda}|G_\Lambda^{\theta,E'}(w',y)|.
				\end{equation}
			\end{itemize}
			By \eqref{1906`} and \eqref{1940`}, we get 
			$$	\sup_{x\in \Lambda}\sum_{y\in \Lambda}|G_\Lambda^{\theta,E'}(x,y)|\leq \delta_{n+1}^{-2}+ \frac{1}{4}	\sup_{w'\in \Lambda}\sum_{y\in \Lambda}|G_\Lambda^{\theta,E'}(w',y)|$$ and hence 
			$$	\sup_{x\in \Lambda}\sum_{y\in \Lambda}|G_\Lambda^{\theta,E'}(x,y)|\leq 2\delta_{n+1}^{-2}.$$
			By Schur's test and self-adjointness, we get 
			\begin{equation}\label{1947`}
				\| G_\Lambda^{\theta,E'}\|\leq	\sup_{x\in \Lambda}\sum_{y\in \Lambda}|G_\Lambda^{\theta,E'}(x,y)|\leq 2\delta_{n+1}^{-2} .
			\end{equation}
			Since the uniform bound \eqref{1947`} holds true for all $E'$ in  $|E'-E|<\delta_{n+1}/2$  except a finite subset of $\operatorname{Spec}(H_{\Lambda}(\theta ))$,   \eqref{1947`} must  hold true for all $E'$ in  $|E'-E|<\delta_{n+1}/2$ as the resolvent is unbounded near the spectrum. Thus we finish the proof of \eqref{sd`}. To prove \eqref{1e`}, we let $x,y\in \Lambda$ such that $\|x-y\|_1\geq l_{n+1}^{\frac{5}{6}}$. For any $q\in \Lambda$, we introduce a subset $U(q)\subset\Lambda$ with $q\in U(q)$ to iterate the resolvent identity.
			
			For $0\leq m\leq n-1$, we define the sets 
			$$S_m^\Lambda:=S_m(\theta,E)\cap\Lambda.$$
			Since by \eqref{h12}, for any $p\in \Z^d$, we have $$ |E_{m+1}(\theta+p\cdot\omega)-E|<\delta_{m+1}  \ \Rightarrow\  |E_m(\theta+p\cdot\omega)-E|<\delta_m , $$ 
			it follows that \begin{equation}\label{ce}
				S_{m+1}^\Lambda\subset S_{m}^\Lambda. 
			\end{equation}
			And since $\Lambda$ is $(n+1)$ regular related to $(\theta,E)$, it follows that for any $p\in S_m^\Lambda$, the block $B_{m+1}(p)\subset \Lambda$ ($0\leq m\leq n$).
			\begin{itemize}
				\item If $\operatorname{dist}_1(q,S_0^\Lambda)\geq l_1^\frac{2}{3}$, we define 
				\begin{equation}\label{sf1`}
					U(q)=\{u\in \Lambda:\ \|q-u\|_1\leq l_1^\frac{1}{2}\}.
				\end{equation}
				By the resolvent identity,  one has 
				$$G_\Lambda^{\theta,E^*}(q,y)=\chi_{U(q)}(y)G_{U(q)}^{\theta,E^*}(q,y)-\varepsilon\sum_{(w,w')\in \partial_\Lambda U(q)} G_{U(q)}^{\theta,E^*}(q,w)G_{\Lambda}^{\theta,E^*}(w',y).$$
				Since $\operatorname{dist}_1(q,S_0^\Lambda)\geq l_1^\frac{2}{3}>l_1^\frac{1}{2}$, it follows that $U(q)\cap S_0(\theta,E)=\emptyset.$ By Proposition \ref{0g} and $\|q-w\|_1\geq l_1^\frac{1}{2}$ for $w\in \partial_\Lambda^- U(q)$, we have 
				\begin{align}\label{fenzi`}
					\nonumber	|G_\Lambda^{\theta,E^*}(q,y)|&\leq\chi_{U(q)}(y)|G_{U(q)}^{\theta,E^*}(q,y)|+\varepsilon\sum_{(w,w')\in \partial_\Lambda U(q)} |G_{U(q)}^{\theta,E^*}(q,w)||G_{\Lambda}^{\theta,E^*}(w',y)|\\
					\nonumber	&\leq\chi_{U(q)}(y)|G_{U(q)}^{\theta,E^*}(q,y)|+2dl_1^d\varepsilon\sup_{(w,w')\in \partial_\Lambda U(q)} |G_{U(q)}^{\theta,E^*}(q,w)||G_{\Lambda}^{\theta,E^*}(w',y)|\\
					\nonumber	&\leq\chi_{U(q)}(y)|G_{U(q)}^{\theta,E^*}(q,y)|+l_1^d\sup_{(w,w')\in \partial_\Lambda U(q)} e^{-\gamma_0\|q-w\|_1}|G_{\Lambda}^{\theta,E^*}(w',y)|\\ 
					\nonumber		&\leq\chi_{U(q)}(y)|G_{U(q)}^{\theta,E^*}(q,y)|+\sup_{(w,w')\in \partial_\Lambda U(q)} e^{-\gamma_0(\|q-w\|_1-d\ln l_1)}|G_{\Lambda}^{\theta,E^*}(w',y)|\\ 
					&\leq\chi_{U(q)}(y)|G_{U(q)}^{\theta,E^*}(q,y)|+\sup_{w'\in \partial_\Lambda^+ U(q)} e^{-\gamma_0(1-l_1^{-\frac13})\|q-w'\|_1}|G_{\Lambda}^{\theta,E^*}(w',y)|.
				\end{align}
				For $y\in U(q)$, since by Proposition \ref{0g}, 
				$$|G_{U(q)}^{\theta,E^*}(q,y)|\leq  e^{-\gamma_0\|q-y\|_1},\ \ \|q-y\|_1\geq 1,$$ 
				$$|G_{U(q)}^{\theta,E^*}(q,y)|\leq\|G_{U(q)}^{\theta,E^*}\| \leq 10\delta_0^{-1} ,\ \ \|q-y\|_1< 1,$$
				we have 
				\begin{equation}\label{f1`}
					|G_{U(q)}^{\theta,E^*}(q,y)|\leq10\delta_0^{-1} e^{-\gamma_0(\|q-y\|_1-1)}.
				\end{equation}
				\item If $\operatorname{dist}_1(q,S_0^\Lambda)< l_1^\frac{2}{3}$, then $\|q-p\|_1\leq l_1^{\frac{2}{3}}$ and $q\in B_1(p)$ for some $p\in S_0^\Lambda$. Recalling \eqref{ce}, we let $m$ be the largest integer of  $0\leq m\leq n$ such that $p\in S_m^\Lambda$ and we  define 
				\begin{equation}\label{sf2`}
					U(q)=\{B_{m+1}(p):\ q\in B_1(p)\text{ with }p\in S_0^\Lambda\}.
				\end{equation}
				By the resolvent identity,  one has 
				$$G_\Lambda^{\theta,E^*}(q,y)=\chi_{U(q)}(y)G_{U(q)}^{\theta,E^*}(q,y)-\varepsilon\sum_{(w,w')\in \partial_\Lambda U(q)} G_{U(q)}^{\theta,E^*}(q,w)G_{\Lambda}^{\theta,E^*}(w',y).$$
				Since  $w\in \partial_\Lambda^- U(q)=\partial_\Lambda^- B_{m+1}(p)$ and $\|q-p\|_1\leq l_1^{\frac{2}{3}}$, it follows that 
				\begin{equation}\label{yy}
					\|q-w\|_1\geq \|w-p\|_1-\|p-q\|_1\geq l_{m+1}-l_1^{\frac{2}{3}}\geq l_{m+1}^{\frac45}. 
				\end{equation}
				By the maximal property of $m$, it follows that $p\notin S_{m+1}(\theta,E)$. Then by Lemma \ref{1727`} and \eqref{yy}, one can get a similar estimate as \eqref{fenzi`}:
				\begin{align}\label{fenzi2`}
					\nonumber|G_\Lambda^{\theta,E^*}(q,y)|&\leq\chi_{U(q)}(y)|G_{U(q)}^{\theta,E^*}(q,y)|+\varepsilon\sum_{(w,w')\in \partial_\Lambda U(q)} |G_{U(q)}^{\theta,E^*}(q,w)||G_{\Lambda}^{\theta,E^*}(w',y)|\\
					\nonumber	&\leq\chi_{U(q)}(y)|G_{U(q)}^{\theta,E^*}(q,y)|+2dl_{m+1}^d\varepsilon\sup_{(w,w')\in \partial_\Lambda U(q)} |G_{U(q)}^{\theta,E^*}(q,w)||G_{\Lambda}^{\theta,E^*}(w',y)|\\
					\nonumber	&\leq\chi_{U(q)}(y)|G_{U(q)}^{\theta,E^*}(q,y)|+l_{m+1}^d\sup_{(w,w')\in \partial_\Lambda U(q)} e^{-\gamma_m(1-l_{m+1}^{-\frac{1}{50}})\|q-w\|_1}|G_{\Lambda}^{\theta,E^*}(w',y)|\\ 
					&\leq\chi_{U(q)}(y)|G_{U(q)}^{\theta,E^*}(q,y)|+\sup_{w'\in \partial_\Lambda^+ U(q)} e^{-\gamma_m(1-l_{m+1}^{-\frac{1}{60}})\|q-w'\|_1}|G_{\Lambda}^{\theta,E^*}(w',y)|.
				\end{align}
				For $y\in U(q)$, since by Lemma \ref{1727`},
				$$|G_{U(q)}^{\theta,E^*}(q,y)|\leq  e^{-\gamma_m(1-l_{m+1}^{-\frac{1}{50}})\|q-y\|_1},\ \ \|q-y\|_1\geq l_{m+1}^\frac{4}{5},$$ 
				$$|G_{U(q)}^{\theta,E^*}(q,y)|\leq\|G_{U(q)}^{\theta,E^*}\| \leq 10\delta_{m+1}^{-1} ,\ \ \|q-y\|_1< l_{m+1}^\frac{4}{5},$$
				we have 
				\begin{equation}\label{f2`}
					|G_{U(q)}^{\theta,E^*}(q,y)|\leq10\delta_{m+1}^{-1} e^{-\gamma_m(1-l_{m+1}^{-\frac{1}{50}})(\|q-y\|_1-l_{m+1}^\frac{4}{5})}.
				\end{equation}
			\end{itemize}
			Let $q_0=x$. By the  relations of decay rate $$\gamma_m(1-l_{m+1}^{-\frac{1}{60}})\geq \gamma_{m+1}\geq \gamma_{m+1}(1-l_{m+2}^{-\frac{1}{60}}), \ \ 0\leq m\leq n-1, $$   applying the estimates \eqref{fenzi`} and \eqref{fenzi2`} $K$ times yields 
			\begin{align}
				\nonumber	|G_\Lambda^{\theta,E^*}(x,y)|&\leq \sum_{k=0}^{K-1}\chi_{U(q_k)}(y)|G_{U(q_k)}^{\theta,E^*}(q_k,y)|\\+&
				\prod_{k=0}^{K-1}e^{-\gamma_n(1-l_{n+1}^{-\frac{1}{60}})\|q_k-q_{k+1}\|_1}|G_{\Lambda}^{\theta,E^*}(q_K,y)|,\label{1217`}
			\end{align}
			where $q_{k+1} \in \partial_\Lambda^+U(q_k)$ takes the supremum in the last line of \eqref{fenzi`} or \eqref{fenzi2`}. 
			We choose $K$ such that $y\notin U(q_k)$ for $k\leq K-1$ and $y\in U(q_K)$, thus the first summation in \eqref{1217`} vanishes. Hence, 
			\begin{align}\label{2221`}
				|G_\Lambda^{\theta,E^*}(x,y)|\leq
				\prod_{k=0}^{K-1}e^{-\gamma_n(1-l_{n+1}^{-\frac{1}{60}})\|q_k-q_{k+1}\|_1}|G_{\Lambda}^{\theta,E^*}(q_K,y)|.
			\end{align}
			By \eqref{fenzi`} and    \eqref{fenzi2`}, we have 
			\begin{align}\label{2204`}
				|G_{\Lambda}^{\theta,E^*}(q_K,y)|\leq |G_{U(q_K)}^{\theta,E^*}(q_K,y)|+\sup_{w'\in \partial_\Lambda^+ U(q_K)} e^{-\gamma_n(1-l_{n+1}^{-\frac{1}{60}})\|q_K-w'\|_1}|G_{\Lambda}^{\theta,E^*}(w',y)|.
			\end{align}
			Since  for $0\leq m\leq n-1$,
			\begin{align*}
				e^{-\gamma_m(1-l_{m+1}^{-\frac{1}{50}})\|q_K-y\|_1}&\leq e^{-\gamma_{m+1}\|q_K-y\|_1}\leq  e^{-\gamma_{m+1}(1-l_{m+2}^{-\frac{1}{50}})\|q_K-y\|_1},\\
				e^{\gamma_m(1-l_{m+1}^{-\frac{1}{50}})l_{m+1}^\frac{4}{5}}&\leq e^{2\gamma_{m+1}l_{m+1}^\frac{4}{5}}\leq e^{\frac{1}{2}\gamma_{m+1}l_{m+2}^\frac{4}{5}}\leq  e^{\gamma_{m+1}(1-l_{m+2}^{-\frac{1}{50}})l_{m+2}^\frac{4}{5}},
			\end{align*}
			it follows from  \eqref{f1`} and \eqref{f2`}  that the first term on the right-hand side of \eqref{2204`} satisfies
			\begin{equation}\label{25`}
				|G_{U(q_K)}^{\theta,E^*}(q_K,y)|\leq10\delta_{n+1}^{-1} e^{-\gamma_n(1-l_{n+1}^{-\frac{1}{50}})(\|q_K-y\|_1-l_{n+1}^\frac{4}{5})}.
			\end{equation}
			For $y\in U(q_K)$, $ w'\in \partial_\Lambda^+ U(q_K)$, we have:
			\begin{itemize}
				\item If $U(q_K)$ is in the  form of \eqref{sf1`}, then $$\|q_K-w'\|_1\geq \|q_K-y\|_1.$$
				\item If $U(q_K)$ is in the  form of \eqref{sf2`} with $U(q_K)=B_{m+1}(p)$, then since $Q_{l_{m+1}}\subset B_{m+1}\subset Q_{l_{m+1}+50l_m}$, it follows that 
				\begin{align*}
					\|q_K-w'\|_1&\geq \|p-w'\|_1-\|p-q_K\|_1 \geq l_{m+1}-\|p-q_K\|_1\\&\geq (l_{m+1}+50l_m)-\|p-q_K\|_1-50l_m \\
					&\geq \|p-y\|_1-\|p-q_K\|_1-50l_m\\&\geq \|q_K-y\|_1-2\|p-q_K\|_1-50l_m\\
					&\geq \|q_K-y\|_1-2l_1^\frac{2}{3} -50l_n.
				\end{align*}
			\end{itemize}  
			It follows that the second term on the right-hand side of \eqref{2204`} satisfies
			\begin{align}\label{2223`}
				\nonumber	&\sup_{w'\in \partial_\Lambda^+ U(q_K)} e^{-\gamma_n(1-l_{n+1}^{-\frac{1}{60}})\|q_K-w'\|_1}|G_{\Lambda}^{\theta,E^*}(w',y)|\\
				\nonumber \leq & \sup_{w'\in \partial_\Lambda^+ U(q_K)} e^{-\gamma_n(1-l_{n+1}^{-\frac{1}{60}})(\|q_K-y\|_1-2l_1^\frac{2}{3} -50l_n)}|G_{\Lambda}^{\theta,E^*}(w',y)|
				\\ \nonumber \leq & e^{-\gamma_n(1-l_{n+1}^{-\frac{1}{60}})(\|q_K-y\|_1 -60l_n)}\|G_{\Lambda}^{\theta,E^*}\|\\
				\leq& 10\delta_{n+1}^{-1}e^{-\gamma_n(1-l_{n+1}^{-\frac{1}{60}})(\|q_K-y\|_1- 60l_{n+1}^\frac12)}.
			\end{align}
			By \eqref{2221`}, \eqref{2204`}, \eqref{25`}, \eqref{2223`} and $\delta_{n+1}^{-1}=e^{l_{n+1}^\frac{2}{3}}$, provided $\|x-y\|\geq l_{n+1}^\frac56 $, we have 
			\begin{align*}
				|G_\Lambda^{\theta,E^*}(x,y)|&\leq
				20\delta_{n+1}^{-1}	e^{-\gamma_n(1-l_{n+1}^{-\frac{1}{60}})(\sum\limits_{k=0}^{K-1}\|q_k-q_{k+1}\|_1+\|q_K-y\|_1-3l_{n+1}^\frac45)}\\
				&\leq e^{-\gamma_{n}(1-l_{n+1}^{-\frac{1}{60}})(\|x-y\|_1-3l_{n+1}^\frac{4}{5}-2l_{n+1}^\frac{2}{3})}\\
				&\leq e^{-\gamma_n(1-l_{n+1}^{-\frac{1}{60}})(1-(3l_{n+1}^\frac{4}{5}+2l_{n+1}^\frac{2}{3})l_{n+1}^{-\frac{5}{6}})\|x-y\|_1}\\
				&\leq e^{-\gamma_n(1-l_{n+1}^{-\frac{1}{80}})\|x-y\|_1}.
			\end{align*}
		\end{proof}
		\begin{rem}\label{lhg}
			Replacing $\theta$ by $\theta-0$, we can estimate the Green's function $G_{\Lambda} ^{\theta-0,E^*}$ of the left limit operator  $H_\Lambda(\theta-0)$      for  $(n+1)$-good set $\Lambda$ related to $(\theta-0,E)$ with  $|E-E^*|<\delta_{n+1}/5$, by analogous definitions
			$$ S_m(\theta-0,E):=\left\{p\in \Z^d:\ |E_m (\theta+p\cdot \omega-0)-E|<\delta_m\right\}, \ 0\leq m\leq n+1.$$
		\end{rem}
		\subsubsection{Proof of Theorem \ref{key2}}
		In this section, we  check  all the hypotheses  for the $(n+1)$-scale one by one. \begin{itemize}
			\item   The $1$-periodic, single-valued, real-valued  properties of $E_{n+1}(\theta)$ are from the construction (c.f. Remark \ref{454n}).  Item (2) and  (3)   follow from Proposition \ref{528}.  
			\eqref{1639} and \eqref{1640} follow from  Proposition \ref{qiang}.
			\item
			Lemma \ref{sln} and \ref{H-*} yield Hypothesis \ref{h2}.
			\item The good Green's function estimates for  $(n+1)$-good sets is verified in  Proposition  \ref{1g`}.
			\item \eqref{h11} and Hypothesis \ref{h4} are verified in  Proposition  \ref{m+1r}.
		\end{itemize}

		\section{The unbounded potential case}\label{ub}
		In this section, we will consider the case that $v$ belongs to the UBLM class.  Since most of  arguments and proofs can be mapped one to one to the bounded potential case, we only focus on the different one. As mentioned in the introduction, it is the crucial Lipschitz monotonicity property of 
		$E_n(\theta)$ (the $n$-generation Rellich function) that enables the induction. When $v$ is bounded,  one can not, in general, expect $E_n(\theta)$ is continuous since a  rank one perturbation at finite energies may affect all eigenvalues. Recall that in the bounded potential  setting, we proved (c.f. Proposition \ref{814} and \ref{qiang}) that $E_n(\theta)$ has only positive jump discontinuities. While in the unbounded potential setting,  Proposition \ref{ubp}  below shows that  $E_n(\theta)$  in fact has no jump discontinuity  and hence satisfies the Lipschitz monotonicity property.
		\begin{lem}\label{cone}
			Let $A(\theta)$  be a continuous self-adjoint matrix-valued function defined on $[\theta_0,\theta_1)$. Suppose $f:(\theta_0,\theta_1)\to \R$ is continuous, and $f(\theta_0+0)=\pm\infty$. 
			Let 
			$$B(\theta)=	\begin{pmatrix}
				f(\theta) &	c^{\rm T} \\
				c & A(\theta)
			\end{pmatrix},$$ where $c$ is an arbitrary real-valued vector.
			Then 
			$$\lim_{\theta\to\theta_0+}\operatorname{Spec}(B(\theta))=\operatorname{Spec}(A(\theta_0))\cup\{\pm\infty\},$$
			where the convergence is understood in the sense of sets with multiplicities.
		\end{lem}
		\begin{proof}
			We refer to \cite{Kac19} for the detailed proof (c.f. Lemma 4.2 of \cite{Kac19}).
		\end{proof}
		As before, we denote $\{-x\cdot\omega\}$ by $\beta_x$ and denote by  $0=\alpha_1<\alpha_2<\cdots<\alpha_{|B_{n}|}<\alpha_{|B_{n}|+1}=1$ the discontinuities of $H_{B_{n}}(\theta)$, where 
		$$\{\alpha_1,\alpha_2,\cdots,\alpha_{|B_{n}|}\}=\left\{\beta_x\right\}_{x\in B_{n}}.$$
		Also denote  the Rellich functions of $H_{B_{n}}(\theta)$ in non-decreasing order by $\lambda_i(\theta)$ ($1\leq i\leq |B_{n}|$), which are $1$-periodic and continuous on $(0,1)$ except on the set $\{\alpha_1,\alpha_2,\cdots,\alpha_{|B_{n}|}\}$ with the  local Lipschitz monotonicity property as  mentioned at the beginning of section \ref{n=1}.
		\begin{prop}\label{ubp}
			In the unbounded potential setting {\rm ($v(0+0)=-\infty,$ $ v(1-0)=+\infty$)}, $E_n(\theta)$, the $n$-generation Rellich function, satisfies{\rm :}
			\begin{itemize}	\item On each small interval $(\alpha_k,\alpha_{k+1})$, $E_{n}(\theta)$ coincides with exactly one branch of $\lambda_i(\theta)$, and hence is continuous and satisfies the local Lipschitz monotonicity property on $(\alpha_k,\alpha_{k+1})${\rm :}
				\begin{equation*}
					E_{n}(\theta_2)-E_{n}(\theta_1)\geq L(\theta_2-\theta_1), \ \ \alpha_k< \theta_1\leq \theta_2<\alpha_{k+1},
				\end{equation*}
				\item $E_1(\theta)$ has   no jump discontinuity at $\beta_x$ for   $x\in B_n\setminus\{o\}${\rm :}  
				$$E_n(\beta_x-0)=E_n(\beta_x+0).$$
			\end{itemize} 
			
		\end{prop}
		\begin{proof}
			We prove this proposition  by induction. When $n=0$,  nothing needs to  be proved since $v(\theta)$ is continuous on $(0,1)$ by our assumption. Suppose that Proposition \ref{ubp} holds true at scale $n-1$. In the following we prove that  it remains valid at scale $n$. The proof of the first item is the same  as that of the first item in Proposition \ref{qiang} so we omit the details. For the second item, 
			let $x\in B_n\setminus\{o\}$ and denote $B_n^x:=B_n\setminus \{x\}$.	 We write 
			$$H_{B_n}(\theta)=\begin{pmatrix}
				f(\theta) &  c^{\rm T}  \\
				c	& H_{B_n^x}(\theta)
			\end{pmatrix},$$ where 	$f(\theta )=v(\theta+x \cdot \omega)$. 
			Then $f(\beta_x-0)=+\infty$, $f(\beta_x+0)=-\infty$ and $H_{B_n^x}(\theta)$ is continuous at $\beta_x$. Since $\beta_x\neq 0$, it follows  that $E_n(\theta )$ remains bounded for $\theta$  near $\beta_x$.  Hence by Lemma \ref{cone}, 
			\begin{equation}\label{1611}
				E_n(\beta_x-0)\in \operatorname{Spec}( H_{B_n^x}(\beta_x)),\ \  E_n(\beta_x+0)\in \operatorname{Spec}( H_{B_n^x}(\beta_x)).
			\end{equation}
			By the  inductive hypothesis, we have $E_{n-1}(\beta_x-0)=E_{n-1}(\beta_x+0)$ and we denote this identical value by $E$.  An analogue of Proposition \ref{528} yields that for any $\theta\notin\{\{-x\cdot\omega\}\}_{x\in B_n}$, 
			\begin{equation}\label{1631}
				|E_n(\theta)-E_{n-1}(\theta)|\leq e^{-l_{n-1}}
			\end{equation}
			and  any other eigenvalues of  $H_{B_{n}}(\theta)$, $\hat E$ except $E_{n}(\theta)$  satisfy $|\hat E- E_{n-1}(\theta)|>10\delta_{n-1}$.
			Taking  limits ($\theta\to \beta_x\mp$) of \eqref{1631} yields 
			\begin{equation}\label{1651}
				|	E_n(\beta_x-0)-E|\leq e^{-l_{n-1}} ,\ \  |	E_n(\beta_x+0)-E|\leq e^{-l_{n-1}}.
			\end{equation}
			Moreover, when $\theta\to \beta_x+$, the limit    (denoted by $\tilde{E}$) of other Rellich curves of $H_{B_n}(\theta)$  except $E_n(\theta)$  satisfies 
			$$|\tilde{E}-E|\geq 10\delta_{n-1}>e^{-l_{n-1}} .$$
			By Lemma \ref{cone} and \eqref{1651},  it follows   that   $H_{B_n^x}(\beta_x)$ has a unique eigenvalue (with multiplicities) in $\overline{D(E,e^{-l_{n-1}})}$, which together with \eqref{1611} and \eqref{1651} yields 
			$$E_n(\beta_x-0)=E_n(\beta_x+0).$$
		\end{proof}

		\section{Localization}\label{localization}
		In this section, we give the proof of Theorem \ref{mainb} and \ref{mainub}. We let $\varepsilon_0$ be so small that Proposition \ref{1832} and Theorem \ref{key2} hold true and hence the inductive hypotheses hold for all $n\geq 1$. 
		\begin{lem}\label{fen}
			Fix $\theta\in \R$ and $E\in \C$. Let $m$ be an integer of $m\geq 0$. Then  we have  
			\begin{equation}\label{fe}
				\inf_{\substack{x,y\in S_m(\theta,E)\\x\neq y}}\|x-y\|_1\geq 100l_{m+1}.
			\end{equation}
		\end{lem}
		\begin{proof}
			Assume $x,y\in S_m(\theta,E)$ with $x\neq y$. By \eqref{DC}, \eqref{Lc}, we have 
			\begin{align*}
				L\gamma \|x-y\|_1^{-\tau}&\leq L\|(x-y)\cdot \omega\|_\T\\
				&\leq |E_m(\theta+x\cdot \omega )-E_m(\theta +y\cdot \omega )| \\&\leq |E_m(\theta+x\cdot \omega )-E|+|E_m(\theta +y\cdot \omega )-E|\\&\leq 2\delta_m.
			\end{align*}
			Hence $\|x-y\|_1\geq (L\gamma/2)^{\frac{1}{\tau}}\delta_m^{-\frac{1}{\tau}}\geq 100l_{m+1}$.
		\end{proof}
		
		\begin{lem}\label{kuo}
			Fix $\theta\in \R$ and $E\in \C$. Then for any  set $B\subset \Z^d$, there exists another set $\tilde B$ satisfying the following properties{\rm :}
			\begin{itemize}
				\item $B\subset \tilde{B} \subset \{x\in \Z^d:\ \operatorname{dist}_1 (x,B)\leq 30l_n\}.$
				\item $\tilde B$ is $n$-regular  related to  $(\theta,E)$.
			\end{itemize}
		\end{lem}
		\begin{proof}
			Fix $\theta\in \R$ and $E\in \C$. 
			We start with $B^{(0)}:=B$. By setting $L=10l_{n-k}$, $X=S_{n-k-1}(\theta,E)$ in Lemma \ref{1227} (the assumption of Lemma \ref{1227} is verified by \eqref{fe}),  we inductively	define  $B^{(k+1)}$ ($0\leq k\leq n-1$), such that
			\begin{equation}\label{1321`}
				B^{(k)}\subset B^{(k+1)}\subset \{x\in \Z^d: \ \operatorname{dist}_1(x,B^{(k)})\leq 20l_{n-k}\},
			\end{equation}
			\begin{equation}\label{1916`}
				\text{ 	 $(Q_{10l_{n-k}}+x)\subset B^{(k+1)}$ for $x\in S_{n-k-1}(\theta,E)$ with $(Q_{10l_{n-k}}+x)\cap B^{(k+1)}\neq \emptyset$.}
			\end{equation}
			We claim that $B^{(n)}$ is the desired set $\tilde{B}$. Since $$\sum_{k=m}^{n-1}20l_{n-k}\leq 30l_{n-m},$$ it follows from \eqref{1321`}  that 
			\begin{equation}\label{1316`}
				B^{(n)}\subset \{x\in \Z^d :\ \operatorname{dist}_1(x,B^{(m)})\leq 30l_{n-m}\}.
			\end{equation}
			Thus  $$B^{(n)}\subset \{x\in \Z^d :\ \operatorname{dist}_1(x,B^{(0)})\leq  30l_n\}.$$
			Let $x$ be such that 	$x\in B^{(n)}\cap S_k(\theta,E)$
			for some $0\leq k\leq n-1$.  By \eqref{1316`} with $m=n-k$, we have
			\begin{equation}\label{1858`}
				B^{(n)}\subset \{x\in \Z^d :\ \operatorname{dist}_1(x,B^{(n-k)})\leq 30l_{k}\}.
			\end{equation}
			By $x\in B^{(n)}$, $30l_k\leq l_{k+1}$ and  \eqref{1858`}, it follows that 
			$$(Q_{10l_{k+1}}+x)\cap B^{(n-k)}\neq \emptyset.$$
			By \eqref{1916`}, it follows that $(Q_{10l_{k+1}}+x)\subset  B^{(n-k)}.$
			Hence, by the inductive hypothesis  \eqref{h11} and $l_{k+1}+50l_k\leq 10l_{k+1}$, it follows that 
			$$B_{k+1}(x)\subset (Q_{10l_{k+1}}+x)\subset  B^{(n-k)}\subset B^{(n)}.$$
			Hence the set $\tilde{B}:=B^{(n)}$ is $n$-regular related to $(\theta,E)$.
		\end{proof}
		\subsection{Anderson localization}Fix $\theta\in \R$. To show that $H(\theta)$ satisfies Anderson localization, it suffices by Schnol's Lemma to show that every generalized eigenfunction $\psi$ of $H(\theta)$ that grows at most polynomial bound ($|\psi(x)|\leq (1+\|x\|_1)^d$) in fact decays exponentially.  In the  remainder of this section, we fix a generalized eigenvalue $E$ and its corresponding nonzero generalized eigenfunction $\psi$.
		\begin{lem}\label{zha}
			There exists some $N>0$ such that for any $n\geq N$, we have 
			\begin{equation*}\label{fk}
				S_n(\theta,E)\cap Q_{50_{l_n}}\neq \emptyset.
			\end{equation*}
		\end{lem}
		\begin{proof}
			If the lemma does not hold true, then there exist a subsequence $\{n_i\}_{i=1}^\infty$  such that  
			\begin{equation}\label{gan}
				S_{n_i}(\theta,E)\cap Q_{50_{l_{n_i}}}= \emptyset.
			\end{equation}
			By Lemma \ref{kuo}, there exists a set sequence $\{U_i\}_{i=1}^\infty$ satisfying \begin{itemize}
				\item  $Q_{20l_{n_i}}\subset  U_i \subset \{x\in \Z^d:\ \operatorname{dist}_1 (x,Q_{20l_{n_i}})\leq 30l_{n_i}\} =Q_{50l_{n_i}}.$
				\item $U_i$ is $n_i$-regular related to $(\theta, E)$.
			\end{itemize}
			Moreover, it follows from \eqref{gan} that  $U_i$ is $n_i$-nonresonant and hence $n_i$-good with good Green's function estimate (c.f. Hypothesis \ref{h3}). Fix $x\in \Z^d$ arbitrarily. Let $i$ be sufficiently large so that $\|x\|_1\leq l_{n_i}$.
			It follows from Poisson formula that 
			\begin{equation}\label{p}
				\psi(x)=\sum_{(w,w')\in \partial_{\Z^d}U_i}G_{U_i}^{\theta,E}(x,w)\psi(w').
			\end{equation}
			For $w\in\partial_{\Z^d}^-U_i $, we have  $$\|x-w\|_1\geq \|w\|_1-\|x\|_1\geq 20l_{n_i}-l_{n_i}>l_{n_i}^\frac56.$$
			Thus \eqref{p} yields 
			\begin{align}
				\nonumber|\psi(x)|&\leq \sum_{(w,w')\in \partial_{\Z^d}U_i}|G_{U_i}^{\theta,E}(x,w)||\psi(w')|\\
				\nonumber&\leq \sum_{(w,w')\in \partial_{\Z^d}U_i} e^{-\gamma_{n_i}\|x-w\|_1}(1+\|w'\|_1)^d\\
				&\leq 2d(200l_{n_i})^de^{-19l_{n_i}}(100l_{n_i})^d\label{00}.
			\end{align}
			Taking $i\to\infty$  in \eqref{00} yields $\psi(x)=0$. Since $x$ is arbitrary, it follows that $\psi\equiv0$, contradicting with the  nonzero assumption of $\psi$.
		\end{proof}
		
		\begin{lem}\label{N}
			There exists some $N>0$ such that for  $n\geq N$, there exists a set $A_n\subset\Z^d$ satisfying 
			\begin{itemize}
				\item 	$Q_{98l_{n+1}}\setminus Q_{80l_n} \subset A_n\subset Q_{99l_{n+1}}\setminus Q_{50l_n}.$
				\item  $A_n$ is $n$-good related to $(\theta,E)$.
			\end{itemize}
		\end{lem}
		\begin{proof}
			Let $N$ be the integer from Lemma \ref{zha}.	It follows from  Lemma \ref{zha} and Lemma \ref{fen} that \begin{equation}\label{kong}
				(Q_{99l_{n+1}}\setminus Q_{50l_n})\cap S_n(\theta,E)=\emptyset.
			\end{equation}
			Otherwise, the exist $x\in  S_n(\theta,E)\cap Q_{50l_n}$, $y\in S_n(\theta,E)\cap (Q_{99l_{n+1}}\setminus Q_{50l_n})$ with $x\neq y$ and $\|x-y\|_1\leq \|x\|_1+\|y\|_1\leq 99l_{n+1}+50l_n<100l_{n+1}$, contradicting with Lemma \ref{fen}.
			By Lemma \ref{kuo}, there exists a set $A_n$ satisfies: 
			\begin{itemize}
				\item  $Q_{98l_{n+1}}\setminus Q_{80l_n} \subset A_n \subset \{x\in \Z^d:\ \operatorname{dist}_1 (x,Q_{98l_{n+1}}\setminus Q_{80l_n})\leq 30l_n\}\subset Q_{99l_{n+1}}\setminus Q_{50l_n}.$
				\item $A_n$ is $n$-regular related to $(\theta,E)$.
			\end{itemize}
			Since \eqref{kong}, it follows that  $A_n$ is $n$-good related to $(\theta,E)$. 
		\end{proof}
		Let $N$ be the integer from Lemma \ref{N}. Let $x\in \Z^d$ with $\|x\|_1>85l_N$. Thus  $x\in  Q_{95l_{n+1}}\setminus Q_{85l_n}$ for some $n\geq N$.
		It follows from Poisson formula that 
		\begin{equation}\label{f}
			\psi(x)=\sum_{(w,w')\in \partial_{\Z^d}A_n}G_{A_n}^{\theta,E}(x,w)\psi(w').
		\end{equation}
		Since for $x\in  Q_{95l_{n+1}}\setminus Q_{85l_n}$, $w\in \partial_{\Z^d}^-A_n$, one has 
		\begin{align*}
			\|x-w\|_1\geq&\min(98l_{n+1}-\|x\|_1,\|x\|_1-80l_n-1)\\
			\geq & \min(\frac{98}{95}\|x\|_1-\|x\|_1,\|x\|_1-\frac{80l_n+1}{85l_n+1}\|x\|_1)\\
			\geq & \frac{3}{95}\|x\|_1>l_n^\frac{5}{6},
		\end{align*}
		it follows that \eqref{f} together with the good Green's function estimate for $n$-good set $A_n$ yields 
		\begin{align*}
			\nonumber|\psi(x)|&\leq \sum_{(w,w')\in \partial_{\Z^d}A_n}|G_{A_n}^{\theta,E}(x,w)||\psi(w')|\\
			\nonumber&\leq \sum_{(w,w')\in \partial_{\Z^d}A_n} e^{-\gamma_{n}\|x-w\|_1}(1+\|w'\|_1)^d\\
			&\leq 2d(1000l_{n+1})^de^{-\frac{3}{95}\|x\|_1}(100l_{n+1})^d\\
			&\leq 2d10^{5d} \|x\|_1^{4d}e^{-\frac{3}{95}\|x\|_1}\\
			&\leq Ce^{-\frac{1}{95}\|x\|_1}
		\end{align*}
		with $$C=\sup_{x\in \Z^d}\frac{ 2d10^{5d} \|x\|_1^{4d}}{e^{\frac{2}{95}\|x\|_1}}.$$
		\subsection{Exponential dynamical localization}
		In the previous section, we proved that for any $\theta\in \R$, $H(\theta)$ satisfies Anderson localization. Hence $H(\theta)$ has a complete eigenfunction basis.  Denote by $\{\psi_{\theta,s}\}_{s=0}^\infty$ an orthonormal basis of $\ell^2(\Z^d)$ consisting of eigenfunctions of $H(\theta)$ and denote by $\mu_{\theta,s}$ the eigenvalue of $\psi_{\theta,s}$. For each $\psi_{\theta,s}$, we let $x_{\theta,s}\in \Z^d$ be such that $$|\psi_{\theta,s}(x_{\theta,s})|=  \|\psi_{\theta,s}\|_{\ell^\infty(\Z^d)}.$$
		The following general lemma is from \cite{JK13} (c.f.  Theorem 2.1 of  \cite{JK13}). We include the proof for  the reader's convenience.
		\begin{lem}
			For any $t\in \R$, we have 
			\begin{align}\label{847}
				|\langle e^{itH(\theta)}\e_x, \e_y\rangle|\leq \sum_{p\in \Z^d}\left(\sum_{s:x_{\theta,s}=p}|\psi_{\theta,s}(x)|^2\sum_{s:x_{\theta,s}=p}|\psi_{\theta,s}(y)|^2\right) ^{\frac{1}{2}}.
			\end{align}
		\end{lem}
		\begin{proof}
			Since  $H(\theta)\psi_{\theta,s}=\mu_{\theta,s}\psi_{\theta,s}$,  it follows that 
			$$e^{itH(\theta)}\psi_{\theta,s}=e^{it\mu_{\theta,s}}\psi_{\theta,s}.$$
			Hence \begin{align*}
				e^{itH(\theta)}\e_x=e^{itH(\theta)}(\sum_s  \overline{\psi_{\theta,s}(x)}\psi_{\theta,s})=\sum_s e^{it\mu_{\theta,s}} \overline{\psi_{\theta,s}(x)}\psi_{\theta,s}.
			\end{align*}
			It follows that 
			\begin{align*}
				|\langle e^{itH(\theta)}\e_x, \e_y\rangle|&=	|\langle\sum_s e^{it\mu_{\theta,s}} \overline{\psi_{\theta,s}(x)}\psi_{\theta,s},\e_y\rangle|\\
				&=|\sum_s e^{it\mu_{\theta,s}}\overline{\psi_{\theta,s}(x)}\psi_{\theta,s}(y)|\\
				&\leq 	\sum_s|\psi_{\theta,s}(x)||\psi_{\theta,s}(y)|	\\
				& =  \sum_{p\in \Z^d}   \sum_{s:x_{\theta,s}=p} |\psi_{\theta,s}(x)||\psi_{\theta,s}(y)|      \\
				&\leq \sum_{p\in \Z^d}\left(\sum_{s:x_{\theta,s}=p}|\psi_{\theta,s}(x)|^2\sum_{s:x_{\theta,s}=p}|\psi_{\theta,s}(y)|^2\right) ^{\frac{1}{2}},
			\end{align*}
			where we apply Cauchy-Schwartz inequality to the inner sum  on the last line.
		\end{proof}
		\begin{lem}\label{feng}
			There exists some $C(\varepsilon)>0$ such that for any $\theta\in \R$ and $\psi_{\theta,s}$ with $x_{\theta,s}=o$, we have 
			\begin{equation}\label{619}
				|\psi_{\theta,s}(x)|\leq C(\varepsilon) |\psi_{\theta,s}(o)|e^{-\frac{1}{400}|\ln\varepsilon|\cdot\|x\|_1}.
			\end{equation}
			We emphasize that here the constant $C(\varepsilon)$ does not depend on the choice of $\theta$ and  $\psi_{\theta,s}$.
		\end{lem}
		\begin{proof}
			Fix  $\theta\in \R$ and $\psi_{\theta,s}$ with $x_{\theta,s}=o$.	We denote $$\varphi_{\theta,s}:=\frac{\psi_{\theta,s}}{|\psi_{\theta,s}(o)|}.$$
			Since $x_{\theta,s}=o$, it follows that $\|\varphi_{\theta,s}\|_{\ell^\infty}\leq 1$ with $|\varphi_{\theta,s}(o)|=1$. To prove \eqref{619}, it suffices to prove 
			$$	|\varphi_{\theta,s}(x)|\leq C(\varepsilon) e^{-\frac{1}{400}|\ln\varepsilon|\cdot\|x\|_1}.$$
			\begin{claim}\label{cl1}
				For any $n\geq 1$, we have 	\begin{equation}\label{c}
					S_n(\theta,\mu_{\theta,s})\cap Q_{50_{l_n}}\neq \emptyset.
				\end{equation}
			\end{claim}
			Assume \eqref{c} does not hold true  for some $n\geq 1$. By Lemma \ref{kuo}, there exists a set $U$ satisfying \begin{itemize}
				\item  $Q_{20l_{n}}\subset  U \subset \{x\in \Z^d:\ \operatorname{dist}_1 (x,Q_{20l_{n}})\leq 30l_{n}\} =Q_{50l_{n}}.$
				\item $U$ is $n$-regular related to $(\theta, \mu_{\theta,s})$. 
			\end{itemize}
			Then by the negation of \eqref{c}, we have $$U\cap S_n(\theta,\mu_{\theta,s})=\emptyset.$$ Hence $U$ is $n$-good related to $(\theta, \mu_{\theta,s})$. 
			It follows from Poisson formula that 
			\begin{equation}\label{p`}
				\varphi_{\theta,s}(o)=\sum_{(w,w')\in \partial_{\Z^d}U}G_{U}^{\theta,\mu_{\theta,s}}(o,w)\varphi_{\theta,s}(w').
			\end{equation}
			For $w\in\partial_{\Z^d}^-U $, we have  $$\|o-w\|_1=\|w\|_1\geq 20l_{n}>l_{n}^\frac56.$$
			Since $\|\varphi_{\theta,s}\|_{\ell^\infty}\leq 1$, \eqref{p`} yields 
			\begin{align*}
				\nonumber|\varphi_{\theta,s}(o)|&\leq \sum_{(w,w')\in \partial_{\Z^d}U}|G_{U}^{\theta,\mu_{\theta,s}}(o,w)||\varphi_{\theta,s}(w')|\\
				\nonumber&\leq \sum_{(w,w')\in \partial_{\Z^d}U} e^{-\gamma_{n}\|o-w\|_1}\\
				&\leq 2d(200l_{n})^de^{-19l_n}\\
				&<1.
			\end{align*}
			The above inequality leads a contradiction to $|\varphi_{\theta,s}(o)|=1$. Hence we finish the proof of Claim \ref{cl1}.
			
			It follows from  Claim \ref{cl1} and Lemma \ref{fen} that for any $n\geq 1$,  \begin{equation}\label{kong`}
				(Q_{99l_{n+1}}\setminus Q_{50l_n})\cap S_n(\theta,\mu_{\theta,s})=\emptyset.
			\end{equation}
			Otherwise, the exist $x\in  S_n(\theta,\mu_{\theta,s})\cap Q_{50l_n}$, $y\in S_n(\theta,\mu_{\theta,s})\cap (Q_{99l_{n+1}}\setminus Q_{50l_n})$ with $x\neq y$ and $\|x-y\|_1\leq \|x\|_1+\|y\|_1\leq 99l_{n+1}+50l_n<100l_{n+1}$, contradicting with Lemma \ref{fen}.
			By Lemma \ref{kuo}, there exists a set $A_n$ satisfies: 
			\begin{itemize}
				\item  $Q_{98l_{n+1}}\setminus Q_{80l_n} \subset A_n \subset \{x\in \Z^d:\ \operatorname{dist}_1 (x,Q_{98l_{n+1}}\setminus Q_{80l_n})\leq 30l_n\}\subset Q_{99l_{n+1}}\setminus Q_{50l_n}.$
				\item $A_n$ is $n$-regular related to $(\theta,\mu_{\theta,s})$.
			\end{itemize}
			Since \eqref{kong`}, it follows that  $A_n$ is $n$-good related to $(\theta,\mu_{\theta,s})$. 
			Let $x\in \Z^d$ with $\|x\|_1>85l_1$. Thus  $x\in  Q_{95l_{n+1}}\setminus Q_{85l_n}$ for some $n\geq 1$.
			It follows from Poisson formula that 
			\begin{equation}\label{f`}
				\varphi_{\theta,s}(x)=\sum_{(w,w')\in \partial_{\Z^d}A_n}G_{A_n}^{\theta,\mu_{\theta,s}}(x,w)\varphi_{\theta,s}(w').
			\end{equation}
			Since for $x\in  Q_{95l_{n+1}}\setminus Q_{85l_n}$, $w\in \partial_{\Z^d}^-A_n$, one has 
			\begin{align*}
				\|x-w\|_1\geq&\min(98l_{n+1}-\|x\|_1,\|x\|_1-80l_n-1)\\
				\geq & \min(\frac{98}{95}\|x\|_1-\|x\|_1,\|x\|_1-\frac{80l_n+1}{85l_n+1}\|x\|_1)\\
				\geq & \frac{3}{95}\|x\|_1>l_n^\frac{5}{6},
			\end{align*}
			it follows that \eqref{f`},   together with  $\gamma_n\geq \frac{1}{2}\gamma_0=\frac{1}{4} |\ln\varepsilon|\geq 10$ (c.f. \eqref{rate}), $\|\varphi_{\theta,s}\|_{\ell^\infty}\leq 1$ and the good Green's function estimate for $n$-good set $A_n$, yields 
			\begin{align*}
				\nonumber|\varphi_{\theta,s}(x)|&\leq \sum_{(w,w')\in \partial_{\Z^d}A_n}|G_{A_n}^{\theta,\mu_{\theta,s}}(x,w)||\varphi_{\theta,s}(w')|\\
				\nonumber&\leq \sum_{(w,w')\in \partial_{\Z^d}A_n} e^{-\gamma_{n}\|x-w\|_1}\\
				&\leq 2d(1000l_{n+1})^de^{-\frac{3}{95}\gamma_n\|x\|_1}\\
				&\leq 2d10^{3d} \|x\|_1^{2d}e^{-\frac{1}{200}|\ln\varepsilon|\cdot\|x\|_1}\\
				&\leq C_0e^{-\frac{1}{400}|\ln\varepsilon|\cdot\|x\|_1}
			\end{align*}
			with  $$C_0=\sup_{x\in\Z^d}\frac{2d10^{3d} \|x\|_1^{2d}}{e^{\frac{1}{400}|\ln\varepsilon|\cdot\|x\|_1}}  \leq \sup_{x\in\Z^d}\frac{2d10^{3d} \|x\|_1^{2d}}{e^{\frac{1}{10}\|x\|_1}}.$$
			Hence, 
			$$|\varphi_{\theta,s}(x)|\leq \begin{cases}
				1 &\|x\|_1\leq 85l_1,\\
				C_0e^{-\frac{1}{400}|\ln\varepsilon|\cdot\|x\|_1} & \|x\|_1>85l_1.
			\end{cases}$$
			Thus we obtain $$ |\varphi_{\theta,s}(x)|\leq C(\varepsilon)e^{-\frac{1}{400}|\ln\varepsilon|\cdot\|x\|_1}$$
			with $$C(\varepsilon)=\max(e^{\frac{85}{400}|\ln\varepsilon|l_1},C_0)\leq \max(e^{\frac{85}{400}|\ln\varepsilon|^5},C_0). $$
		\end{proof}
		\begin{cor}
			With the constant  $C(\varepsilon)$ from Lemma \ref{feng},  for any $\theta\in \R$ and $\psi_{\theta,s}$, we have 
			\begin{equation}\label{802}
				|\psi_{\theta,s}(x)|\leq C(\varepsilon) |\psi_{\theta,s}(x_{\theta,s})|e^{-\frac{1}{400}|\ln\varepsilon|\cdot\|x-x_{\theta,s}\|_1}.
			\end{equation}
		\end{cor}
		\begin{proof}
			We define $\tilde\psi\in \ell^2(\Z^d)$ by 
			$$\tilde\psi(x):=\psi_{\theta,s}(x+x_{\theta,s}).$$
			Write the equation $H(\theta)\psi_{\theta,s}=\mu_{\theta,s}\psi_{\theta,s}$ in the form of 	 $$v(\theta+x\cdot\omega)\psi_{\theta,s}(x)+\sum_{y:\|y-x\|_1=1}\psi_{\theta,s}(y)=\mu_{\theta,s}\psi_{\theta,s}(x).$$
			It follows that $$v(\theta+x\cdot\omega+x_{\theta,s}\cdot\omega)\tilde\psi(x)+\sum_{y:\|y-x\|_1=1}\tilde\psi(y)=\mu_{\theta,s}\tilde\psi(x).$$ 
			Hence $\tilde\psi$ with $|\tilde\psi(o)|=\|\tilde\psi\|_{\ell^\infty(\Z^d)}$ is an $\ell^2(\Z^d)$ eigenfunction of $H(\tilde\theta)$ with $\tilde\theta=\theta+x_{\theta,s}\cdot\omega$.
			By Lemma \ref{619}, we have \begin{equation*}\label{61}
				|\tilde\psi(x)|\leq C(\varepsilon) |\tilde\psi(o)|e^{-\frac{1}{400}|\ln\varepsilon|\cdot\|x\|_1}.
			\end{equation*}
			Hence $$	|\psi_{\theta,s}(x+x_{\theta,s})|\leq C(\varepsilon) |\psi_{\theta,s}(x_{\theta,s})|e^{-\frac{1}{400}|\ln\varepsilon|\cdot\|x\|_1}.$$ Rewriting the above inequality yields \eqref{802}.
		\end{proof}
		It follows from \eqref{802} and  Bessel inequality that 
		\begin{align}\label{x}
			\nonumber	\sum_{s:x_{\theta,s}=p}|\psi_{\theta,s}(x)|^2&\leq C(\varepsilon)^2 \sum_{s:x_{\theta,s}=p}|\psi_{\theta,s}(p)|^2e^{-\frac{1}{200}|\ln\varepsilon|\cdot\|x-p\|_1}\\
			& \leq C(\varepsilon)^2  e^{-\frac{1}{200}|\ln\varepsilon|\cdot\|x-p\|_1}.
		\end{align}
		Similarly, 
		\begin{align}\label{y}
			\nonumber	\sum_{s:x_{\theta,s}=p}|\psi_{\theta,s}(y)|^2&\leq C(\varepsilon)^2 \sum_{s:x_{\theta,s}=p}|\psi_{\theta,s}(p)|^2e^{-\frac{1}{200}|\ln\varepsilon|\cdot\|y-p\|_1}\\
			& \leq C(\varepsilon)^2  e^{-\frac{1}{200}|\ln\varepsilon|\cdot\|y-p\|_1}.
		\end{align}
		It follows from \eqref{847}, \eqref{x}, \eqref{y} that 
		\begin{align}
			\nonumber	|\langle e^{itH(\theta)}\e_x, \e_y\rangle|&\leq  C(\varepsilon)^2 \sum_{p\in \Z^d}   e^{-\frac{1}{400}|\ln\varepsilon|\cdot\|x-p\|_1} e^{-\frac{1}{400}|\ln\varepsilon|\cdot\|y-p\|_1}\\
			\nonumber	&= C(\varepsilon)^2\sum_{p\in \Z^d}  e^{-\frac{1}{400}|\ln\varepsilon|\cdot\|p\|_1}e^{-\frac{1}{400}|\ln\varepsilon|\cdot\|y-x-p\|_1}\\
			&=C(\varepsilon)^2\left(\sum_{\|p\|_1\geq \|x-y\|_1}+\sum_{\|p\|_1<\|x-y\|_1}\right)\label{907}.
		\end{align}
		The  first summation in \eqref{907} is bounded by 
		\begin{align}
			\nonumber&e^{-\frac{1}{400}|\ln\varepsilon|\cdot\|x-y\|_1}	\sum_{\|p\|_1\geq \|x-y\|_1}e^{-\frac{1}{400}|\ln\varepsilon|\cdot\|y-x-p\|_1}\\
			\leq \nonumber& e^{-\frac{1}{400}|\ln\varepsilon|\cdot\|x-y\|_1}\sum_{p\in \Z^d}e^{-\frac{1}{400}|\ln\varepsilon|\cdot\|p\|_1}\\
			\leq &C_1 e^{-\frac{1}{400}|\ln\varepsilon|\cdot\|x-y\|_1}\label{2134}
		\end{align}
		with $$C_1= \sum_{p\in \Z^d}e^{-\frac{1}{400}|\ln\varepsilon|\cdot\|p\|_1} \leq \sum_{p\in \Z^d}e^{-\frac{1}{10}\|p\|_1}.$$
		The second summation  in \eqref{907} is bounded by 
		
		\begin{align}
			\nonumber	\sum_{\|p\|_1<\|x-y\|_1}e^{-\frac{1}{400}|\ln\varepsilon|\cdot\|x-y\|_1}&\leq  (2\|x-y\|_1)^d e^{-\frac{1}{400}|\ln\varepsilon|\cdot\|x-y\|_1}\\
			&\leq C_2e^{-\frac{1}{800}|\ln\varepsilon|\cdot\|x-y\|_1}\label{2135}
		\end{align}
		with $$C_2= \sup_{p\in \Z^d}\frac{(2\|p\|_1)^d}{e^{\frac{1}{800}|\ln\varepsilon|\cdot\|p\|_1}}\leq  \sup_{p\in \Z^d}\frac{(2\|p\|_1)^d}{e^{\frac{1}{20}\|p\|_1}}.$$
		It follows from \eqref{907}, \eqref{2134} and \eqref{2135} that 
		$$\sup_{t\in\R}|\langle e^{itH(\theta)}\e_x, \e_y\rangle|\leq C(\varepsilon)^2(C_1+C_2)e^{-\frac{1}{800}|\ln\varepsilon|\cdot\|x-y\|_1}.$$
		
		This completes the proof of exponential dynamical localization. 
		
		\normalem
		
		\section*{Acknowledgments}
		H. Cao  and Z. Zhang are partially supported by the National Key R\&D Program of China under Grant 2023YFA1008801.  H. Cao is supported by NSFC  (123B2004),  Y. Shi is  partially supported by NSFC  (12271380) and Z. Zhang is  partially supported by NSFC  (12288101).    
		

\begin{thebibliography}{CSZ24b}
			
			\bibitem[AM93]{AM93}
			M.~Aizenman and S.~Molchanov.
			\newblock Localization at large disorder and at extreme energies: an elementary
			derivation.
			\newblock {\em Comm. Math. Phys.}, 157(2):245--278, 1993.
			
			\bibitem[BLS83]{BLS83}
			J.~B{e}llissard, R.~Lima, and E.~Scoppola.
			\newblock Localization in {$v$}-dimensional incommensurate structures.
			\newblock {\em Comm. Math. Phys.}, 88(4):465--477, 1983.
			
			\bibitem[Bou00]{Bou00}
			J.~Bourgain.
			\newblock H\"{o}lder regularity of integrated density of states for the almost
			{M}athieu operator in a perturbative regime.
			\newblock {\em Lett. Math. Phys.}, 51(2):83--118, 2000.
			
			\bibitem[Cra83]{Cra83}
			W.~Craig.
			\newblock Pure point spectrum for discrete almost periodic {S}chr\"{o}dinger
			operators.
			\newblock {\em Comm. Math. Phys.}, 88(1):113--131, 1983.
			
			\bibitem[CSZ23]{CSZ23}
			H.~Cao, Y.~Shi, and Z.~Zhang.
			\newblock Localization and regularity of the integrated density of states for
			{S}chr\"{o}dinger operators on {$\Bbb Z^d$} with {$C^2$}-cosine like
			quasi-periodic potential.
			\newblock {\em Comm. Math. Phys.}, 404(1):495--561, 2023.
			
			\bibitem[CSZ24a]{CSZ24b}
			H.~Cao, Y.~Shi, and Z.~Zhang.
			\newblock On the spectrum of quasi-periodic {S}chr\"{o}dinger operators on
			$\mathbb{Z}^d$ with {${C}^2$}-cosine type potentials.
			\newblock {\em Comm. Math. Phys.}, 405(8):Paper No. 174, 84, 2024.
			
			\bibitem[CSZ24b]{CSZ24a}
			H.~Cao, Y.~Shi, and Z.~Zhang.
			\newblock Quantitative {G}reen's function estimates for lattice quasi-periodic
			{S}chr\"{o}dinger operators.
			\newblock {\em Sci. China Math.}, 67(5):1011--1058, 2024.
			
			\bibitem[FGP82]{FGP82}
			S.~Fishman, D.~R. Grempel, and R.~E. Prange.
			\newblock Chaos, quantum recurrences, and {A}nderson localization.
			\newblock {\em Phys. Rev. Lett.}, 49(8):509--512, 1982.
			
		
			\bibitem[FP84]{FP84}
			A.~L. Figotin and L.~A. Pastur.
			\newblock An exactly solvable model of a multidimensional incommensurate
			structure.
			\newblock {\em Comm. Math. Phys.}, 95(4):401--425, 1984.
			
			\bibitem[FS83]{FS83}
			J.~Fr\"{o}hlich and T.~Spencer.
			\newblock Absence of diffusion in the {A}nderson tight binding model for large
			disorder or low energy.
			\newblock {\em Comm. Math. Phys.}, 88(2):151--184, 1983.
			
			\bibitem[FSW90]{FSW90}
			J.~Fr\"{o}hlich, T.~Spencer, and P.~Wittwer.
			\newblock Localization for a class of one-dimensional quasi-periodic
			{S}chr\"{o}dinger operators.
			\newblock {\em Comm. Math. Phys.}, 132(1):5--25, 1990.
			
			\bibitem[FV21]{FV21}
			Y.~Forman and T.~VandenBoom.
			\newblock Localization and {C}antor spectrum for quasiperiodic discrete
			{S}chr\"odinger operators with asymmetric, smooth, cosine-like sampling
			functions.
			\newblock {\em arXiv:2107.05461, Mem. Am. Math. Soc. (to appear)}, 2021.
			
			\bibitem[GFP82]{GFP82}
			D.~Grempel, S.~Fishman, and R.~Prange.
			\newblock Localization in an incommensurate potential: {A}n exactly solvable
			model.
			\newblock {\em Phys. Rev. Lett.}, 49(11):833, 1982.
			
			\bibitem[HJY22]{HJY22}
			R.~Han, S.~Jitomirskaya, and F.~Yang.
			\newblock Anti-resonances and sharp analysis of {M}aryland localization for all
			parameters.
			\newblock {\em arXiv:2205.04021}, 2022.
			
			\bibitem[Jit94]{Jit94}
			S.~Jitomirskaya.
			\newblock Anderson localization for the almost {M}athieu equation: a
			nonperturbative proof.
			\newblock {\em Comm. Math. Phys.}, 165(1):49--57, 1994.
			
			\bibitem[Jit97]{Jit97}
			S.~Jitomirskaya.
			\newblock Continuous spectrum and uniform localization for ergodic
			{S}chr\"{o}dinger operators.
			\newblock {\em J. Funct. Anal.}, 145(2):312--322, 1997.
			
			\bibitem[Jit99]{Jit99}
			S.~Jitomirskaya.
			\newblock Metal-insulator transition for the almost {M}athieu operator.
			\newblock {\em Ann. of Math. (2)}, 150(3):1159--1175, 1999.
			
			\bibitem[JK13]{JK13}
			S.~Jitomirskaya and H.~Kr\"{u}ger.
			\newblock Exponential dynamical localization for the almost {M}athieu operator.
			\newblock {\em Comm. Math. Phys.}, 322(3):877--882, 2013.
			
			\bibitem[JK19]{JK19}
			S.~Jitomirskaya and I.~Kachkovskiy.
			\newblock All couplings localization for quasiperiodic operators with monotone
			potentials.
			\newblock {\em J. Eur. Math. Soc. (JEMS)}, 21(3):777--795, 2019.
			
			\bibitem[JK24]{JK24}
			S.~Jitomirskaya and I.~Kachkovskiy.
			\newblock Sharp arithmetic localization for quasiperiodic operators with monotone potentials.
			\newblock{\em 	arXiv:2407.00703}, 2024.
			
			
			
			\bibitem[JL17]{JL17}
			S.~Jitomirskaya and W.~Liu.
			\newblock Arithmetic spectral transitions for the {M}aryland model.
			\newblock {\em Comm. Pure Appl. Math.}, 70(6):1025--1051, 2017.
			
			\bibitem[JY21]{JY21}
			S.~Jitomirskaya and F.~Yang.
			\newblock Pure point spectrum for the {M}aryland model: a constructive proof.
			\newblock {\em Ergodic Theory Dynam. Systems}, 41(1):283--294, 2021.
			
			\bibitem[Kac19]{Kac19}
			I.~Kachkovskiy.
			\newblock Localization for quasiperiodic operators with unbounded monotone
			potentials.
			\newblock {\em J. Funct. Anal.}, 277(10):3467--3490, 2019.
			
			\bibitem[KPS22]{KPS22}
			I.~Kachkovskiy, L.~Parnovski, and R.~Shterenberg.
			\newblock Convergence of perturbation series for unbounded monotone
			quasiperiodic operators.
			\newblock {\em Adv. Math.}, 409(part B):Paper No. 108647, 54, 2022.
			
				\bibitem[KPS24]{KPS24}
			I.~Kachkovskiy, L.~Parnovski, and R.~Shterenberg.
			\newblock Perturbative diagonalization and spectral gaps of quasiperiodic operators on $\ell^2(\Z^d)$ with monotone potentials.
			\newblock {\em arXiv:2408.05650}, 2024.
			
			
			\bibitem[P{\"{o}}s83]{Pos83}
			J.~P{\"{o}}schel.
			\newblock Examples of discrete {S}chr\"{o}dinger operators with pure point
			spectrum.
			\newblock {\em Comm. Math. Phys.}, 88(4):447--463, 1983.
			
			\bibitem[Shi23]{Shi23}
			Y.~Shi.
			\newblock Localization for almost-periodic operators with power-law long-range
			hopping: a {N}ash-{M}oser iteration type reducibility approach.
			\newblock {\em Comm. Math. Phys.}, 402(2):1765--1806, 2023.
			
			\bibitem[Sim85]{Sim85}
			B.~Simon.
			\newblock Almost periodic {S}chr\"{o}dinger operators {I}{V}. {T}he maryland
			model.
			\newblock {\em Annals of Physics}, 159(1):157--183, 1985.
			
		\end{thebibliography}
		
		\section*{Data Availability}
		The manuscript has no associated data.
		\section*{Declarations}
		{\bf Conflicts of interest} \ The authors  state  that there is no conflict of interest.

	\end{document}